\documentclass[11pt]{article}
 
\usepackage[a4paper,left=17mm,top=25mm,right=17mm,bottom=25mm]{geometry}

\usepackage{authblk}
\usepackage[numbers, sort&compress]{natbib}
\usepackage[margin=30pt,font=small,labelfont=bf]{caption}
\usepackage{amsmath,amsopn,amsthm,amssymb}
\usepackage{thmtools,thm-restate}
\usepackage{fontawesome}

\usepackage{graphicx}
\usepackage{array} 
\usepackage{multirow} 
\usepackage{newtxtext}
\usepackage{newtxmath}
\usepackage[mathscr]{euscript}
\usepackage{mathtools}
\usepackage{float}
\usepackage{enumitem}
\setlist[enumerate,1]{label={(\arabic*)}} 
\usepackage{xspace}
\usepackage{scalerel}
\usepackage{algorithm} 
\usepackage{booktabs} 

\usepackage[
 	colorlinks=true,
	urlcolor=black,
	linkcolor=blue,
    citecolor=blue
]{hyperref}

\usepackage{booktabs}
\usepackage{multirow}
\usepackage{graphicx}
\usepackage{array}

\usepackage{algorithmicx}
\usepackage{algorithm} 
\usepackage[noend]{algpseudocode}
\algrenewcommand\algorithmicrequire{\textbf{Input:}}
\algrenewcommand\algorithmicensure{\textbf{Output:}}

\usepackage{tabularx}
\makeatletter
\newcommand{\multiline}[1]{%
  \begin{tabularx}{\dimexpr\linewidth-\ALG@thistlm}[t]{@{}X@{}}
    #1
  \end{tabularx}
}
\makeatother

\newtheorem{theorem}{Theorem}[section]

\newtheorem{lemma}[theorem]{Lemma} 
\newtheorem{corollary}[theorem]{Corollary}
\newtheorem{proposition}[theorem]{Proposition}
\newtheorem{definition}[theorem]{Definition}
\newtheorem{observation}[theorem]{Observation}

\newtheorem{remark}{Remark}
\newtheorem*{problem}{Problem}
\newcommand{\PROBLEM}[1]{\texttt{#1}}

\newcommand{\tc}{\ensuremath{\operatorname{tc}}}
\newcommand{\lca}{\ensuremath{\operatorname{lca}}}
\newcommand{\LCA}{\ensuremath{\operatorname{LCA}}}
\newcommand{\RF}{\texttt{RF}}
\newcommand{\srel}{\blacktriangleleft}
\newcommand{\rel}{\trianglelefteq}
\newcommand{\axiom}[1]{\textnormal{\textbf{(#1)}}}
\newcommand{\pairs}{\mathbb{P}_2}
\newcommand{\support}{\ensuremath{\operatorname{supp}}}
\newcommand{\atsupp}{\ensuremath{\operatorname{at-supp}}}
\newcommand{\tsupp}{\ensuremath{\operatorname{t-supp}}}
\newcommand{\Fcl}{\ensuremath{\operatorname{cl}_F}}
\newcommand{\cl}{\ensuremath{\operatorname{cl}}}
\newcommand{\FR}{F_{|R}}
\newcommand{\RR}{\mathcal{R}}
\newcommand{\FF}{\mathcal{F}}
\newcommand{\XR}{X_{\RR}}
\newcommand{\XRF}{X_{\RR,\FF}}
\newcommand{\FFRR}{\FF_{|\RR}}
\newcommand{\GG}{\mathscr{G}}
\newcommand{\Hasse}{\mathscr{H}}
\newcommand{\restr}{|^{\neq}_{xyz}}
\newcommand{\relrestr}{\mid^{\neq}_{xyz}}
\newcommand{\Rext}{R^{\textup{ext}}}
\newcommand{\anchor}{\text{\normalfont\tiny\faAnchor}}

\DeclareMathOperator{\indeg}{indeg}
\DeclareMathOperator{\outdeg}{outdeg}

\providecommand{\keywords}[1]{\textbf{\textit{Keywords: }} #1}
\providecommand{\MSC}[1]{\textbf{\textit{MSC: }} #1}

\title{Novel Triple-Based Problems for the Construction of Phylogenetic Networks via Least Common Ancestors}

\author[ ]{Patricia A.\ Ebert} 

\author[ ]{Anna Lindeberg}

\author[*]{Marc Hellmuth} 

\affil[ ]{Department of Mathematics, Faculty of Science,
  Stockholm University, SE-10691 Stockholm, Sweden} 

\affil[*]{corresponding author}

\date{\ }

\graphicspath{{./}{Final-Fig-Triples/}}

\begin{document}
\sloppy

\maketitle

\abstract{ 
Evolutionary histories are often represented by rooted phylogenetic networks, whose leaves
correspond to extant taxa and whose internal vertices represent ancestral lineages. Since such
histories must usually be inferred from incomplete data, in particular from genomic sequences of
present-day taxa, one often obtains only local information about relative evolutionary proximity.
For instance, sequence data may suggest that two taxa $x$ and $y$ are more closely related to each
other than either is to a third taxon $z$. This information is classically encoded by a rooted
triple $xy|z$.

In this paper, we study rooted triples in phylogenetic networks under an ancestor-based
interpretation: $xy|z$ is displayed if the unique least common ancestor (LCA) of $x$ and $y$ lies
strictly below the unique LCA of $x$ and $z$, respectively of $y$ and $z$, and the latter two LCAs
coincide. We also introduce anchored triples $\underline{x}y|z$, which retain only the asymmetric
comparison that the LCA of $x$ and $y$ lies below the LCA of $x$ and $z$. This relaxation is natural
in networks, where different pairwise ancestral relationships need not behave as they do in trees.

We consider several variants of consistency problems for ordinary and anchored triples, both with
and without forbidden triples. Somewhat surprisingly, these ancestor-based consistency questions for
triples in phylogenetic networks do not appear to have been addressed before despite their direct
biological interpretation and the fact that such constraints can be inferred naturally from genomic
sequence data. By translating these questions into realization problems for required and forbidden
LCA-constraints, we show that all resulting problems can be solved in polynomial time. Moreover,
whenever a solution exists, a suitable realizing DAG and phylogenetic network can be constructed
within the same time bound.
}

\smallskip
\noindent

\smallskip
\noindent
\MSC{92-08, 92D15, 68R10, 06A07}

\section{Introduction}

Rooted triples play a central role in phylogenetics as local constraints describing the relative
evolutionary proximity of three taxa. To be more precise, a rooted triple (also known as a rooted
triplet) $xy|z$ expresses that taxa $x$ and $y$ are more closely related to each other than either
is to $z$. For rooted phylogenetic trees, such local information is particularly powerful: the
classical BUILD algorithm of Aho et al.~\cite{Aho:81} decides in polynomial time whether there
exists a tree that displays a given set of rooted triples and, if so, constructs such a tree.

However, many evolutionary histories are not adequately described by trees. Processes such as
hybridization or horizontal gene transfer create reticulate patterns of ancestry, in which lineages
may merge as well as split. Phylogenetic networks, and more generally directed acyclic graphs
(DAGs), provide a natural framework for representing such non-tree-like evolution \cite{HRS:11}.
Consequently, it is important to understand when local triple information can be combined into a DAG
or network rather than into a tree.

For the discussion below, we require the following basic notation, further detailed in
Section~\ref{sec:prelim}. For two distinct vertices $v$ and $w$ of a DAG or network $G$, we write $w
\prec_G v$ if $w$ is a descendant of $v$ in $G$. A least common ancestor (LCA) of two taxa $x, y$ is
a vertex $v$ of $G$ that is an ancestor of both $x$ and $y$, and such that no proper descendant of
$v$ has this property. We write $\lca_G(xy)$ whenever the LCA of $x$ and $y$ is uniquely determined
in $G$ and say that $\lca_G(xy)$ is well-defined. 

A (rooted) triple $xy|z$ is a binary tree $t$ on three leaves $x,y$, and $z$ and with two internal
vertices such that
\[
\lca_t(xy) \prec_t \lca_t(xz) = \lca_t(yz).
\]
We now specify when a tree, or more generally a DAG, displays the information encoded by a triple.
There are several natural ways to interpret such information in a DAG setting. Here, we consider the
following two notions.

\begin{definition}
Let  $G$ be a DAG with leaf set $X$ and  $x,y,z\in X$. 
\begin{enumerate}
\item[\axiom{T1}]
A triple $xy|z$ is \emph{displayed} by $G$, if the least common ancestors $\lca_G(xy)$, $\lca_G(xz)$,
and $\lca_G(yz)$ are well-defined and
\[
\lca_G(xy) \prec_G \lca_G(xz) = \lca_G(yz).
\]

\item[\axiom{T2}] 
A triple $xy|z$ is \emph{topologically-displayed} by $G$, if there exist distinct vertices $u$ and $v$ 
in $G$ and pairwise internally vertex-disjoint directed paths
\[
v \leadsto u, \qquad
u \leadsto x, \qquad
u \leadsto y, \qquad
\text{ and } \qquad 
v \leadsto z .
\]
\end{enumerate}
\end{definition}

For rooted trees $T$, the two notions \axiom{T1} and \axiom{T2} are equivalent, i.e., $xy|z$ is
displayed by $T$ if and only if $xy|z$ is topologically-displayed by $T$. For DAGs or phylogenetic
networks, however, these two notions are, in general not equivalent. For example, consider the
networks $N_1$, $N_2$, and $N_3$ in Figure~\ref{fig:tiplesVSlca} that topologically-display several
triples but do not display them. In particular, in DAGs $G$ it holds that 
\[
\text{\axiom{T1}} \;\Longrightarrow\; \text{\axiom{T2}},
\]
by putting $u = \lca_G(xy)$ and $v = \lca_G(xz)$ for each triple $xy|z$ displayed by $G$; cf.\
Lemma~\ref{lem:T1=>T2}. However, the converse does not hold in general. Indeed, networks may satisfy
\axiom{T2} even when the relevant least common ancestors are not uniquely defined.

\begin{figure}[t]
    \centering
    \includegraphics[width=0.8\textwidth]{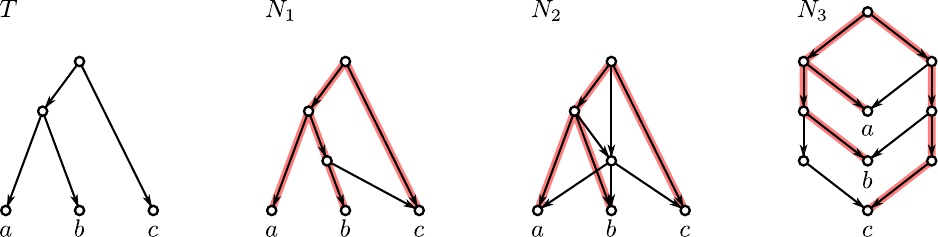}
    \caption{Shown are four networks $T$, $N_1$, $N_2$, and $N_3$ with leaf set $X=\{a,b,c\}$. The
    phylogenetic tree $T$ displays and topologically-displays the triple $ab|c$. All three networks
    $N_1$, $N_2$, and $N_3$ topologically-display the triple $ab|c$ according to \axiom{T2}
    (highlighted by the shaded red arcs) but do not display the triple $ab|c$ according to
    \axiom{T1}. In $N_1$, we have $\lca_{N_1}(bc)\prec_{N_1} \lca_{N_1}(ab)=\lca_{N_1}(ac)$. Hence,
    instead of $ab|c$, $N_1$ displays the triple $bc|a$ according to \axiom{T1}. In $N_2$, we have
    $\lca_{N_2}(ab)=\lca_{N_2}(ac)=\lca_{N_2}(bc)$. Hence, none of the triples $ab|c$, $ac|b$, and
    $bc|a$ is displayed in $N_2$ according to \axiom{T1}. The network $N_3$ topologically-displays
    all of the triples $ab|c$, $ac|b$, and $bc|a$. However, none of the LCAs $\lca_{N_3}(ab)$,
    $\lca_{N_3}(ac)$, and $\lca_{N_3}(bc)$ is well-defined in $N_3$, i.e., $N_3$ does not display
    any of the triples $ab|c$, $ac|b$, and $bc|a$ according to \axiom{T1}. This figure also appears
    as \cite[Fig~13]{LAMSH:25}.
    }
    \label{fig:tiplesVSlca}
\end{figure}

To the best of our knowledge, all algorithmic work on triples and phylogenetic networks is based on
the notion of \emph{topological-display} according to Condition~\axiom{T2}. In general, for every
set $\RR$ of rooted triples, there exists a network that topologically-displays all triples in
$\RR$. A simple example is the network $N_3$ in Figure~\ref{fig:tiplesVSlca}; see also
\cite{Huson:11,JANSSON200660} for further discussion. However, such networks are not very
informative, since they may topologically-display all three possible rooted triples for any given
set of three leaves. Consequently, it is natural to seek more specific and, therefore, more
informative networks topologically-displaying all triples in $\RR$. Unfortunately, once natural
structural restrictions are imposed on the network $N$ that topologically-displays a given set of
triples -- for instance requiring $N$ to be \emph{binary} (restricting the degrees of vertices) or
to have bounded \emph{level}-$k$ (measuring the degree of ``tree-likeness'' of $N$, with level-$0$
corresponding to trees) -- the underlying problems of verifying if there is such a network that
topologically-displays all triples in $\RR$ become NP-hard \cite{JanssonEtAl2006,vanIersel2009}. On
the other hand, if strong restrictions are placed on the set of triples $\RR$, polynomial time
algorithms become possible. In particular, if $\RR$ is \emph{dense}, that is, for every set of three
distinct leaves $x,y,z$ in some ground set $X$, at least one of the triples $xy|z$, $xz|y$, or $yz|x$
belongs to $\RR$, then binary level-$k$ networks with $k\in\{0,1,2\}$ topologically-displaying all
triples in $\RR$ can be constructed in polynomial time, or one can decide that no such network
exists \cite{Aho:81,vanIersel2011,lev2-09,HuberEtAl2011}. For less restricted, non-dense triple sets
$\RR$, a variety of heuristic approaches have been proposed to construct ``simple'' or
``informative'' networks that topologically-display as many triples in $\RR$ as possible
\cite{PoormohammadiEtAl2014,HuberEtAl2011}. Moreover, problems involving \emph{forbidden triples}
have also been studied. In particular, the problem of determining whether there exists a binary tree
that does not topologically-display any forbidden triple is NP-complete \cite{Bryant1997}.
Additionally, He et al. \cite{HHJS:06} showed that it is NP-hard to construct a binary phylogenetic
network that topologically-displays all required triples but none of the forbidden ones. More
generally, given a set of forbidden rooted binary trees, the problem of constructing a binary tree
having no subtree homeomorphic to one of the forbidden trees is NP-hard \cite{NgEtAl2000}. Inferring
phylogenetic networks from triples can also be viewed from the perspective of phylogenetic
supernetworks, where smaller local objects are combined into a single network. Besides triples,
trinets~\cite{vanIersel2022,Semple2021,vanIersel2014,HuberMoulton2013}, 
binets~\cite{HuberIerselMoultonScornavaccaWu2017},
quartets~\cite{GruenewaldEtAl2013}, and general
phylogenetic trees~\cite{Willson2011,HusonEtAl2004,HDKS:04} have been studied as such building
blocks.

In contrast, the LCA-based notion \axiom{T1} has received comparatively less attention so far,
although it is closer to situations in which triples are inferred from pairwise evolutionary
proximity, for instance from genomic sequence data. Such data typically consists of alignments of
homologous sequences, from which one obtains pairwise sequence dissimilarities $d(xy)$ between genes
or genomes $x$ and $y$. Small values of $d(xy)$ indicate a high sequence similarity and are
therefore commonly interpreted as evidence that $x$ and $y$ share a more recent evolutionary history
than more dissimilar pairs. Thus, if the observed dissimilarities satisfy, for example,
\[
d(xy) < d(xz),
\]
then this suggests that $x$ and $y$ are evolutionarily closer to each other than $x$ and $z$. In a
rooted tree-like history, this proximity is naturally expressed in terms of the relative position of
least common ancestors and $d(xy) < d(xz)$ can be interpreted as evidence for the LCA-constraint
\begin{equation}\label{eq:lca-prec}
\lca_N(xy) \prec_N \lca_N(xz)
\end{equation}
which states that $x$ and $y$ have a more recent common ancestor than $x$ and $z$. 
This interpretation directly matches the LCA-based notion \axiom{T1}: a triple $xy|z$ records that
the pair $x,y$ is more closely related than either pair involving $z$, as witnessed by the relative
depths of their least common ancestors. Hence, \axiom{T1} provides a natural way to translate
pairwise proximity information, such as sequence-based dissimilarities, into local constraints on a
DAG or phylogenetic network. This line of reasoning has already proved fruitful in the construction
of phylogenetic trees from best-match relations~\cite{SchallerEtAl2021,Geiss2019,Geiss2020}.

However, the additional equality condition $\lca_N(xz)=\lca_N(yz)$ required by \axiom{T1} is often
too restrictive for networks. This motivates the following relaxation: instead of requiring equality
of the two LCAs, we only require the strict inequality as in Equation~\eqref{eq:lca-prec}. In trees
this relaxation still implies that a triple $xy|z$ is displayed, since 
\[
\lca_T(xy) \prec_T \lca_T(xz)\quad\text{if and only if}\quad
\lca_T(xy) \prec_T \lca_T(yz), 
\]  
and, in particular, $\lca_T(xz) = \lca_T(yz)$. In general DAGs or networks, however, these
(in)equalities are no longer equivalent. For example, in the network $N_2$ in
Figure~\ref{fig:anchordisplay}, it holds that $\lca_{N_2}(xy) \prec_{N_2} \lca_{N_2}(xz)$ but
$\lca_{N_2}(xy) \nprec_{N_2} \lca_{N_2}(yz)$. To capture this asymmetry, we introduce \emph{anchored
triples} $\underline{x}y|z$ that indicate that $x$ is the distinguished leaf.

\begin{definition}
Let $G$ be a DAG with leaf set $X$ and $x,y,z \in X$. An anchored triple $\underline{x}y|z$ is
\emph{\anchor-displayed} by $G$ if the least common ancestors $\lca_G(xy)$ and $\lca_G(xz)$ are
well-defined and $\lca_G(xy) \prec_G \lca_G(xz)$.
\end{definition}

As we shall see later in Lemma~\ref{lem:display-vs-anchordisplay}, a triple $xy|z$ is displayed
according to \axiom{T1} in $G$ if and only if the anchored triples $\underline{x}y|z$ and
$\underline{y}x|z$ are \anchor-displayed in $G$. However, a DAG or network $G$ may $\anchor$-display
only one of the anchored triples $\underline{x}y|z$ and $\underline{y}x|z$. For example, the
networks $N_1$ and $N_2$ in Figure~\ref{fig:anchordisplay} both \anchor-display $\underline{x}y|z$
but not $\underline{y}x|z$.

\begin{figure}[t]
    \centering
    \includegraphics[width=0.8\textwidth]{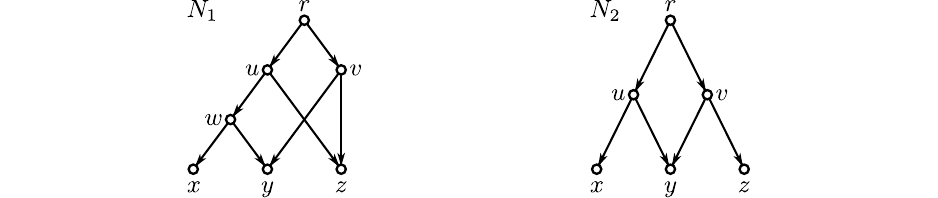}
    \caption{
    Shown are two networks $N_1$ and $N_2$ with leaf set $\{x,y,z\}$. Both \anchor-display
    $\underline{x}y|z$ since $\lca_{N_1}(xy)=w\prec_{N_1}u=\lca_{N_1}(xz)$ and
    $\lca_{N_2}(xy)=u\prec_{N_2}r=\lca_{N_2}(xz)$. In $N_1$, the leaves $y$ and $z$ have two
    distinct least common ancestors, namely $u$ and $v$. In $N_2$, $v=\lca_{N_2}(yz)$ is the unique
    least common ancestor of $y$ and $z$, but $v$ is not an ancestor of $u=\lca_{N_2}(xy)$.
    Therefore, neither $N_1$ nor $N_2$ \anchor-displays $\underline{y}x|z$. In particular, $xy|z$ is
    not displayed by $N_1$ and $N_2$. 
    }
    \label{fig:anchordisplay}
\end{figure}

In this contribution, we consider several fundamental algorithmic ``triple consistency
(\PROBLEM{TC})'' and ``anchored triple consistency (\PROBLEM{ATC})'' problems. To this end, let
$\RR$ be a set of required triples and, possibly, let $\FF$ be a set of forbidden triples, where
triples may either be rooted (``ordinary'') or anchored. Now, the task is to decide whether there
exists a phylogenetic network, respectively a DAG, that displays all required triples and displays
none of the forbidden ones. For ordinary triples, ``display'' is understood in the sense of
\axiom{T1}; for anchored triples, it is understood as $\anchor$-display. This gives the following
four basic decision problems:
\[
\begin{array}{l|l|l}
\text{Problem} & \text{Input} & 
\text{Question: Is there a phylogenetic network (resp., DAG) $G$  that \dots } \\ \hline 
\PROBLEM{TC} & \text{triple set } \RR &
\text{\dots\ displays all triples in }\RR \,?\\[0.2em]
\PROBLEM{ATC} & \text{anchored triple set } \RR &
\text{\dots\ \anchor-displays all anchored triples in }\RR \,? \\[0.2em]
\PROBLEM{TC-F} & \text{triple sets } \RR,\FF &
\text{\dots\ displays all triples in }\RR\text{ and none in }\FF \,?\\[0.2em]
\PROBLEM{ATC-F} & \text{anchored triple sets } \RR,\FF  &
\text{\dots\ \anchor-displays all anchored triples in }\RR\text{ and none in }\FF \,? \\[0.2em]
\end{array}
\]

While every set of rooted triples can be topologically-displayed by some
network~\cite{Huson:11,JANSSON200660}, this is, in general, not the case for the LCA-based notion of
display given by \axiom{T1}. Hence, already the basic consistency problems \PROBLEM{TC} and
\PROBLEM{ATC} are non-trivial.

In addition, we consider two natural strengthened versions of the forbidden-triple problems. In
\PROBLEM{TC-F}, a forbidden triple $xy|z\in\FF$ may fail to be displayed simply because one of the
LCAs $\lca_G(xy)$, $\lca_G(xz)$, or $\lca_G(yz)$ is not well-defined. In the \emph{strong} variant,
denoted \PROBLEM{strong-TC-F}, we additionally require these three LCAs to be well-defined for every
forbidden triple. Analogously, in \PROBLEM{strong-ATC-F}, we require the relevant LCAs to be
well-defined for every forbidden anchored triple. Finally, one may impose this LCA-well-definedness
requirement globally. We say that a DAG $G$ has the \emph{2-lca-property} if $\lca_G(xy)$ is
well-defined for all leaves $x$ and $y$ of $G$. The variants \PROBLEM{LCA-TC-F} and \PROBLEM{LCA-ATC-F} ask
for solutions to \PROBLEM{TC-F}, respectively, \PROBLEM{ATC-F} among phylogenetic networks or DAGs
with the 2-lca-property.

In the following, we show that all of these problems can be solved in polynomial time. To this end,
we build on recent work in which we established results for LCA-based constraints~\cite{EH:26arxiv,
LAMSH:25}. More precisely, let $R$ and $F$ be two binary relations containing pairs of the form
$(ab,cd)$ where $a,b,c,d$ are not necessarily distinct. A pair $(ab,cd) \in R$ specifies
\emph{required} pairwise LCA relationships. In particular, a phylogenetic network or, more
generally, a DAG $G$ realizes $R$ if for all $(ab,cd) \in R$ the least common ancestors $\lca_G(ab)$
and $\lca_G(cd)$ are well-defined and, depending on additional properties of $R$, either
$\lca_G(ab)\prec_G \lca_G(cd)$ or $\lca_G(ab)=\lca_G(cd)$ holds. On the other hand, the relation $F$
represents \emph{forbidden} pairwise LCA relationships: for $(ab,cd)\in F$, a DAG $G$ is required to
satisfy $\lca_G(ab)\not\prec_G\lca_G(cd)$ whenever both LCAs are well-defined. Several other ways of
defining forbidden LCA relationships exist, see \cite{EH:26arxiv}. A pair $(R,F)$ of such
relations is then realized by a DAG $G$ if $G$ realizes both $R$ and $F$. In \cite{LAMSH:25,
EH:26arxiv}, it was shown that it can be decided in polynomial time whether there exists a
phylogenetic network, respectively, DAG realizing $(R,F)$ and, if so, that such a DAG can be
constructed within the same time bound. These results form the foundation for solving the
consistency problems for rooted and anchored triples.

This paper is organized as follows. In Section~\ref{sec:prelim}, we provide the basic definitions
needed throughout this paper. In Section~\ref{sec:LCA-rel}, we start with providing some known
results about pairs $(R,F)$ of required and forbidden LCA-constraints and their realization by DAGs
or networks. In Section~\ref{sec:ATC-F}, respectively, \ref{sec:TC-F} we show that the problems
\PROBLEM{ATC} and \PROBLEM{ATC-F}, respectively,
\PROBLEM{TC} and \PROBLEM{TC-F} are solvable in polynomial time and provide phylogenetic DAGs
and networks that satisfy the corresponding requirements of the problems (if such graphs exist).
In Section~\ref{sec:strong_and_LCA_problem}, we then study further variants in which selected
LCAs, or even all pairwise LCAs, are additionally required to be well-defined. We show that these
strengthened consistency problems are also solvable in polynomial time.

%
%
\section{Preliminaries}
\label{sec:prelim}
\paragraph{Sets and Relations.}

In what follows, $X$ will always be a finite non-empty set. We denote with $\mathcal{P}(X)$ the
power set of $X$. Moreover, we let $\pairs(X)\coloneqq\{\{a,b\} \, : \, a,b\in X\}$ denote the set
system consisting of all 1- and 2-element subsets of $X$. We will often write $ab$, respectively,
$aa$ for elements $\{a,b\}$, respectively, $\{a\}$ in $\pairs(X)$. Thus, $ab=ba$ always holds. 

Given a set $A$, a  subset $R\subseteq A\times A$ is a \emph{binary relation (on $A$)}.

\begin{remark}
As all relations considered in this work are binary, we shall simply refer to them as \emph{relations}.
\end{remark}

Furthermore, we define the \emph{support $\support_R$} of a relation $R$ on $A$ as
\[
\support_R \coloneqq \{p\in A \, : \, \text{ there is some } q\in A \text{ with }  (p,q) \in R \text{ or } (q,p) \in R\},
\]   
that is, the subset of $A$ that contains precisely those $p\in A$ that are in $R$-relation with some
$q\in A$. We often consider relations $R$ on $A =\pairs(X)$ in which case we extend $\support_R$ to
obtain $\support_{R}^+ \coloneqq \support_{R}\ \cup\ \{xx\, : \, x\in X\}$.

We shall sometimes write $a\ R\ b$ instead of $(a,b)\in R$. In addition, we write $p\ R\ \dots\ R\
q$ if there is a \emph{$(p,q)$-chain in $R$}, i.e., there are some $a_0,\dots,a_k$, $k\geq 1$ such
that $a_0 = p\ R\ a_1\ R\ \dots\ R\ a_{k-1} R\ a_k=q$.

Let $R$ be a relation on $A$. Then, $R$ is \emph{asymmetric} if $(p,q) \in R$ implies $(q,p) \notin
R$ for all $p,q\in A$, and it is \emph{anti-symmetric} if $(p,q) \in R$ and $(q,p) \in R$ implies
$p=q$ for all $p,q\in A$. Moreover, $R$ is \emph{transitive}, if $(p,q) \in R$ and $(q,r) \in R$
implies $(p,r) \in R$ for all $p,q,r\in A$. For a subset $B\subseteq A$, we say that $R$ is
\emph{$B$-reflexive} if $(b,b)\in R$ for every $b\in B$. A relation $R$ on $A$ is \emph{reflexive}
if it is $A$-reflexive. If $F$ is a relation on $A$, then $R$ is \emph{$F$-conditional-symmetric} if
$(p,q) \in F\cap R$ implies $(q,p) \in R$. Moreover, a relation $R$ on $\pairs(X)$ is
\emph{cross-consistent} if $ab\in \support_R$, $(ac,xy) \in R$, and $(bd,xy) \in R$ for some $c,d\in
X$ implies that $(ab,xy) \in R$. We emphasize that, by definition, $ac=ca$ and $bd=db$ holds and,
therefore, if $(ac,xy) \in R$ and $(bd,xy) \in R$, and if $R$ is cross-consistent, then for every
$pq \in \{ab,ad,cb,cd\} \cap \support_R$, we have $(pq,xy) \in R$.

A \emph{poset} $(A, \leq)$ is a set $A$ equipped with a partial order $\leq$, i.e., a relation
$\leq$ on $A$ that is reflexive, transitive, and anti-symmetric.

A \emph{closure operator} on a set $S$ is a map $\phi \,:\,\mathcal{P}(S) \to\mathcal{P}(S) $ that
satisfies the following three properties for all $R,R'\in \mathcal{P}(S)$: \emph{Extensivity}: $R
\subseteq \phi(R)$, \emph{Monotonicity}: $R \subseteq R'$ implies $\phi(R) \subseteq \phi(R')$, and
\emph{Idempotency}: $\phi(\phi(R)) = \phi(R)$ \cite{CASPARD2003241,BS:95,SH:18}. We let $\tc(R)$
denote the \emph{transitive closure} of a relation $R$, that is, the inclusion-minimal relation that
is transitive and that contains $R$ (see e.g. \cite[p.39]{matouvsek2009invitation}). It is
straightforward to verify that $\tc$ is indeed a closure operator on $S = A\times A$ for all
relations $R$ on $A$.
We further define for two relations $R$ and $F$ on $\pairs(X)$, the relation 
\[
\FR \coloneqq \{(ab,cd) \in F \, : \, ab,cd \in \support^+_R\}. 
\]

\paragraph{DAGs and Networks.}

A \emph{directed graph $G=(V,E)$} is a pair with non-empty vertex set $V(G)\coloneqq V$ and arc set
$E(G) \coloneqq E \subseteq V\times V$. We put $\outdeg_G(v)\coloneqq\left|\left\{u\in V \colon
(v,u)\in E\right\}\right|$ and $\indeg_G(v)\coloneqq\left|\left\{u\in V \,:\, (u,v)\in
E\right\}\right|$ to denote the \emph{out-degree} and \emph{in-degree} of a vertex $v$,
respectively. A vertex $v$ with $\outdeg_G(v)=0$ is a \emph{leaf} of $G$ and a vertex $v$ with
$\indeg_G(v)=0$ is a \emph{root} of $G$. We write $L(G)$ for the set of leaves of $G$. If $G$ is a
DAG whose leaf set $L(G)$ is $X$, then $G$ is a \emph{DAG on $X$}. A directed graph $G$ is
\emph{phylogenetic} if it does not contain a vertex $v$ such that $\outdeg_G(v)=1$ and
$\indeg_G(v)\leq 1$. We sometimes use $u\to v$ to denote the arc $(u,v)\in E(G)$ and $u\leadsto v$
to denote a directed $uv$-path in $G$. Note that we allow ``trivial'' paths $u\leadsto u$ consisting
just of the single vertex $u$.

Directed graphs $G$ without directed cycles are called \emph{directed acyclic graphs (DAGs)}
\cite{book:digraph}. Let $G$ be a DAG. If $u\to v$ is an arc in $G$, then we call $v$ a \emph{child}
of $u$ and $u$ a \emph{parent} of $v$. We write $v\preceq_G u$ if and only if there is a directed
$uv$-path $u\leadsto v$ in $G$ and call, in this case, $v$ a \emph{descendant} of $u$ and $u$ an
\emph{ancestor} of $v$. If $v\preceq_G u$ and $v\neq u$, we write $v\prec_G u$. 

A \emph{network} is a DAG with a single root. A \emph{(rooted) tree} is a network that does not
contain vertices $v$ with $\indeg_G(v)>1$. A tree is \emph{binary}, if each non-leaf vertex has exactly
two children.

For a poset $({Q},\leq)$, the \emph{Hasse diagram} $\Hasse(Q,\leq)$ is the DAG with vertex set $Q$
and arcs $(A,B)$ if (i) $B\leq A$ and $A\neq B$ and (ii) there is no $C\in Q$ with $B\leq C\leq A$
and $C\neq A,B$. 

\paragraph{Least Common Ancestors and Triples.}
For a given DAG $G$ on $X$ and $x,y \in X$, a vertex $v\in V(G)$ is a \emph{common ancestor of $x$
and $y$} if $v$ is an ancestor of both $x$ and $y$. Moreover, $v$ is a \emph{least common ancestor}
(LCA) of $x$ and $y$ if $v$ is a $\preceq_G$-minimal common ancestor of $x$ and $y$. The set
$\LCA_G(xy)$ comprises all LCAs of $x$ and $y$ in $G$. In a network $N$ on $X$, the unique root is a
common ancestor for all $x,y \in X$ and, therefore, $\LCA_N(xy)\neq\emptyset$. Moreover, we are
interested in DAGs where $|\LCA_G(xy)|=1$ holds for certain $x,y\in X$. For simplicity, we will
write $\lca_G(xy)=v$ in case that $\LCA_G(xy)=\{v\}$ and say that \emph{$\lca_G(xy)$ is
well-defined}. A DAG $G$ on $X$ has the \emph{2-lca-property} if the LCA $\lca_G(xy)$ is
well-defined for all $x,y \in X$. 

For a given DAG $G$ on $X$, we further define the relation $\rel_G$ on $\pairs(X)$ as
\[\rel_G \coloneqq \{(ab,xy) \, : \, \lca_G(ab), \lca_G(xy) \text{ are well-defined and }
\lca_G(ab)\preceq_G \lca_G(xy)\}.\]

For later reference, we provide the following simple but useful result. 

\begin{lemma}\label{lem:xyz-lca(yz)} 
  Let $G$ be a DAG on $X$ such that $\lca_G(xy)$, $\lca_G(xz)$, and $\lca_G(yz)$ are well-defined
  for three leaves $x,y,z\in X$. If $\lca_G(xy)\prec_G\lca_G(xz)$, then
  $\lca_G(yz)\preceq_G\lca_G(xz)$.
\end{lemma}
\begin{proof}
	Let $G$ be a DAG on $X$ and suppose that $\lca_G(xy)$, $\lca_G(xz)$, and $\lca_G(yz)$ are
	well-defined for three leaves $x,y,z\in X$. Assume that $\lca_G(xy)\prec_G\lca_G(xz)$. This
	together with $y\preceq_G \lca_G(xy)$ and $z\preceq_G \lca_G(xz)$ implies that the vertex
	$\lca_G(xz)$ must be a common ancestor of $y$ and $z$. By the defining property of LCAs, we
	therefore have $\lca_G(yz)\preceq_G\lca_G(xz)$. 
\end{proof}

To recall, a (rooted) triple $t = xy|z$ is a binary tree on three distinct leaves $x$, $y$, and $z$ and with two
non-leaf vertices such that $\lca_t(xy) \prec_t \lca_t(xz) = \lca_t(yz)$. Moreover, an anchored
triple $t=\underline{x}y|z$ is an ordered tuple defined on three distinct elements $x,y,z$. We
consider here sets of rooted and anchored triples and define the set of leaves occurring in those
triple sets as follows. 
\begin{definition}\label{def:XRF}
	For a rooted triple $t=xy|z$ or anchored triple $t=\underline{x}y|z$, we put $L(t)=\{x,y,z\}$. For
	sets $\RR$ and $\FF$ of rooted or anchored triples, we define \[\XR\coloneqq \bigcup_{t\in \RR}
	L(t) \ \text{and}\ \XRF \coloneqq \bigcup_{t\in \RR \cup \FF} L(t).\]
\end{definition}

We now show that \axiom{T1} implies \axiom{T2} as mentioned in the introduction. 
\begin{lemma}\label{lem:T1=>T2}
Let $G$ be a DAG on $X$ and $x,y,z\in X$. If $xy|z$ is displayed by $G$, then $xy|z$ is
topologically-displayed by $G$.
\end{lemma}
\begin{proof}
	Let $G$ be a DAG on $X$. Suppose that $xy|z$ is displayed by $G$. Hence, $x$, $y$, and $z$ are
	pairwise distinct. Moreover, $u\prec_G v$ for 
    $u\coloneqq\lca_G(xy)$ and $v\coloneqq\lca_G(xz)=\lca_G(yz)$. In particular, $u$ and $v$ are
	distinct. Since $x,y\preceq_G u$, there are directed paths $P_{ux}=u \leadsto x$ and $P_{uy} =u
	\leadsto y$. Since $u=\lca_G(xy)$, there is, by definition of LCAs, no descendant $w\neq u$ of $u$
	that is an ancestor of both $x$ and $y$. This immediately implies that $P_{ux}$ and $P_{uy}$ can
	only intersect in $u$ and are, therefore, internally vertex-disjoint. Moreover, $u,z\preceq_G v$
	implies that there are directed paths $P_{vu}=v \leadsto u$ and $P_{vz} =v \leadsto z$. Note that
	$P_{vu}$ can be extended to a $vx$-path $P_{vx}$ by joining $P_{vu}$ and $P_{ux}$. By the same
	arguments used to show that $P_{ux}$ and $P_{uy}$ are internally vertex-disjoint, $v=\lca_G(xz)$
	implies that $P_{vx}$ and $P_{vz}$ must be internally vertex-disjoint. Since $P_{vx}$ is composed
	of $P_{vu}$ and $P_{ux}$, the paths $P_{vu}$ and $P_{vz}$ as well as $P_{ux}$ and $P_{vz}$ are
	internally vertex-disjoint. By analogous argumentation, $P_{uy}$ and $P_{vz}$ are internally
	vertex-disjoint. Finally, if $P_{vu}$ and $P_{ux}$ are not internally vertex-disjoint, then there
	is a vertex $w\neq u,v,x$ that is located on $P_{vu}$ and $P_{ux}$. In this case, we have
	$u\prec_G w$ and $w\prec_G u$; a contradiction to $G$ being a DAG. Hence, $P_{vu}$ and $P_{ux}$
	are internally vertex-disjoint and, by similar arguments, $P_{vu}$ and $P_{uy}$ are internally
	vertex-disjoint. In summary, $xy|z$ is topologically-displayed by $G$.
\end{proof}

The converse of Lemma~\ref{lem:T1=>T2} is, in general, not satisfied. For example, the network $N_1$
in Figure~\ref{fig:tiplesVSlca} topologically-displays $ab|c$ but does not display $ab|c$, since
$\lca_{N_1}(bc)\prec_G \lca_{N_1}(ab)=\lca_{N_1}(ac)$. In particular, a DAG $G$ can
topologically-display several triples on the same leaf set; for instance, $ab|c$ and $bc|a$ are both
topologically-displayed by $N_1$ in Figure~\ref{fig:tiplesVSlca}. By contrast, any DAG $G$ can display
at most one triple on any fixed three-element leaf set $\{a,b,c\}$ in the LCA sense.

\paragraph{Transformation of DAGs and Networks.}

Throughout this paper, we often transform a DAG into a network and will use the following simple result.

\begin{lemma}[From DAGs to networks, {\cite[Lem~2.2]{EH:26arxiv}}]\label{lem:DAG2Network}
	Let $G$ be a DAG on $X$. Let $N$ be the directed graph obtained from $G$ by either putting
	$N\coloneqq G$ if $G$ is a network or by adding a new vertex $\rho$ to $G$ together with the arcs
	$(\rho,\rho_i)$ for all roots $\rho_1,\dots,\rho_k$, $k\geq 2$ of $G$. Then, $N$ is a network on
	$X$ such that $u\preceq_G v$ if and only if $u\preceq_N v$ for all $u,v\in V(G)$. In particular,
	if $\LCA_G(xy)\neq \emptyset$, then $\LCA_G(xy)=\LCA_N(xy)$. Hence, if $\lca_G(xy)$ is
	well-defined in $G$, then $\lca_G(xy)=\lca_N(xy)$. Moreover, if $G$ is phylogenetic, then $N$ is
	phylogenetic.
\end{lemma}

Additionally, we make frequent use of an $xy$-extension based on two distinct leaves $x,y$ of a DAG
$G$. This construction, transforms $G$ into a DAG $G'$ in which $\lca_{G'}(xy)$ is not well-defined.
Let $G$ be a DAG on $X$ and $x$ and $y$ be two distinct leaves. An \emph{$xy$-extension of $G$} is
obtained from $G$ by adding two new vertices $u,v$ to $V(G)$ and the arcs
$\{(u,x),(u,y),(v,x),(v,y)\}$ to $E(G)$. One easily verifies that no cycles are introduced by the
additional vertices and arcs in an $xy$-extension of a DAG. Moreover, by construction, both $u$ and
$v$ are parents of $x$ and $y$, and neither has any further children. Hence, both $u$ and $v$ are
LCAs of $x$ and $y$ in the $xy$-extension. Furthermore, in any $xy$-extension $G'$ of $G$, we have
$a\preceq_G b$ if and only if $a\preceq_{G'} b$ for all $a,b\in V(G)=V(G')\setminus \{u,v\}$.
Moreover, the new vertices $u$ and $v$ have only $x,y$, and respectively $u$ and $v$, as
descendants, and have no ancestors. These observations imply that $\LCA_G(ab)=\LCA_{G'}(ab)$ for all
$a,b\in X$ with ${a,b}\neq\{x,y\}$. We summarize this discussion into
\begin{observation}\label{obs:xy-extension1} Let $G$ be a DAG on $X$ and $x,y \in X$ be distinct
leaves. The $xy$-extension $G'$ of $G$ is a DAG on $X$ and it holds that 
\begin{itemize}
    \item[(i)] $|\LCA_{G'}(xy)|>1$ and, thus, $\lca_{G'}(xy)$ is not well-defined, and
    \item[(ii)] $a \preceq_G b$ if and only if $a \preceq_{G'} b$ for $a,b \in V(G)$, and
    \item[(iii)] $\LCA_G(ab) = \LCA_{G'}(ab)$ for all $a,b \in X$ with $\{a,b\} \neq \{x,y\}$. 
\end{itemize}
\end{observation}

We group together several such extensions as follows.

\begin{definition}\label{def:FRextension}
	For a given pair $(R,F)$ of relations on $\pairs(X)$, the \emph{$\FR$-extension of a DAG $G$ on
	$X$} is obtained from $G$ by stepwise application of $xy$-extensions for all $xy \in \support_F
	\setminus \support^+_R$.   
\end{definition}

%
%

\section{Realization of LCA-constraints} 
\label{sec:LCA-rel}

We provide polynomial-time solutions for the triple-based problems in
Sections~\ref{sec:ATC-F}--\ref{sec:strong_and_LCA_problem}. To this end, we translate required and
forbidden sets of triples into suitable relations encoding LCA-constraints. The resulting relations
serve as input for realization problems, whose solutions yield solutions to the corresponding
triple-based problems. We therefore recall, in condensed form, the relevant definitions and results
on realizable relations and \RF-realizable pairs of relations introduced
in~\cite{LAMSH:25,EH:26arxiv}.

\begin{definition}[(Strict) Realization, {\cite[Def~3, Def~4]{LAMSH:25}}]\label{def:strict-real}
    Let $R$ be a relation on $\pairs(X)$ and $G$ be a DAG on $X$. Then, $R$ \emph{is strictly
    realized by} $G$ if, for all $ab, cd \in \support^+_{R}$, the LCAs $\lca_G(ab)$ and $\lca_G(cd)$
    are well-defined and the following implication holds for $G$ and $R$:
	 \begin{itemize}
      \item[\axiom{I0}] $(ab,cd) \in R$ implies that $\lca_G(ab)\prec_G\lca_G(cd)$.
   \end{itemize} 
\medskip

\noindent
	Moreover, $R$ \emph{is realized by} $G$ if, for all $ab, cd \in \support^+_{R}$, the LCAs
	$\lca_G(ab)$ and $\lca_G(cd)$ are well-defined and the following implications hold for $G$ and
	$R$: 
\begin{itemize}
    \item[\axiom{I1}] $(ab,cd)\in R$ and $(cd,ab) \notin \tc(R)$ implies that $\lca_G(ab) \prec_G \lca_G(cd)$. 
    \item[\axiom{I2}] $(ab,cd)\in R$ and $(cd,ab) \in \tc(R)$ implies that $\lca_G(ab) = \lca_G(cd)$.  
\end{itemize}
In this case, we say that $R$ is \emph{realizable}. 
\end{definition}

Note that a relation $R$ is strictly realized by $G$ if and only if $R$ is realized by $G$ and
$\tc(R)$ is asymmetric, i.e., strict realizability is just a special case of realizability
\cite[Lem~5]{LAMSH:25}.

We further consider pairs $(R,F)$ of relations on $\pairs(X)$, i.e., both $R$ and $F$ are relations
on $\pairs(X)$. In what follows, given such a pair $(R,F)$, the relation $R$ represents ``required
(\texttt{R})'' LCA-constraints, while $F$ represents ``forbidden (\texttt{F})'' LCA-constraints,
which motivates the following definition of realization.

\begin{definition}[{(Strict) \RF-Realization, \cite[Def~4.1, Def~4.20]{EH:26arxiv}}]
Let $(R, F)$ be a pair of relations on $\pairs(X)$ and $G$ be a DAG on $X$. Then, $(R, F)$ \emph{is
\RF-realized} (resp., \emph{strictly \RF-realized}) by $G$, if $R$ is realized (resp., strictly
realized) by $G$ and the following implication holds for $G$ and $F$: 
\begin{itemize}
    \item[\axiom{F}] If $(ab,cd) \in F$ and $\lca_G(ab)$ and $\lca_G(cd)$ are well-defined, 
                     then $\lca_G(ab) \nprec_G \lca_G(cd)$.
\end{itemize}
In this case, we say that $(R,F)$ is \emph{\RF-realizable} (resp., \emph{strictly \RF-realizable)}.
\end{definition}

Central to this theory are the closure operators $\cl$ and $\Fcl$. To be more precise, let $R$ and
$F$ be relations on $\pairs(X)$. We define $\mathfrak{R}_{R}$ as the set of all relations $R'$ on
$\pairs(X)$ that are $\support^+_{R}$-reflexive, transitive, cross-consistent, and satisfy
$R\subseteq R'$. Moreover, $\mathfrak{R}_{R,F} \subseteq \mathfrak{R}_{R}$ is the set that contains
all relations in $\mathfrak{R}_{R}$ that are, in addition, $F$-conditional-symmetric. This, in turn,
gives rise to the following intersections
\[\cl(R) \coloneqq \bigcap_{R' \in \mathfrak{R}_{R}} R' \quad \text{ and } \quad
\Fcl(R) \coloneqq \bigcap_{R' \in \mathfrak{R}_{R,F}} R'.
\]

Although abstractly defined, the closure $\cl(R)$ has a natural semantic interpretation when $R$ is
realizable: it is the smallest relation extending $R$ that is realized by every DAG realizing $R$.
In other words, $\cl(R)$ consists precisely of the LCA-constraints that are forced by realizability
of $R$. Analogously, for pairs $(R,F)$, the closure $\Fcl(R)$ captures the constraints forced by
\RF-realizability of $(R,F)$. 
Readers interested in further intuition and examples are referred to
\cite{LAMSH:25,EH:26arxiv}.

As shown in \cite[Prop~14]{LAMSH:25} and \cite[Prop~4.8]{EH:26arxiv}, $\cl$ and $\Fcl$ are closure
operators, i.e., they satisfy the classical closure axioms \emph{extensivity}, \emph{monotonicity},
and \emph{idempotency}. Moreover, one easily verifies that $\cl(R)=\cl_{\emptyset}(R)\subseteq
\Fcl(R)$. In addition, there is a simple rule set to efficiently compute $\cl(R)$ and $\Fcl(R)$
\cite[Thm~17]{LAMSH:25} and \cite[Thm~4.7]{EH:26arxiv}. We collect these results in

\begin{theorem}[{\cite{LAMSH:25,EH:26arxiv}}]\label{thm:cl-polytime}
For any relations $R$ and $F$ on $\pairs(X)$, $\cl$ and $\Fcl$ are closure operators and $\cl(R)$
and $\Fcl(R)$ can be determined in polynomial time in $|X|$.  
\end{theorem} 

To characterize realizable relations $R$, respectively, \RF-realizable relations $(R,F)$, we
construct a \emph{canonical DAG} $\GG_R$, respectively, $\GG_{R,F}$ based on the structure of
$\Fcl(R)$.

\begin{definition}[{The quotient poset, \cite[Def~4.14]{EH:26arxiv}}]
	Let $(R,F)$ be a pair of relations on $\pairs(X)$. We define the equivalence relation
	$\sim_{\Fcl(R)}$ on $\support^+_R$ by putting, for all $p,q \in \support^+_R$, \[p \sim_{\Fcl(R)}
	q \iff (p,q) \in \Fcl(R) \text{ and } (q,p) \in \Fcl(R). \] Let $[p]$ denote the equivalence class
	of $\sim_{\Fcl(R)}$ that contains $p\in\support_R^+$, and let $Q$ denote the set of all such
	equivalence classes. Define the partial order $\leq_{\Fcl(R)}$ on $Q$ by putting, for all classes
	$[p]$ and $[q]$ in $Q$, \[[p] \leq_{\Fcl(R)} [q] \iff (p,q) \in \Fcl(R).\] We refer to the poset
	$(Q,\le_{\Fcl(R)})$ as the \emph{quotient poset of $\Fcl(R)$}.
\end{definition}

An analogous quotient poset was defined in~\cite{LAMSH:25} for $R$ and its closure $\cl(R)$. The
definition given here is more general and reduces naturally to the poset $(Q,\leq_{\cl(R)})$ in the
case $F=\emptyset$, since $\cl(R)=\cl_{\emptyset}(R)$. The well-definedness and the facts that
$\sim_{\Fcl(R)}$ is an equivalence relation and that $(Q,\le_{\Fcl(R)})$ is a poset follows from the
standard ``quotient'' construction that turns a preorder into a partial order; see
\cite[Prop~5.2.4]{Schroder:03} for a full proof. The quotient poset $(Q,\le_{\Fcl(R)})$ associated
with $\Fcl(R)$ provides a canonical way to encode the order information imposed by the
LCA-constraints in $R$ and $F$. Its Hasse diagram therefore serves as a natural candidate for a
realizing DAG. This yields the following construction.
\begin{definition}[{Canonical DAG, \cite[Def~4.15]{EH:26arxiv}}]\label{def:canonDAGN}
Let $(R,F)$ be a pair of relations on $\pairs(X)$ and let $(Q, \le_{\Fcl(R)})$ be the quotient poset
of $\Fcl(R)$. The \emph{canonical DAG $\GG_{R,F}$ of $(R,F)$} is defined as the DAG obtained from
the Hasse diagram $\Hasse(Q,\leq_{\Fcl(R)})$ by relabeling each vertex $[aa]\in Q$ with $a$. The
\emph{canonical DAG $\GG_{R}$ of $R$} is defined as $\GG_R\coloneqq \GG_{R,\emptyset}$.
\end{definition}

It is technically possible that $[aa]=[bb]\in Q$ for distinct $a,b\in X$ in the quotient poset
$(Q,\leq_{\Fcl(R)})$ associated with an arbitrary pair of relations $(R,F)$. In this case, the same
vertex would have to be relabeled both by $a$ and by $b$, so the construction of $\GG_{R,F}$ would
not be well-defined without an additional convention. We avoid this ambiguity by using a fixed but
arbitrary linear order on the finite ground set of the relation $R$ to choose a canonical label for
such classes. However, all relations appearing in Section~\ref{sec:ATC-F} and later are
constructed from sets of rooted or anchored triples in such a way that their canonical DAGs are
well-defined without referring to this tie-breaking convention.

\begin{proposition}[\cite{LAMSH:25,EH:26arxiv}]\label{prop:properties_of_canonical_DAG}
Let $(R,F)$ be a pair of relations on $\pairs(X)$. If $R$ is realizable, then $\GG_R$ is a
phylogenetic DAG on $X$. If $(R,F)$ is \RF-realizable, then $\GG_{R,F}$ is a phylogenetic DAG on
$X$.
\end{proposition}

The following results summarize the realizability criteria that will be used throughout the paper.
In the setting without forbidden constraints, realizability of a relation $R$ is characterized by
two conditions, \axiom{X1} and \axiom{X2}. For pairs $(R,F)$ of required and forbidden relations,
analogous conditions \axiom{Y1} and \axiom{Y2} characterize \RF-realizability. In both settings, the
characterizations are constructive and yield polynomial-time recognition and construction
algorithms.

\begin{theorem}[{\cite[Thm~31, Thm~38]{LAMSH:25}}]\label{thm:char}
Let $R$ be a relation on $\pairs(X)$. Then the following statements are equivalent:
\begin{enumerate}
    \item $R$ is realizable.

    \item $R$ satisfies the following two conditions:
    \begin{description}
        \item[$\qquad$\axiom{X1}:]
        For all $a,b,x\in X$, if $ab\neq xx$, then
        $(ab,xx)\notin R$.

        \item[$\qquad$\axiom{X2}:]
        For all $a,b,x,y\in X$, if $(ab,xy)\in R$ and
        $(xy,ab)\notin \tc(R)$, then $(xy,ab)\notin \cl(R)$.
    \end{description}

    \item $R$ is realized by its canonical DAG $\GG_R$.
\end{enumerate}
Moreover, it can be decided in polynomial time in $|X|$ whether $R$ is realizable. If this is the
case, then the canonical DAG $\GG_R$ realizing $R$ can be constructed within the same time bound.
\end{theorem}

The corresponding characterization for pairs $(R,F)$ is obtained by replacing $R$ in \axiom{X1} and
$\cl(R)$ in \axiom{X2} by $\Fcl(R)$. Strict \RF-realizability is characterized by the same
conditions together with asymmetry of $\tc(R)$. The following summarizes \cite[Thm~4.18, Thm~4.19,
Thm~4.22, Thm~4.23]{EH:26arxiv}.

\begin{theorem}[{\cite{EH:26arxiv}}]\label{thm:characterization_AF_realized}
Let $(R,F)$ be a pair of relations on $\pairs(X)$.
\begin{enumerate}
    \item The following statements are equivalent:
    \begin{enumerate}
        \item $(R,F)$ is \RF-realizable.

        \item $(R,F)$ satisfies the following two conditions:
        \begin{description}
            \item[$\qquad$\axiom{Y1}:]
            For all $a,b,x\in X$, if $ab\neq xx$, then
            $(ab,xx)\notin \Fcl(R)$.

            \item[$\qquad$\axiom{Y2}:]
            For all $a,b,x,y\in X$, if $(ab,xy)\in R$ and
            $(xy,ab)\notin \tc(R)$, then $(xy,ab)\notin \Fcl(R)$.
        \end{description}

        \item $(R,F)$ is \RF-realized by the $\FR$-extension of its
        canonical DAG $\GG_{R,F}$.
    \end{enumerate}

    \item The following statements are equivalent:
    \begin{enumerate}
        \item $(R,F)$ is strictly \RF-realizable.

        \item $(R,F)$ satisfies \axiom{Y1} and \axiom{Y2}, and $\tc(R)$ is
        asymmetric.

        \item $(R,F)$ is strictly \RF-realized by the $\FR$-extension of its
        canonical DAG $\GG_{R,F}$.
    \end{enumerate}
\end{enumerate}
Moreover, it can be decided in polynomial time in $|X|$ whether $(R,F)$ is \RF-realizable,
respectively, strictly \RF-realizable. If this is the case, then the $\FR$-extension of the
canonical DAG $\GG_{R,F}$ can be constructed within the same time bound.
\end{theorem}

For further details as well as many illustrative examples, we refer the reader to
\cite{LAMSH:25,EH:26arxiv}. We also note that LCA-relations have recently been studied in related
contexts, including the simplification and the encoding of
phylogenetic networks~\cite{HL:24,HML:26}.

%
%
%

\section{Consistency Problems for Anchored Triples}
\label{sec:ATC-F}

We are now able to tackle triple-based problems, starting with consistency problems related to
anchored triples. Note that, by definition, for every anchored triple $\underline{x}y|z$, the
elements $x$, $y$, and $z$ are pairwise distinct. Recall that an anchored triple $\underline{x}y|z$
is \emph{\anchor-displayed} by a DAG or network $G$ if \begin{center} (i)\ $\lca_G(xy)$ and
$\lca_G(xz)$ are well-defined and (ii)\ $\lca_G(xy)\prec_G \lca_G(xz)$. \end{center} It is
well-known that all LCAs are unique in trees. Hence, a tree $T$ \anchor-displays $\underline{x}y|z$
if and only if Condition~(ii) holds. In general DAGs or networks, however, $\lca_G(xy)$ or
$\lca_G(xz)$ may fail to be well-defined; see for instance $N_1$ in Figure~\ref{fig:anchordisplay}
and the anchored triple $\underline{y}x|z$. We now consider the following two problems. 

\begin{problem}[\PROBLEM{Anchored Triples Consistency (ATC)}]\ \\
  \begin{tabular}{ll}
    \emph{Input:}    & A set $\RR$ of anchored triples. \\
    \emph{Question:} & Is there a phylogenetic network (resp., DAG) on $\XR$ that \anchor-displays
                       all anchored triples in $\RR$?
  \end{tabular}
\end{problem}

\begin{problem}[\PROBLEM{Anchored Triples Consistency with Forbidden Triples (ATC-F)}]\ \\
  \begin{tabular}{ll}
    \emph{Input:}    & Two sets $\RR$ and $\FF$ of anchored triples. \\
    \emph{Question:} & Is there a phylogenetic network (resp., DAG) on $\XRF$ that \anchor-displays
                       all anchored triples in $\RR$ \\ & but none in $\FF$?
  \end{tabular}
\end{problem}

Note that \PROBLEM{ATC} is just a special case of \PROBLEM{ATC-F} obtained by taking
$\FF=\emptyset$. Consequently, a polynomial-time solution for \PROBLEM{ATC-F} immediately yields one
for \PROBLEM{ATC}.

Now consider an anchored triple $\underline{x}y|z$ that is not \anchor-displayed by a DAG $G$. By
definition, this happens precisely when one of the two requirements for \anchor-display fails or,
equivalently, when one of the following mutually exclusive conditions holds:
\begin{enumerate}
    \item[(i')] $\lca_G(xy)$ or $\lca_G(xz)$ is not well-defined or
    \item[(ii')] $\lca_G(xy)$ and $\lca_G(xz)$ are well-defined and $\lca_G(xy)\nprec_G \lca_G(xz)$.
\end{enumerate}

For convenience, we introduce the following terminology.
\begin{definition}\label{def:anchor-display}
Let $(\RR,\FF)$ be a pair of anchored triple sets. A DAG $G$ \emph{\anchor-agrees} with $(\RR, \FF)$
if $G$ \anchor-displays all anchored triples in $\RR$ but none of the ones in $\FF$.  
\end{definition}

Definition~\ref{def:anchor-display} is illustrated in Figure~\ref{fig:anchor_displaying_RF}, which
shows a phylogenetic DAG $G$ and a phylogenetic network $N$ that $\anchor$-agree with a pair
$(\RR,\FF)$ of anchored triple sets. We use this as a running example throughout this section.

\begin{figure}
    \centering
    \includegraphics[width=0.8\textwidth]{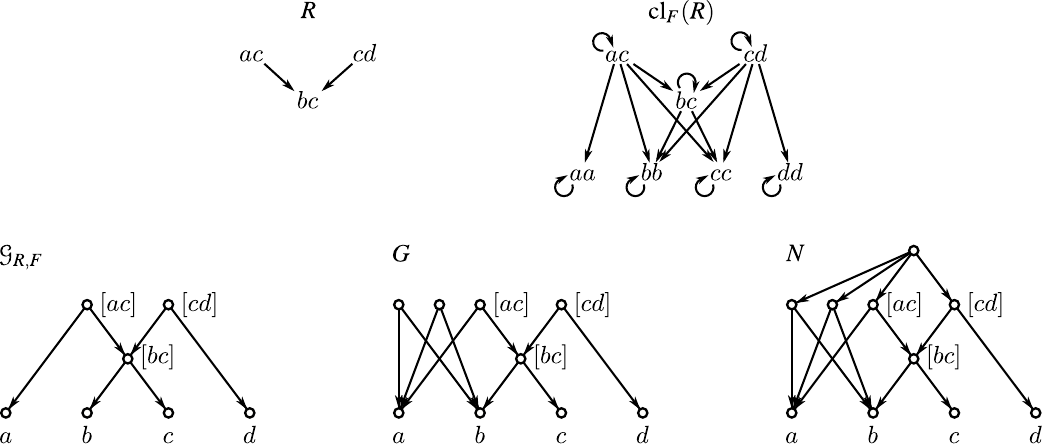}
    \caption{
    Consider the pair $(\RR,\FF)$ of anchored triple sets with $\RR = \{\underline{c}b|a, \underline{c}b|d\}$ and
    $\FF = \{\underline{b}c|a\}$. Then $R = \{(bc,ac),(bc,cd)\}$ and $F = \{(bc,ab)\}$ are the
    relations on $\pairs(\XRF)$ with $\XRF = \{a,b,c,d\}$ as defined in
    Theorem~\ref{thm:weaklyagrees_iff_strictly_real}. For $S \in \{R,\Fcl(R)\}$, an arc $p\to q$ is
    drawn precisely if $(q,p) \in S$. Then, $\GG_{R,F}$ is the canonical DAG of $(R,F)$ and $G$ is
    the $\FR$-extension of $\GG_{R,F}$, obtained by applying an $ab$-extension for the unique
    element in $\support_F\setminus\support_R^+=\{ab\}$ (cf. Definition~\ref{def:FRextension}).
    Since $\lca_{\GG_{R,F}}(bc)=[bc] \prec_{\GG_{R,F}} [ac]=\lca_{\GG_{R,F}}(ab)$, it follows that 
    $\GG_{R,F}$ \anchor-displays the forbidden triple $\underline{b}c|a$. Hence, $\GG_{R,F}$
    does not \anchor-agree with $(\RR,\FF)$. However, in accordance with
    Theorem~\ref{thm:weaklyagrees_iff_strictly_real}, $G$ \anchor-agrees with $(\RR,\FF)$ and
    strictly \RF-realizes $(R,F)$. In particular, $G$ \RF-realizes $(R,F)$ and $(ac,bc),(cd,bc)
    \notin \tc(R) = R$. Moreover, in accordance with Theorem~\ref{thm:weaklyagrees_invariant_DAG_network},
    the phylogenetic network $N$ obtained from $G$ by
    Lemma~\ref{lem:DAG2Network} also \anchor-agrees with $(\RR,\FF)$.}
    \label{fig:anchor_displaying_RF}
\end{figure}

As we shall see now, a pair $(\RR,\FF)$ of anchored triple sets can be translated into a pair
$(R,F)$ of relations in such a way that a DAG $G$ \anchor-agrees with $(\RR,\FF)$ if and only if $G$
strictly \RF-realizes $(R,F)$.

\begin{theorem}\label{thm:weaklyagrees_iff_strictly_real}
Let $(\RR, \FF)$ be a pair of anchored triple sets. Moreover, let $R \coloneqq \{(xy,xz) \, : \,
\underline{x}y|z \in \RR\}$ and $F \coloneqq \{(xy,xz) \, : \, \underline{x}y|z \in \FF\}$ be
relations on $\pairs(\XRF)$. Then, the following statements are equivalent for every DAG $G$ on
$\XRF$. 
\begin{enumerate}
    \item $G$ \anchor-agrees with $(\RR,\FF)$. 
    \item $G$ strictly \RF-realizes $(R,F)$.
    \item $G$ \RF-realizes $(R,F)$ and $(xz,xy) \notin \tc(R)$ for all $\underline{x}y|z \in \RR$.
\end{enumerate}
\end{theorem}
\begin{proof}
Let $(\RR,\FF)$ be a pair of anchored triple sets and define $R$ and $F$ as in the statement.
Moreover, let $G$ be a DAG on $\XRF$.

We first show that Statement~(1) implies Statement~(2). To this end, assume that $G$ \anchor-agrees
with $(\RR,\FF)$. We show that $(R,F)$ is strictly \RF-realized by $G$. Let $(xy,xz) \in R$. By
definition of $R$, there is an anchored triple $\underline{x}y|z \in \RR$. Since $G$ \anchor-displays
$\underline{x}y|z$, it follows that $\lca_G(xy)$ and $\lca_G(xz)$ are well-defined and $\lca_G(xy)
\prec_G \lca_G(xz)$. Thus, \axiom{I0} is satisfied for $(xy,xz)$ in $G$. Now let $(xy,xz) \in F$.
Hence, there is some $\underline{x}y|z \in \FF$. Since $G$ \anchor-agrees with $(\RR,\FF)$, either
at least one of the vertices $\lca_G(xy)$ and $\lca_G(xz)$ is not
well-defined or both are well-defined and it holds that $\lca_G(xy) \nprec_G \lca_G(xz)$. In both
cases, \axiom{F} is satisfied for $(xy,xz)$. Therefore, $G$ strictly \RF-realizes $(R,F)$. Hence,
Statement~(1) implies Statement~(2). 

Suppose now that Statement~(2) holds. To show that Statement~(3) holds, assume that $G$ strictly
\RF-realizes $(R,F)$. Lemma~4.21 of \cite{EH:26arxiv} implies that $G$ \RF-realizes $(R,F)$ and
$\tc(R)$ is asymmetric. This together with $(xy,xz) \in R \subseteq \tc(R)$ implies $(xz,xy) \notin
\tc(R)$ for all $\underline{x}y|z \in \RR$. Hence, Statement~(3) holds.

It remains to show that Statement~(3) implies Statement~(1). Thus, assume that $G$ \RF-realizes
$(R,F)$ and that $(xz,xy) \notin \tc(R)$ for all $\underline{x}y|z \in \RR$. We now show that $G$
\anchor-agrees with $(\RR,\FF)$. Suppose that $\underline{x}y|z \in \RR$ and, thus, $(xy,xz) \in R$.
Since $G$ \RF-realizes $(R,F)$, $G$ realizes $R$ and the vertices $\lca_G(xy)$ and $\lca_G(xz)$ are
well-defined. Moreover, by assumption, $(xz,xy) \notin \tc(R)$. Thus, the pre-conditions of
\axiom{I1} are met and it follows that $\lca_G(xy) \prec_G \lca_G(xz)$. Thus, $G$ \anchor-displays
$\underline{x}y|z$. Now suppose $\underline{x}y|z \in \FF$ and, thus, $(xy,xz) \in F$. Since $G$ and
$F$ satisfy \axiom{F}, either at least one of the vertices $\lca_G(xy)$ and $\lca_G(xz)$ is not
well-defined or both are well-defined and $\lca_G(xy) \nprec_G \lca_G(xz)$ is satisfied. In both
cases, $\underline{x}y|z$ is not \anchor-displayed by $G$. Thus, $G$ \anchor-agrees with $(\RR,\FF)$
and Statement~(3) implies Statement~(1). In summary, Statements~(1), (2), and (3) are equivalent. 
\end{proof}

We next show that the term ``\anchor-agree'' does not depend on whether one works with arbitrary
DAGs or with phylogenetic networks.

\begin{theorem}\label{thm:weaklyagrees_invariant_DAG_network}
Let $(\RR, \FF)$ be a pair of anchored triple sets. Then, the following statements are equivalent. 
\begin{enumerate}
    \item There is a DAG on $\XRF$ that \anchor-agrees with $(\RR,\FF)$. 
    \item There is a phylogenetic network on $\XRF$ that \anchor-agrees with $(\RR,\FF)$.
\end{enumerate}
In particular, suppose that there exists a DAG that \anchor-agrees with $(\RR,\FF)$, and let $R
\coloneqq \{(xy,xz) \, : \, \underline{x}y|z \in \RR\}$ and $F \coloneqq \{(xy,xz) \, : \,
\underline{x}y|z \in \FF\}$ be the corresponding relations on $\pairs(\XRF)$. Then the
$\FR$-extension $G$ of the canonical DAG $\GG_{R,F}$ as well as the network $N$ obtained from $G$ by
Lemma~\ref{lem:DAG2Network} are phylogenetic and \anchor-agree with $(\RR,\FF)$. 
\end{theorem}
\begin{proof} 
Let $(\RR,\FF)$ be a pair of anchored triple sets and let $R$ and $F$ be defined as in the
statement. We start with showing that Statement~(1) implies Statement~(2). Assume that there is a
DAG that \anchor-agrees with $(\RR,\FF)$. By Theorem~\ref{thm:weaklyagrees_iff_strictly_real},
$(R,F)$ is strictly \RF-realizable. Then, Theorem~\ref{thm:characterization_AF_realized} implies
that the $\FR$-extension $G$ of the canonical DAG $\GG_{R,F}$ strictly \RF-realizes $(R,F)$. 
This together with $R$ and $F$ being relations on 
$\pairs(\XRF)$ implies that the leaf set of $G$ is $\XRF$.  By
Theorem~\ref{thm:weaklyagrees_iff_strictly_real}, $G$ \anchor-agrees with $(\RR,\FF)$. Let $N$ be
the network obtained from $G$ according to Lemma~\ref{lem:DAG2Network}.
 By Lemma~\ref{lem:DAG2Network}, $N$ is a network on $\XRF$. We will
now show that $N$ \anchor-agrees with $(\RR,\FF)$. By Lemma~\ref{lem:DAG2Network}, $u\preceq_N v$ if 
and only if $u \preceq_G v$ for all $u,v \in V(G)$ 
and $\lca_N(xy)=\lca_G(xy)$ for all $x,y\in\XRF$
for which $\lca_G(xy)$ is well-defined. It readily follows that all anchored triples
$\underline{x}y|z\in \RR$ are \anchor-displayed by $N$. Now let $\underline{x}y|z \in \FF$ and,
thus, $(xy,xz) \in F$. There are two cases: either $xy\in \support_R^+$ or $xy\in
\support_F\setminus \support_R^+$. In the first case, $\lca_G(xy)$ is well-defined and, in the
second case an $xy$-extension was applied to obtain $G$ and Observation~\ref{obs:xy-extension1}
implies that $|\LCA_G(xy)|>1$. In either case, $\LCA_G(xy) \neq \emptyset$ and, thus, by
Lemma~\ref{lem:DAG2Network}, $\LCA_G(xy) = \LCA_N(xy)$ holds. Analogously, $\LCA_G(xz) = \LCA_N(xz)$
holds. Recall that Lemma~\ref{lem:DAG2Network} implies that, for all vertices $u,v\in V(G)$, we
obtain $u\preceq_G v$ if and only if $u\preceq_N v$. The latter two arguments together with the fact
that $G$ \anchor-agrees with $(\RR,\FF)$ imply that either at least one of the
vertices $\lca_N(xy)$ and $\lca_N(xz)$ is not well-defined or both vertices are well-defined and it
holds that $\lca_N(xy) \nprec_N \lca_N(xz)$. Consequently, $N$ does not \anchor-display
$\underline{x}y|z$. Hence, $N$ \anchor-agrees with $(\RR,\FF)$. Recall that $(R,F)$ is strictly
\RF-realizable and thus, by Theorem~\ref{thm:characterization_AF_realized}, also \RF-realizable.
Hence, Proposition~\ref{prop:properties_of_canonical_DAG} implies that the canonical DAG $\GG_{R,F}$
is phylogenetic. By construction, the $\FR$-extension $G$ remains phylogenetic and, thus, by
Lemma~\ref{lem:DAG2Network}, $N$ is phylogenetic. Hence, Statement~(1) implies Statement~(2).
Clearly, Statement~(2) trivially implies (1) and the equivalence between these statements follows.

Note that the used arguments, in particular, imply that if there exists a DAG that \anchor-agrees
with $(\RR,\FF)$, then the $\FR$-extension $G$ of the canonical DAG $\GG_{R,F}$ as well as the
network $N$ obtained from $G$ according to Lemma~\ref{lem:DAG2Network} are phylogenetic and
\anchor-agree with $(\RR,\FF)$. 
\end{proof}

\begin{figure}
\centering
\begin{minipage}{.9\textwidth}
\begin{algorithm}[H]
   \small
  \caption{\textsc{Anchored Triples Consistency with Forbidden Triples}}
  \label{alg:ATC-F}
  \begin{algorithmic}[1]
    \Require  A pair $(\RR,\FF)$ of anchored triple sets
    \Ensure  A phylogenetic DAG and a phylogenetic network on $\XRF$ that \anchor-agrees with $(\RR,\FF)$
    if one exists and, otherwise, \texttt{false} is returned
    \State Compute the relation $R \coloneqq \{(xy,xz) \, : \, \underline{x}y|z \in \RR\}$ on $\pairs(\XRF)$
    \State Compute the relation $F \coloneqq \{(xy,xz) \, : \, \underline{x}y|z \in \FF\}$ on $\pairs(\XRF)$
    \If{$(R,F)$ satisfies Condition \axiom{Y1} and \axiom{Y2}, and $\tc(R)$ is asymmetric}   \label{l:Y12}
      \State Compute the canonical DAG $\GG_{R,F}$ as in Definition~\ref{def:canonDAGN}
      \State Compute the $\FR$-extension $G$ of $\GG_{R,F}$ 
      \State Compute the network $N$ obtained from $G$ according to Lemma~\ref{lem:DAG2Network}
	   \State \Return $G$ and  $N$ 
 	\Else \ \Return \texttt{false} 
    \EndIf \label{l:EndIf}
  \end{algorithmic}
\end{algorithm}		
 \end{minipage}
\end{figure}

In Figure~\ref{fig:anchor_displaying_RF}, we provide an example of the canonical DAG $\GG_{R,F}$,
its $\FR$-extension $G$, and the phylogenetic network $N$ obtained from $G$. We are now able to show
that, whenever a DAG or network exists that \anchor-agrees with $(\RR,\FF)$, it can be constructed
in polynomial time.

\begin{theorem}\label{thm:ATCF-problem}
Let $(\RR,\FF)$ be a pair of anchored triple sets. Then, Algorithm~\ref{alg:ATC-F}, applied to $(\RR,\FF)$ decides 
in polynomial time in $|\XRF|$ whether there exists a phylogenetic DAG, equivalently a phylogenetic network, on $\XRF$
that \anchor-agrees with $(\RR,\FF)$ and, in the affirmative case, constructs such a phylogenetic DAG and a 
corresponding phylogenetic network within the same time bound.
\end{theorem}
\begin{proof}
Let $(\RR,\FF)$ be a pair of anchored triple sets and let $R \coloneqq \{(xy,xz) \, : \,
\underline{x}y|z \in \RR\}$ and $F \coloneqq \{(xy,xz) \, : \, \underline{x}y|z \in \FF\}$ be
relations on $\pairs(\XRF)$. By Theorem~\ref{thm:weaklyagrees_iff_strictly_real}, there is a DAG
that \anchor-agrees with $(\RR,\FF)$ if and only if $(R,F)$ is strictly \RF-realizable. Moreover, by
Theorem~\ref{thm:characterization_AF_realized}, the latter holds if and only if $(R,F)$ satisfies
\axiom{Y1} and \axiom{Y2} and $\tc(R)$ is asymmetric. Hence, if the condition in Line~\ref{l:Y12} of
Algorithm~\ref{alg:ATC-F} is not satisfied, the algorithm correctly returns \texttt{false}.
Otherwise, the algorithm proceeds and Theorem~\ref{thm:weaklyagrees_invariant_DAG_network} implies
that the $\FR$-extension of the canonical DAG $\GG_{R,F}$ and the network $N$ on $\XRF$ obtained
from it by Lemma~\ref{lem:DAG2Network} are phylogenetic and \anchor-agree with $(\RR, \FF)$. In
summary, Algorithm~\ref{alg:ATC-F} is correct.

For the runtime of Algorithm~\ref{alg:ATC-F}, observe that the construction of $R$ and $F$ can be
achieved in polynomial time in $|\XRF|$, since $|\RR|,|\FF| \in O(|\XRF|^3)$.
Moreover, Line~\ref{l:Y12}--\ref{l:EndIf} in Algorithm~\ref{alg:ATC-F} are, in essence, identical to
Line~3--8 in Algorithm 1 in \cite{EH:26arxiv} which runs in polynomial time in $|\XRF|$; cf.\
\cite[Thm.~4.23]{EH:26arxiv}. Hence, Algorithm~\ref{alg:ATC-F} runs in polynomial time in $|\XRF|$. 
\end{proof}

Taking $\FF=\emptyset$ in Theorem~\ref{thm:weaklyagrees_invariant_DAG_network} and
\ref{thm:ATCF-problem} immediately yields the corresponding result for \PROBLEM{ATC}. In this case
$F=\emptyset$, no $\FR$-extension is applied, and the canonical DAG is $\GG_R$.

\begin{corollary}\label{cor:ATC-problem}
Let $\RR$ be a set of anchored triples. Then, Algorithm~\ref{alg:ATC-F}, applied to
$(\RR,\emptyset)$, decides in polynomial time in $|\XR|$ whether there exists a
phylogenetic DAG, equivalently a phylogenetic network, on $\XR$ that \anchor-displays all anchored
triples in $\RR$ and, in the affirmative case, 
constructs such a phylogenetic DAG and a corresponding phylogenetic network within the same time
bound.

In particular, if there is a DAG that \anchor-displays all anchored triples in $\RR$, then the
canonical DAG $\GG_{R}$ of $R \coloneqq \{(xy,xz) \,:\, \underline{x}y|z \in \RR\}$ as well as the
network $N$ obtained from it by Lemma~\ref{lem:DAG2Network} are phylogenetic and \anchor-display all
anchored triples in $\RR$.
\end{corollary}

To summarize, we have shown that both problems \PROBLEM{ATC} and \PROBLEM{ATC-F} can be decided in
polynomial time.

%
%

\section{Consistency Problems for Rooted Triples}
\label{sec:TC-F}

In this section, we return to ``ordinary'' rooted triples, as opposed to anchored triples. Our aim
is to construct a phylogenetic DAG or network that displays all required triples and displays none
of the forbidden ones. To be more precise, we consider the following two problems and show that they
can be solved in polynomial time.

\begin{problem}[\PROBLEM{Triples Consistency (TC)}]\ \\
  \begin{tabular}{ll}
    \emph{Input:}    & A set $\RR$ of triples. \\
    \emph{Question:} & Is there a phylogenetic network (resp., DAG) on $\XR$ that displays all
                       triples in $\RR$?
  \end{tabular}
\end{problem}

\begin{problem}[\PROBLEM{Triples Consistency with Forbidden Triples (TC-F)}]\ \\
  \begin{tabular}{ll}
    \emph{Input:}    & Two sets $\RR$ and $\FF$ of triples. \\
    \emph{Question:} & Is there a phylogenetic network (resp., DAG) on $\XRF$ that displays all
                       triples in $\RR$ but none in $\FF$?
  \end{tabular}
\end{problem}
To recall, $x, y$, and $z$ are pairwise distinct for any triple $xy|z$. Moreover, a triple $xy|z$ is
displayed by a DAG $G$ if 
\begin{center}
	(i) $\lca_G(xy)$, $\lca_G(xz)$, and $\lca_G(yz)$ are all well-defined and 	
	(ii) $\lca_G(xy) \prec_G \lca_G(xz) = \lca_G(yz)$. 
\end{center}
Thus, a triple $xy|z$ is not displayed by a DAG $G$ if 
\begin{enumerate}
    \item[](i')\ $\lca_G(xy)$, $\lca_G(xz)$, or $\lca_G(yz)$ is not well-defined or\smallskip

		all LCAs $\lca_G(xy)$, $\lca_G(xz)$, and $\lca_G(yz)$ are well-defined and one of the following
		conditions holds
    \item[](ii')\ $\lca_G(xy)\nprec_G \lca_G(xz)$$\quad$ or $\quad$(iii')\ $\lca_G(xy)\nprec_G
                  \lca_G(yz)$$\quad$ or$\quad$ (iv')\ $\lca_G(xz)\neq \lca_G(yz)$.
\end{enumerate}

In the anchored triple setting, we considered the \PROBLEM{ATC} problem as a special case of
\PROBLEM{ATC-F} by taking $\FF=\emptyset$. In this setting, it was very useful to define the
forbidden relation $F=\{(xy,xz)\, : \, \underline{x}y|z\in\FF\}$. This definition works because an
anchored triple $\underline{x}y|z$ is not \anchor-displayed by a DAG $G$ whenever Condition~(ii')
holds, assuming that both $\lca_G(xy)$ and $\lca_G(xz)$ are well-defined. For ordinary rooted
triples, however, the situation is more subtle. In this case, non-displaying is characterized by the
disjunction of Conditions~(ii'), (iii'), and (iv'). This ``or'' structure makes it more difficult to
define a simple relation $F$ directly from $\FF$.

We therefore first consider the \PROBLEM{TC} problem and provide a polynomial-time algorithm for
deciding whether a set of triples $\RR$ can be displayed by a DAG or network and, if so, for
constructing such a DAG or network. To this end, we will show how a set of rooted triples can be
translated into a ``symmetric'' set of anchored triples, for which we can reuse the results of
Section~\ref{sec:ATC-F} in a straightforward manner. For the \PROBLEM{TC-F} problem, the connection
to anchored triples is no longer so direct. Therefore, we first translate the set of required
triples $\RR$ into the relation 
\[  R_\RR \coloneqq \{(xy,xz),(xy,yz) \,:\, xy|z \in \RR\},\]
on $\pairs(\XRF)$ and subsequently saturate it with additional
LCA-constraints induced by $\FF$. This saturated version of $R_{\RR}$ is then used to decide the
\PROBLEM{TC-F} problem in polynomial time. We note in passing that $R_{\RR}$ 
satisfies \axiom{X1} for all triple sets $\RR$. 

\subsection{Required Rooted Triples}
\label{subs:TC}

We first provide a simple but useful result
that links the notion of displaying rooted triples and \anchor-displaying anchored triples.

\begin{lemma}\label{lem:display-vs-anchordisplay}
	Let $G$ be a DAG on $X$ and $x,y,z\in X$. Then the following two statements are equivalent
    \begin{enumerate}
        \item $G$ displays $xy|z$.
        \item $G$ \anchor-displays both $\underline{x}y|z$ and $\underline{y}x|z$. 
    \end{enumerate}
\end{lemma}    
\begin{proof}
 Let $G$ be a DAG on $X$ and $x,y,z \in X$. If $G$ displays $xy|z$, then $\lca_G(xy)$, $\lca_G(xz)$,
 and $\lca_G(yz)$ are well-defined and it holds that $\lca_G(xy)\prec_G \lca_G(xz) = \lca_G(yz)$. By
 definition, $G$ \anchor-displays both $\underline{x}y|z$ and $\underline{y}x|z$. For the converse,
 assume that $G$ \anchor-displays both $\underline{x}y|z$ and $\underline{y}x|z$. Hence,
 $\lca_G(xy)$, $\lca_G(xz)$, and $\lca_G(yz)$ are well-defined and it holds that $\lca_G(xy) \prec_G
 \lca_G(xz)$ and $\lca_G(xy) \prec_G \lca_G(yz)$. Moreover, by Lemma~\ref{lem:xyz-lca(yz)}, we
 obtain $\lca_G(yz) \preceq_G \lca_G(xz)$ and $\lca_G(xz) \preceq_G \lca_G(yz)$. Consequently,
 $\lca_G(xy)\prec_G\lca_G(yz)=\lca_G(xz)$ holds. In other words, $G$ displays the triple $xy|z$. 
\end{proof}

The latter result allows us to reduce the problem of displaying rooted triples to the corresponding
problem for anchored triples.

\begin{theorem}\label{thm:TC-char} 
Let $\RR$ be a set of triples and put $\RR'\coloneqq \{\underline{x}y|z, \underline{y}x|z \,:\,
xy|z\in \RR\}$. Then the following statements are equivalent for every DAG $G$ on $\XR$. 
\begin{enumerate}
	\item $G$ displays all triples in $\RR$. 
	\item $G$ $\anchor$-displays all anchored triples in $\RR'$. 
    \item $G$ strictly realizes $R_\RR$.
    \item $G$ realizes $R_\RR$ and $(xz,xy), (yz,xy) \notin \tc(R_\RR)$ for all $xy|z \in \RR$.
\end{enumerate} 
\end{theorem}
\begin{proof}
 Let $\RR$, $\RR'$, and $G$ be as stated and let $R'\coloneqq\{(xy,xz)\,:\,
 \underline{x}y|z\in\RR'\}$ be a relation on $\pairs(\XRF)$. One easily verifies that $R'$ and
 $R_\RR$ coincide. By Lemma~\ref{lem:display-vs-anchordisplay}, Statements (1) and (2) are
 equivalent. Moreover, by definition, $G$ \anchor-displays all anchored triples in $\RR'$ if and
 only if $G$ \anchor-agrees with $(\RR',\emptyset)$. Therefore,
 Theorem~\ref{thm:weaklyagrees_iff_strictly_real} implies that Statement (2) is equivalent to $G$
 strictly \RF-realizing $(R',\emptyset)$ which, by definition, is equivalent to Statement (3).
 Similarly, Theorem~\ref{thm:weaklyagrees_iff_strictly_real} states that Statement (3) is equivalent
 to $G$ realizing $R'$ and $(xz,xy) \notin \tc(R')$ for all $\underline{x}y|z \in \RR'$. Note that
 by definition of $\RR'$ and $R'$, we have $\underline{x}y|z \in \RR'$ if and only if $xy|z\in \RR$
 if and only if $\underline{y}x|z \in \RR'$. Thus, $G$ realizes $R'$ and $(xz,xy) \notin \tc(R')$
 for all $\underline{x}y|z \in \RR'$ if and only if $G$ realizes $R'$ and $(xz,xy), (yz,xy) \notin
 \tc(R')$ for all $xy|z \in \RR$. Thus, Statements (3) and (4) are equivalent.
\end{proof}

In particular, if there exists a DAG that displays all required triples in $\RR$, there also exists
a phylogenetic network that does so.

\begin{theorem}\label{thm:TC-DAG_iff_network}
Let $\RR$ be a set of triples. Then, the following statements are equivalent. 
\begin{enumerate}
    \item There is a DAG on $\XR$ that displays all triples in $\RR$. 
    \item There is a phylogenetic network on $\XR$ that displays all triples in $\RR$. 
\end{enumerate}
In particular, suppose that there exists a DAG that displays all triples in $\RR$. Then the
canonical DAG $\GG_{R_\RR}$ as well as the network $N$ obtained from it by
Lemma~\ref{lem:DAG2Network} are phylogenetic and display all triples in $\RR$.
\end{theorem}
\begin{proof}
Let $\RR$ be a set of triples and $\RR' \coloneqq \{\underline{x}y|z, \underline{y}x|z \,:\, xy|z\in
\RR\}$. By Theorem~\ref{thm:TC-char}, a DAG displays all triples in $\RR$ if and only if it
\anchor-displays all anchored triples in $\RR'$, equivalently, if and only if it \anchor-agrees with
$(\RR',\emptyset)$. By Theorem~\ref{thm:weaklyagrees_invariant_DAG_network}, there is a DAG on
$X_{\RR',\emptyset}$ that \anchor-agrees with $(\RR',\emptyset)$ if and only if there is a
phylogenetic network $N$ on $X_{\RR',\emptyset}$ that \anchor-agrees with $(\RR',\emptyset)$. This
together with Theorem~\ref{thm:TC-char} and $X_{\RR',\emptyset} = \XR$ implies the equivalence of
Statement~(1) and (2).

Suppose now that there is a DAG $G$ that displays all triples in $\RR$. By
Theorem~\ref{thm:TC-char}, $G$ \anchor-displays all anchored triples in $\RR'$. Hence, by Corollary
\ref{cor:ATC-problem}, the canonical DAG $\GG_{R'}$ of $R'\coloneqq\{(xy,xz)\,:\,
\underline{x}y|z\in\RR'\}$ and the network $N$ obtained from it by Lemma~\ref{lem:DAG2Network} are
phylogenetic and \anchor-display all anchored triples in $\RR'$. Since $R' = R_\RR$, we have 
$\GG_{R'} = \GG_{R_\RR}$. Taking the latter arguments together with Theorem~\ref{thm:TC-char},
this implies that both $\GG_{R_\RR}$ and $N$ display all triples in $\RR$, proving the second part
of the theorem.
\end{proof}

Clearly, Theorem~\ref{thm:TC-char} together with Corollary~\ref{cor:ATC-problem} now implies that
the \PROBLEM{TC} problem can be decided in polynomial time.

\begin{theorem}\label{thm:newAlgo-TC}
Let $\RR$ be a set of triples and put $\RR' \coloneqq \{\underline{x}y|z, \underline{y}x|z \,:\,
xy|z\in \RR\}$. 
Then, Algorithm~\ref{alg:ATC-F}, applied to $(\RR',\emptyset)$, decides in polynomial time in
$|\XR|$ whether there exists a phylogenetic DAG, equivalently a phylogenetic network, on $\XR$ that
displays all triples in $\RR$ and, in the affirmative case, 
constructs such a phylogenetic DAG and a corresponding phylogenetic network within the same time
bound.
\end{theorem}
\begin{proof}
Let $\RR $ and $\RR'$ be as stated. By Corollary~\ref{cor:ATC-problem}, Algorithm~\ref{alg:ATC-F},
applied to $(\RR',\emptyset)$, decides in polynomial time in $|\XR|$ whether there exists a
phylogenetic DAG, equivalently a phylogenetic network, that \anchor-displays all anchored triples in
$\RR'$ and, in the affirmative case, Algorithm~\ref{alg:ATC-F}, applied to $(\RR',\emptyset)$,
constructs such a phylogenetic DAG and network within the same time bound. By
Theorem~\ref{thm:TC-char}, this is equivalent to the existence of a phylogenetic DAG or network that
displays all triples in $\RR$.
\end{proof}

\subsection{Required and Forbidden Rooted Triples} \label{subs:TC-F}

We now consider pairs $(\RR,\FF)$ of triple sets and design a polynomial-time algorithm for
\PROBLEM{Triples Consistency with Forbidden Triples (TC-F)}. Analogously to the notion of a DAG
\anchor-agreeing with a pair of anchored triple sets, we use the following terminology.

\begin{definition}
Let $(\RR,\FF)$ be a pair of triple sets. A DAG $G$ \emph{agrees} with $(\RR, \FF)$ if $G$ displays
all triples in $\RR$ but none of the ones in $\FF$.  
\end{definition}

To recall, a triple $xy|z$ is not displayed by a DAG $G$ if (i') at least one of $\lca_G(xy)$,
$\lca_G(xz)$, or $\lca_G(yz)$ is not well-defined or, otherwise, (ii') $\lca_G(xy)\nprec_G
\lca_G(xz)$ or (iii') $\lca_G(xy)\nprec_G \lca_G(yz)$ or (iv') $\lca_G(xz)\neq \lca_G(yz)$. We note
in passing that one can show that if a triple is not displayed by a DAG $G$, then at least one of
the conditions (i'), (ii'), and (iii') must be satisfied. Indeed, if all three LCAs are well-defined
and neither (ii') nor (iii') holds, then $\lca_G(xy)\prec_G\lca_G(xz)$ and
$\lca_G(xy)\prec_G\lca_G(yz)$. By applying Lemma~\ref{lem:xyz-lca(yz)} twice, we obtain
$\lca_G(xz)=\lca_G(yz)$ and, hence, $xy|z$ is displayed.

For us some other observations are more important. Our approach to obtain a DAG $G'$ that agrees
with a pair $(\RR,\FF)$ of triple sets is based, in essence, on the following ideas. First, recall
that for a given DAG $G$ on $X$, we defined the relation $\rel_G$ on $\pairs(X)$ as 
\[\rel_G \coloneqq \{(ab,xy) \, : \, \lca_G(ab), \lca_G(xy) \text{ are well-defined and }
  \lca_G(ab)\preceq_G \lca_G(xy)\}.
\] 

Now consider a forbidden triple $xy|z\in\FF$ and let $G$ be a DAG that displays all triples in
$\RR$. Two different situations can occur. First, suppose that not all of the three LCAs
$\lca_G(xy), \lca_G(xz),\lca_G(yz)$ are forced to be well-defined by the required triples in $\RR$.
In other words, there is some $ab\in\{xy,xz,yz\}$ such that no triple in $\RR$ requires $\lca_G(ab)$
to be well-defined. If $\lca_G(ab)$ is already not well-defined, then $G$ does not display $xy|z$, and
there is nothing to do. If, on the other hand, $\lca_G(ab)$ is well-defined, then we can apply an
$ab$-extension to obtain a DAG $G'$ in which $\lca_{G'}(ab)$ is no longer well-defined (cf.\
Observation~\ref{obs:xy-extension1}). By construction, this prevents $G'$ from displaying the
forbidden triple $xy|z$, while the triples in $\RR$ remain displayed. This first case is handled in
Definition~\ref{def:FFRR_extension} and Proposition~\ref{prop:reduction_of_forbidden_triple_sets}.

Second, suppose that the required triples force all three LCAs $\lca_G(xy)$, $\lca_G(xz)$, and
$\lca_G(yz)$ to be well-defined. In this case, non-display of $xy|z$ cannot be ensured merely by
making one of these LCAs not well-defined, since that would, in turn, prevent $G$ from displaying
all required triples. Instead, we have to modify the underlying LCA-constraints. If $xy|z$ is
displayed by $G$, then $ \lca_G(xy) \prec_G \lca_G(xz) = \lca_G(yz)$ and, thus, $
\{(xy,xz),(xy,yz),(xz,yz),(yz,xz)\}\subseteq \rel_G$ holds. The idea is then to add the ``reverse''
elements $(xz,xy)$ and $(yz,xy)$. For any DAG $G'$ realizing the resulting relation $Q$, this
enforces $ \lca_{G'}(xy)=\lca_{G'}(xz)=\lca_{G'}(yz)$. Consequently, $G'$ does not display $xy|z$.
Provided that the required triples in $\RR$ are still displayed in $G'$, this DAG $G'$ agrees with
$(\RR,\{xy|z\})$. This second case is handled systematically by the saturation procedure introduced
in Definition~\ref{def:method}, which forms the basis for solving the \PROBLEM{TC-F} problem in
polynomial time. We now introduce the main notation needed to formalize these ideas and start with
the following

\begin{definition}\label{def:tsupp}
If $\mathcal S$ is a set of rooted triples, we define the \emph{triple-support} as
$\tsupp_{\mathcal{S}}\coloneqq \bigcup_{xy|z \in \mathcal{S}} \{xy,xz,yz\}$.
\end{definition}

Note that since a triple $xy|z$ is defined on three distinct leaves $x,y,z$, every
$xy\in\tsupp_{\mathcal{S}}$ satisfies $x\neq y$. Hence, if all triples in $\RR$ are displayed by
$G$, then $\tsupp_{\mathcal{\RR}}$ contains precisely those elements $ab$ for which there is some
triple in $\RR$ enforcing $\lca_G(ab)$ to be well-defined in $G$. This motivates to define a
restriction $\FFRR$ of $\FF$ for a pair $(\RR, \FF)$ of triple sets as follows 
\[
\FFRR \coloneqq \{ xy|z  \in \FF \ \,:\, xy,xz,yz \in \tsupp_{\RR}\}. 
\]

Thus, the set $\FFRR$ contains all triples $xy|z \in \FF$ for which $\lca_G(xy)$, $\lca_G(xz)$, and
$\lca_G(yz)$ must be well-defined in any DAG $G$ that displays all triples in $\RR$ and, in
particular, in any DAG $G$ that agrees with $(\RR,\FF)$. Consequently, if $ab\in
\tsupp_{\FF}\setminus \tsupp_{\RR}$ and, thus, $a \neq b$, then $\lca_G(ab)$ does not need to be
well-defined in a DAG $G$ that agrees with $(\RR,\FF)$. In this case, we can apply an $ab$-extension
to $G$ to obtain $G'$ for which $|\LCA_{G'}(ab)|>1$ holds (cf.\ Observation~\ref{obs:xy-extension1}). This
observation, in turn, motivates the following definition, that we need in the upcoming results. 
\begin{definition}\label{def:FFRR_extension}
Let $(\RR,\FF)$ be a pair of triple sets. Given a DAG $G$ on $\XRF$, the \emph{$\FFRR$-extension of
$G$} is obtained from $G$ by stepwise application of $xy$-extensions for all $xy \in \tsupp_\FF
\setminus \tsupp_\RR$.
\end{definition}

The definition of the $\FFRR$-extension is similar to that of the $\FR$-extension
(cf.~Definition~\ref{def:FRextension}). In particular, suppose that $R$ and $F$ are relations on
$\pairs(\XRF)$ such that $\tsupp_\RR=\support_R$ and $\tsupp_\FF=\support_F$. Since $\tsupp_\FF$
contains no elements $xx$, the set $\tsupp_\FF\setminus\tsupp_\RR$ coincides with
$\support_F\setminus\support_R^+$. Hence, in this situation, the $\FFRR$-extension of a DAG
coincides with its $\FR$-extension.

The next result shows that, for the purpose of deciding whether a DAG exists
that agrees with $(\RR,\FF)$, it suffices to consider the restricted
forbidden set $\FFRR$ instead of $\FF$.

\begin{proposition}\label{prop:reduction_of_forbidden_triple_sets}
For every pair $(\RR,\FF)$ of triple sets and every DAG $G$ on $\XRF$, the following statements
hold. 
\begin{enumerate}
    \item If $G$ agrees with $(\RR,\FF)$, then $G$ agrees with $(\RR,\FFRR)$.
    \item If $G$ agrees with $(\RR,\FFRR)$, then the $\FFRR$-extension of $G$ agrees with $(\RR,\FF)$. 
\end{enumerate}
Thus, there exists a DAG that agrees with $(\RR,\FF)$ if and only if there exists a DAG that agrees
with $(\RR,\FFRR)$. 
\end{proposition}
\begin{proof}
Let $(\RR,\FF)$ be a pair of triple sets and $G$ be a DAG on $\XRF$. Clearly, if $G$ agrees with
$(\RR, \FF)$, then $G$ agrees with $(\RR, \FFRR)$, since $\FFRR \subseteq \FF$, which proves
Statement (1).
 
Suppose now that $G$ agrees with $(\RR,\FFRR)$, and let $H$ be the $\FFRR$-extension of $G$. This
extension $H$ is obtained by applying $pq$-extensions for all pairs
$pq\in\tsupp_\FF\setminus\tsupp_\RR$. Hence, for every $ab\in\tsupp_\RR$, no $ab$-extension is
applied. By Observation~\ref{obs:xy-extension1}, each such extension preserves the set of LCAs of
every pair in $\tsupp_\RR$ and preserves all ancestor relations among the vertices of the previous
DAG. Consequently, after all extensions have been applied, we have $ \LCA_H(ab)=\LCA_G(ab)$ for all
$ab\in\tsupp_\RR$, and the ancestor relations among the vertices of $G$ are the same in $G$ and in
$H$. Hence, every triple in $\RR$ remains displayed by $H$. Moreover, since every triple in $\FFRR$
involves only pairs from $\tsupp_\RR$, every triple in $\FFRR$ remains not displayed by $H$. It
remains to consider a triple $xy|z\in\FF\setminus\FFRR$. By definition of $\FFRR$, at least one of
the three pairs $xy,xz,yz$ lies in $\tsupp_\FF\setminus\tsupp_\RR$. For this pair, say $pq$, a
$pq$-extension was applied to obtain $H$. Hence, by Observation~\ref{obs:xy-extension1},
$\lca_H(pq)$ is not well-defined. Consequently, at least one of $\lca_H(xy)$, $\lca_H(xz)$, and
$\lca_H(yz)$ is not well-defined, and so $H$ does not display $xy|z$. Thus, $H$ agrees with
$(\RR,\FF)$. This proves Statement~(2).

Statement~(1) and (2) together imply that there exists a DAG that agrees with $(\RR,\FF)$ if and
only if there exists a DAG that agrees with $(\RR,\FFRR)$.
\end{proof}

\begin{figure}
    \centering
    \includegraphics[width=0.8\textwidth]{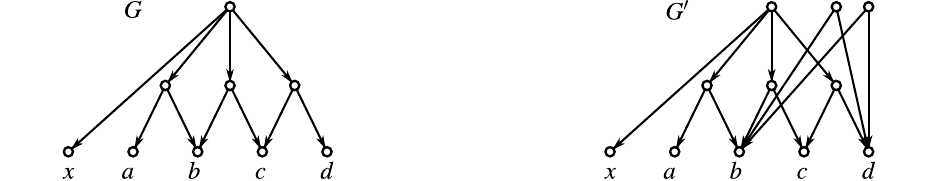}
    \caption{Consider the pair $(\RR,\FF)$ of triple sets with 
    $\RR=\{ab|x, bc|x, cd|a\}$ and $\FF=\{ac|x, ab|d\}$.
    Then $\FFRR = \{ac|x\}$ and the DAG $G$ agrees with $(\RR,\FFRR)$ but does not agree with $(\RR,\FF)$, since $ab|d$ is displayed. 
    In accordance with Proposition~\ref{prop:reduction_of_forbidden_triple_sets}, the $\FFRR$-extension $G'$ of $G$ agrees with $(\RR,\FF)$. 
    }
    \label{fig:FFRR_extension}
\end{figure}

In Figure~\ref{fig:FFRR_extension}, we show an example of a DAG $G$ that agrees with some
$(\RR,\FFRR)$ but not with $(\RR,\FF)$, together with the $\FFRR$-extension of $G$, which agrees
with $(\RR,\FF)$. We will use this example throughout Section~\ref{subs:TC-F}. Recall that our main
aim is to construct a DAG, respectively, network on $\XRF$ that agrees with a given pair $(\RR,\FF)$
of triple sets. We will solve this problem by designing a relation $Q_K$, based on the properties of
$\RR$ and $\FF$, such that the $\FFRR$-extension of the canonical DAG $\GG_{Q_K}$ agrees with
$(\RR,\FF)$ if and only if there exists some DAG $G$ that agrees with $(\RR,\FF)$. Since we want to
obtain DAGs and networks defined on the common leaf set $\XRF$, all LCA-constraints associated with
$(\RR,\FF)$ are regarded as relations on a common ground set, which we fix for the remainder of this
section as follows.

\begin{remark}
From now on, unless explicitly stated otherwise, all relations considered in this section are
relations on $\pairs(\XRF)$. 
\end{remark} 

To recall, we have defined
$
    R_\RR \coloneqq \{(xy,xz),(xy,yz) \,:\, xy|z \in \RR\}.
$
Hence, for $\RR=\{xy|z\}$, we have
\[
    R_{\{xy|z\}}=\{(xy,xz),(xy,yz)\}.
\]
We now consider the following extended relation associated with $xy|z$:
\[
    \Rext_{\{xy|z\}}
    \coloneqq
    R_{\{xy|z\}} \cup \{(xz,yz),(yz,xz)\} = 
    \{(xy,xz),(xy,yz),(xz,yz),(yz,xz)\}. 
\]

The relation $\Rext_{\{xy|z\}}$ encodes precisely the LCA-constraints imposed by the triple $xy|z$.
Indeed, if a DAG $G$ displays $xy|z$, then $ \lca_G(xy)\prec_G \lca_G(xz) =\lca_G(yz)$. The strict
inequalities $\lca_G(xy)\prec_G\lca_G(xz)$ and $\lca_G(xy)\prec_G\lca_G(yz)$ are represented by the
pairs $(xy,xz)$ and $(xy,yz)$, respectively, while the equality $\lca_G(xz)=\lca_G(yz)$ is
represented by the two pairs $(xz,yz)$ and $(yz,xz)$ in $\Rext_{\{xy|z\}}$.

We define now a particular restriction of a given relation $R$ on $\pairs(X)$ to three elements of
$X$. 

\begin{definition} \label{def:restrRxyz}
    Given a relation $R$ on $\pairs(X)$ and three distinct elements $x,y,z$ $\in X$, we define the
    \emph{restriction of $R$ to $x,y,z$} as 
    \[
    R\restr  \coloneqq 
    R \cap \{(xy,xz),(xy,yz),(xz,xy),(xz,yz),(yz,xy),(yz,xz)\}. 
    \]
\end{definition}

Since $x,y,z$ are pairwise distinct, the restriction $R\restr$ only records comparisons between the
three distinct pairs $xy$, $xz$, and $yz$. In particular, elements involving repeated pairs, such as
$(xy,xy)$, as well as elements involving repeated leaves, such as $(xx,xy)$ or $(zz,zx)$, do not
occur in $R\restr$. We next relate this restricted relation to the relation associated with a single
triple.

\begin{lemma}\label{lem:displays-iff-restr-equal}
Let $xy|z$ be a triple and $G$ be a DAG with leaves $x,y$, and $z$. Then, $G$ displays $xy|z$ if and
only if $\rel_G\relrestr=\Rext_{\{x y |z\}}$.  
\end{lemma}
\begin{proof}
Let $G$ be a DAG on $X$ and $xy|z$ a triple with $x,y,z \in X$. Since, by definition, $x,y,z$ are
pairwise distinct and $\rel_G$ is a relation on $\pairs(X)$, it follows that $\rel_G\relrestr$ is
well-defined. Suppose $G$ displays $xy|z$, then $\lca_G(xy)$, $\lca_G(xz)$, and $\lca_G(yz)$ are
well-defined and $\lca_G(xy) \prec_G \lca_G(xz) = \lca_G(yz)$ holds. Thus,
$(xy,xz),(xy,yz),(xz,yz),(yz,xz) \in \rel_G$ and $(xz,xy),(yz,xy) \notin \rel_G$ holds. This implies
$\rel_G\relrestr = \{(xy,xz),(xy,yz),(xz,yz),(yz,xz)\} = \Rext_{\{x y|z\}}$. Now suppose that
$\rel_G\relrestr=\Rext_{\{x y |z\}}$. Since $(xy,xz),(xy,yz) \in \rel_G$, we conclude that
$\lca_G(xy)$, $\lca_G(xz),$ and $\lca_G(yz)$ are well-defined. Moreover, $(xy,xz)\in \rel_G$ implies
$\lca_G(xy)\preceq_G\lca_G(xz)$ and $(xz,xy) \notin \rel_G$ implies $\lca_G(xz)\npreceq_G\lca_G(xy)$
and, therefore, $\lca_G(xy) \prec_G \lca_G(xz)$. Moreover, $(xz,yz),(yz,xz)\in \rel_G$ imply
$\lca_G(xz) = \lca_G(yz)$. In summary, $\lca_G(xy) \prec_G \lca_G(xz) = \lca_G(yz)$ holds. Hence,
$xy|z$ is displayed by $G$.
\end{proof}

In what follows, we consider pairs $(\RR, \FF)$ of triple sets and want to determine if there exists
a DAG that agrees with $(\RR, \FF)$. To this end, we will saturate $R_\RR$ using the following
procedure.

\begin{definition}\label{def:method}
For a pair $(\RR, \FF)$ of triple sets, we define the following procedure:

  \hspace{2cm}
  \begin{minipage}{.7\linewidth}
    \begin{algorithm}[H]
    \small
    \caption*{\textsc{Saturate~$R_{\RR}$~with~$\FF$} (in short \textsc{Sat($R_{\RR},\FF$)})}
  \label{alg:def-method}
        \begin{algorithmic}[0]   \small
         \State $Q_0\coloneqq\cl(R_\RR)$ 
        \State $\FF_0\coloneqq\left\{xy|z \in \FF \,:\, Q_{0}\restr=\Rext_{\{xy|z\}}\right\}$
         \State $i\coloneqq 0$
            \While{$\FF_i\neq\emptyset$}
                \State Let $xy|z$ be an arbitrary element in $\FF_i$ 
                \State $Q_{i+1}\coloneqq \cl\left(Q_i \cup\{(xz, xy),(yz, xy)\}\right)$
                \State $\FF_{i+1}\coloneqq\left\{xy|z \in \FF \,:\, Q_{i+1}\restr=\Rext_{\{xy|z\}}\right\}$
                \State $i\coloneqq i+1 $
             \EndWhile  
          \State \Return $Q_i$
        \end{algorithmic}
        \end{algorithm}     
   \end{minipage} \smallskip
\end{definition} 

To provide some intuition behind the greedy-like method \textsc{Sat($R_{\RR},\FF$)}, let us consider
$Q_0 = \cl(R_\RR)$. Clearly, if $\FF_0 = \emptyset$, then \textsc{Sat($R_{\RR},\FF$)} stops and we
still have $Q_0 = \cl(R_\RR)$. Suppose now that $\FF_0\neq \emptyset$. In this case, $Q_1$ is
obtained from $Q_0$ by adding the two elements $(xz,xy)$ and $(yz,xy)$ for some triple $xy|z\in \FF$
with $Q_0\restr = \Rext_{\{xy|z\}}$. The purpose of adding these pairs is to ensure that $xy|z$ is
not displayed by any DAG realizing the updated relation $Q_1$. Indeed, if $G$ realizes $Q_1$, then
$(xz,xy),(yz,xy)\in Q_1\subseteq \rel_G$ (cf.\ \cite[Lem~7]{LAMSH:25}). Consequently,
$\rel_G\relrestr\neq \Rext_{\{xy|z\}}$ and, hence, $G$ does not display $xy|z$ by
Lemma~\ref{lem:displays-iff-restr-equal}. This is exactly what is required, since $xy|z\in \FF$ is a
forbidden triple. After adding the two elements to $Q_0$, we compute the closure in order to obtain
all further LCA-constraints implied by the constraints in $Q_0\cup\{(xz,xy),(yz,xy)\}$ that must be
respected by any DAG realizing $Q_0\cup\{(xz,xy),(yz,xy)\}$.

As we shall see, this eventually yields a final relation $Q_K$ that is returned by
\textsc{Sat($R_{\RR},\FF$)} and for which the canonical DAG $\GG_{Q_K}$ displays all triples in $\RR$
precisely if there is some DAG that agrees with $(\RR,\FF)$, see Theorem~\ref{thm:char-RF-display}.
In particular, this in turn can be used to prove that the Problem \PROBLEM{TC-F} can be solved in
polynomial time, see Theorem~\ref{thm:TCF-problem}. The latter ideas are also illustrated in
Figure~\ref{fig:saturate}. 

Note that the constructed relations $Q_0, Q_1, Q_2, ...$ depend on the chosen arbitrary element of
$\FF_i$. Nevertheless, the following results hold for each order in which those elements are chosen.
Indeed, we now show that \textsc{Sat($R_{\RR},\FF$)} always terminates and provide some additional
properties of the constructed relations $Q_i$.

\begin{figure}[t]
    \centering
    \includegraphics[width=0.8\textwidth]{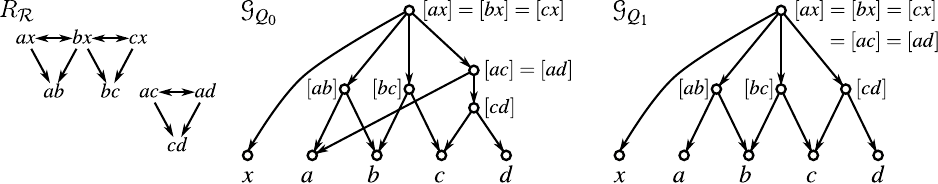}
    \caption{Consider the pair $(\RR,\FF)$ of triple sets with $\RR=\{ab|x, bc|x, cd|a\}$ and
    $\FF=\{ac|x, ab|d\}$. Let $Q_0=\cl(R_\RR)$ whose canonical DAG $\GG_{Q_0}$ is shown in the
    middle. Here, $Q_0|^{\neq}_{acx} = \Rext_{ac|x} = \{(ac,ax), (ac,cx), (ax,cx), (cx,ax)\}$ and,
    in particular, $\GG_{Q_0}$ displays the triple $ac|x\in \FF$. In this case, we apply a
    saturation of $Q_0$ by adding first the elements $(ax,ac)$ and $(cx,ac)$ and then computing the
    closure which results in $Q_{1}\coloneqq \cl\left(Q_0 \cup\{(ax,ac),(cx,ac)\}\right)$. The
    canonical DAG $\GG_{Q_1}$ is shown to the right. Here, $\GG_{Q_1}$ agrees with $(\RR,\FFRR)$,
    cf.\ Figure~\ref{fig:FFRR_extension} where $G = \GG_{Q_1}$. Since $\FF_1 = \emptyset$,
    \textsc{Sat($R_{\RR},\FF$)} terminates after one iteration of the \textbf{while}-loop.}
    \label{fig:saturate}
\end{figure}

\begin{lemma}\label{lem:Qi_properties_realizable}
For every pair $(\RR, \FF)$ of triple sets, \textsc{Sat($R_{\RR},\FF$)} terminates. Let
$Q_0,\dots,Q_K$ be the constructed relations during some run of \textsc{Sat($R_{\RR},\FF$)}. Then,
it holds that $Q_i = \tc(Q_i)=\cl(Q_i)$ for all $i\in \{0,\dots, K\}$ and $Q_i\subseteq Q_{i+1}$ for
all $i\in \{0,\dots, K-1\}$. Moreover, $Q_i$ is a realizable relation for all $i \in \{0, \dots,
K\}$.
\end{lemma}
\begin{proof} 
Let $(\RR, \FF)$ be a pair of triple sets. We first argue that \textsc{Sat($R_{\RR},\FF$)}
terminates. By construction and the monotonicity of the closure operator $\cl$, it holds that $Q_i
\subseteq Q_{i+1}$ in each step of \textsc{Sat($R_{\RR},\FF$)}. In particular, if $xy|z \in \FF_{i}$
is the currently considered triple used to compute $Q_{i+1}= \cl\left(Q_i \cup\{(xz, xy),(yz,
xy)\}\right)$, then $(xz,xy),(yz,xy)\in Q_{i+1}$ which together with $(xz,xy),(yz,xy)\notin
\Rext_{\{xy|z\}}$ implies that $Q_j\restr \neq \Rext_{\{xy|z\}}$ for all $j>i$. This, in turn,
implies that $xy|z \notin \FF_j$ for all $j > i$. Since $\FF$ is finite, and $\FF_i \subseteq \FF$
for all $i$, \textsc{Sat($R_{\RR},\FF$)} terminates after at most $|\FF|$ iterations of the
\textbf{while}-loop of \textsc{Sat}($R_{\RR},\FF$).

We continue with showing that $Q_i \subseteq Q_{i+1}$ for each $0 \leq i < K$ and that $Q_i=\tc(Q_i)
=\cl(Q_i)$ for each $0 \leq i \leq K$. For $0 \leq i < K$, we have, by monotonicity of $\cl$ and
definition of $Q_{i+1}$, that $Q_i \subseteq \cl(Q_i \cup \{(xz,xy),(yz,xy)\}) = Q_{i+1}$. Moreover,
idempotency of $\cl$ implies for $Q_0 = \cl(R_\RR)$ that $\cl(Q_0)=\cl(\cl(R_\RR)) = \cl(R_\RR) =
Q_0$ and for $Q_i = \cl(Q_{i-1} \cup \{(xz,xy),(yz,xy)\})$ that $\cl(Q_i)= \cl(\cl(Q_{i-1} \cup
\{(xz,xy),(yz,xy)\}))=\cl(Q_{i-1} \cup \{(xz,xy),(yz,xy)\})=Q_i$, $1 \leq i \leq K$. This together
with extensivity of $\tc$ and transitivity of $\cl(Q_i)$ implies that $Q_i\subseteq \tc(Q_i)=
\tc(\cl(Q_i))=\cl(Q_i)=Q_i$ and, therefore, $Q_i = \tc(Q_i)= \cl(Q_i)$ for all $0 \leq i \leq K$.

It remains to show that $Q_i$ is realizable for all $0 \leq i \leq K$. We proceed by induction on
$i$ and start with $i=0$. By definition, $Q_0 = \cl(R_\RR)$ and, for every triple $xy|z \in \RR$,
the leaves $x,y,z$ are pairwise distinct. Thus, for all $a,b,x \in \XRF$ with $ab \neq xx$, it holds
that $(ab,xx) \notin R_\RR$. By definition, $R_\RR$ satisfies \axiom{X1}. This together with
\cite[Prop~28]{LAMSH:25} implies that $\cl(R_\RR)$ and, therefore, $Q_0 = \cl(R_\RR)$ is realizable.
Now assume that there exist some $0 \leq k < K$ such that $Q_k$ is realizable. Consider
$Q_{k+1}\coloneqq \cl\left(Q_k \cup\{(xz, xy),(yz, xy)\}\right)$, which exists since $k < K$. Since
$Q_k$ is realizable, Theorem~\ref{thm:char} implies that $Q_k$ satisfies \axiom{X1}. Hence, the set
$Q_k \cup\{(xz, xy),(yz, xy)\}$ does, by construction, satisfy $\axiom{X1}$. This together with
\cite[Prop~28]{LAMSH:25} implies that $\cl(Q_k \cup\{(xz, xy),(yz, xy)\}) = Q_{k+1}$ is realizable,
which completes this proof.  
\end{proof}

\begin{remark}\label{rem:QK}
	By Lemma~\ref{lem:Qi_properties_realizable}, \textsc{Sat($R_{\RR},\FF$)} always terminates and,
	thus, there is one last call of the \textbf{while}-loop in which it is verified that
	$\FF_i=\emptyset$ and the algorithm stops. Hence, there are $K+1\geq 1$ calls of the
	\textbf{while}-loop condition $\FF_i\neq \emptyset$ and we assume, from here on, that
	$Q_0,\dots,Q_K$ are the constructed relations during some run of \textsc{Sat($R_{\RR},\FF$)}.
	Moreover, we call $Q_K$ a \emph{final relation} of \textsc{Sat($R_{\RR},\FF$)}. 
\end{remark}

Recall that the choice of the triple $xy|z\in\FF_i$ in each iteration may affect the final relation
$Q_K$ returned by \textsc{Sat($R_{\RR},\FF$)}. In particular, for different runs of
\textsc{Sat($R_{\RR},\FF$)}, the number of iterations may differ, and it remains open whether the
final relation $Q_K$ of one run always coincides with the final relation $Q_L$ of another run.

The following lemma establishes the soundness of the saturation procedure: if a DAG $G$ agrees with
$(\RR,\FF)$, then every relation $Q_i$ constructed during some run of \textsc{Sat($R_{\RR},\FF$)} is
already contained in the LCA-relation $\rel_G$ of $G$.

\begin{lemma}\label{lem:Qi-subset-relG}
 Assume that $G$ is a DAG that agrees with the pair $(\RR,\FF)$ of triple sets. Then, it holds that
 $Q_i\subseteq\rel_G$ for all relations $Q_0,\dots,Q_K$ constructed during some run of
 \textsc{Sat($R_{\RR},\FF$)}. 
\end{lemma}
\begin{proof}
	Let $G$ be a DAG that agrees with the pair $(\RR,\FF)$ of triple sets. Let $Q_0,\dots,Q_K$ be the
	relations constructed during some run of \textsc{Sat($R_{\RR},\FF$)}. We proceed by induction on
	$i\in \{0,\dots,K\}$. For the base case, let $i=0$. By definition, $Q_0\coloneqq\cl(R_\RR)$. Since
	$G$ displays every triple in $\RR$, Theorem~\ref{thm:TC-char} implies that $G$ realizes $R_\RR$.
	By \cite[Lem~7]{LAMSH:25}, $R_\RR\subseteq \rel_G$ follows. Moreover, $\cl(\rel_G)=\rel_G$ holds
	by \cite[Prop~21]{LAMSH:25}. The latter two arguments together with monotonicity of $\cl$ imply
	that $Q_0 = \cl(R_\RR)\subseteq \cl(\rel_G)=\rel_G$. 

	Assume now that for some $i\in \{0,\dots,K-1\}$ it holds that $Q_i\subseteq\rel_G$. To recall,
	\[R\restr = R \cap \{(xy,xz),(xy,yz),(xz,xy),(xz,yz),(yz,xy),(yz,xz)\}\text{ for any relation }R\]
	while $\Rext_{\{xy|z\}}=\{(xy,xz),(xy,yz),(xz,yz),(yz,xz)\}$. Consider now $Q_{i+1}= \cl\left(Q_i
	\cup\{(xz, xy),(yz, xy)\}\right)$. Since $i<K$, the relation $Q_{i+1}$ exists and has, in
	particular, been constructed using some triple $xy|z\in\FF$ for which
	$Q_{i}\restr=\Rext_{\{xy|z\}}$. Since $Q_{i}\restr=\Rext_{\{xy|z\}}$, it holds that
	$\Rext_{\{xy|z\}}\subseteq Q_{i}\subseteq\rel_G $. In particular, the three LCAs $\lca_G(xy)$, 
    $\lca_G(xz)$, and $\lca_G(yz)$ are well-defined and $\lca_G(xy) \preceq_G \lca_G(xz) = \lca_G(yz)$
	must hold. However, as $G$ agrees with $(\RR,\FF)$, the triple $xy|z\in\FF$ is not displayed by
	$G$. Hence, we obtain $\lca_G(xy) = \lca_G(xz) = \lca_G(yz)$. Thus, $(xz,xy),(yz,xy) \in \rel_G$
	holds. This together with $Q_i \subseteq \rel_G$ implies that $Q_i \cup \{(xz,xy),(yz,xy)\}
	\subseteq \rel_G$. By \cite[Prop~21]{LAMSH:25}, $\cl(\rel_G)=\rel_G$ holds. The latter two
	arguments together with monotonicity of $\cl$ imply that
	$Q_{i+1}=\cl(Q_i\cup\{(xz,xy),(yz,xy)\})\subseteq\cl(\rel_G)=\rel_G$. In summary, we verified that
	$Q_{i+1}\subseteq\rel_G$. The full statement follows by induction.
\end{proof}

We next record a non-trivial consequence of Lemma~\ref{lem:Qi_properties_realizable} and
\ref{lem:Qi-subset-relG}. In particular, 
if a DAG agreeing with $(\RR,\FF)$ exists, 
the canonical DAG associated with $Q_i$ still realizes the relation $R_{\RR}$ induced by the
required triple set $\RR$.

\begin{lemma}\label{lem:Si-auxillary}
Let $(\RR,\FF)$ be a pair of triple sets and $Q_0,\dots,Q_K$ be the relations constructed during
some run of \textsc{Sat($R_{\RR},\FF$)}. If there is a DAG that agrees with $(\RR,\FF)$, then the
canonical DAG $\GG_{Q_i}$ realizes $R_\RR$ for all $i\in \{0,\dots,K\}$. 
\end{lemma} 
\begin{proof}
Let $G$ be a DAG that agrees with the pair $(\RR,\FF)$ of triple sets. Let $Q_0,\dots,Q_K$ be the
relations constructed during some run of \textsc{Sat($R_{\RR},\FF$)}.

We show that $\GG_{Q_i}$ realizes $R_\RR$ for all $i\in \{0,\dots,K\}$. To this
end, let $i\in \{0,\dots,K\}$ be chosen arbitrarily. Observe first that transitivity of $\cl$
implies that $\tc(R_\RR)\subseteq\cl(R_\RR)=Q_0$. Moreover, by
Lemma~\ref{lem:Qi_properties_realizable}, $Q_0 \subseteq Q_i=\tc(Q_i)$ holds and, hence, $\tc(R_\RR)
\subseteq Q_i = \tc(Q_i)$. To verify that $\GG_{Q_i}$ realizes $R_\RR$, we show that \axiom{I1} and
\axiom{I2} hold for $\GG_{Q_i}$ and $R_\RR$. We start with \axiom{I2}. Hence, suppose $(ab,cd) \in
R_\RR\subseteq Q_i$ and $(cd,ab) \in \tc(R_\RR) \subseteq \tc(Q_i)$.
Lemma~\ref{lem:Qi_properties_realizable} implies that $Q_i$ is realizable and by
Theorem~\ref{thm:char}, the canonical DAG $\GG_{Q_i}$ realizes $Q_i$. Hence, $\lca_{\GG_{Q_i}}(ab) =
\lca_{\GG_{Q_i}}(cd)$ must hold and, in particular, \axiom{I2} is satisfied for $\GG_{Q_i}$ and
$R_\RR$. For \axiom{I1}, let $(ab,cd)\in R_\RR\subseteq Q_i$ and $(cd,ab)\notin\tc(R_\RR)$. Assume,
for contradiction, that $(cd,ab)\in \tc(Q_i)$. In this case, Lemma~\ref{lem:Qi-subset-relG} implies
that $\tc(Q_i) = Q_{i}\subseteq\rel_G$ and, therefore, $(ab,cd),(cd,ab)\in\rel_G$. In other words,
$\lca_G(ab)=\lca_G(cd)$, which implies that $G$ does not realize $R_\RR$. However, $G$ displays all
triples in $\RR$ and, by Theorem~\ref{thm:TC-char}, realizes $R_\RR$; a contradiction. Hence,
$(cd,ab)\notin \tc(Q_i)$. Since $Q_i$ is realized by $\GG_{Q_i}$ and since $(ab,cd)\in Q_i$ and
$(cd,ab)\notin \tc(Q_i)$, it holds that $\lca_{\GG_{Q_i}}(ab)\prec_{\GG_{Q_i}}
\lca_{\GG_{Q_i}}(cd)$. Therefore, \axiom{I1} is satisfied for $\GG_{Q_i}$ and $R_\RR$. In summary,
$\GG_{Q_i}$ realizes $R_\RR$.
\end{proof}

We now combine the preceding ingredients to characterize when a pair $(\RR,\FF)$ of required and
forbidden triple sets agrees with some DAG. The theorem shows, in particular, that the final relation
$Q_K$ produced by the saturation procedure is sufficient to decide whether a DAG agreeing with
$(\RR,\FF)$ exists.

\begin{theorem}\label{thm:char-RF-display} 
Let $(\RR,\FF)$ be a pair of triple sets and $Q_K$ be the final relation constructed during some run
of \textsc{Sat($R_{\RR},\FF$)}. Then, the following statements are equivalent. 
\begin{enumerate}   
    \item There exists a DAG on $\XRF$ that agrees with $(\RR, \FF)$. 
    \item The canonical DAG $\GG_{Q_K}$ displays all triples in $\RR$. 
    \item The canonical DAG $\GG_{Q_K}$ realizes $R_\RR$ and $(xz,xy),(yz,xy) \notin \tc(R_\RR)$ for
          all $xy|z \in \RR$.
    \item $(xz,xy),(yz,xy) \notin Q_K$ for all $xy|z \in \RR$.
\end{enumerate}
\end{theorem}
\begin{proof}
Let $(\RR,\FF)$ be a pair of triple sets, $Q_K$ be the final relation constructed during some run of
\textsc{Sat($R_{\RR},\FF$)} and put $\GG \coloneqq \GG_{Q_K}$. We prove this theorem by verifying
the following implications: $(1) \Rightarrow (4) \Rightarrow (3) \Rightarrow (2) \Rightarrow (1)$.

Suppose that Statement~(1) holds. Hence, there exists a DAG $G$ on $\XRF$ that agrees with
$(\RR,\FF)$. 
By Lemma~\ref{lem:Qi-subset-relG}, we have $Q_K\subseteq \rel_G$.
Assume, for contradiction, that $(xz,xy) \in Q_K$ for some $xy|z \in \RR$.
By construction and Lemma~\ref{lem:Qi_properties_realizable}, it holds that $(xy,xz) \in R_\RR \subseteq \cl(R_\RR) =
Q_0 \subseteq Q_K $.
The latter two arguments imply that 
both $(xz,xy)$ and $(xy,xz)$ belong to $\rel_G$. It follows that $ \lca_G(xy)\preceq_G\lca_G(xz)$
and $ \lca_G(xz)\preceq_G\lca_G(xy)$. Therefore, $ \lca_G(xy)=\lca_G(xz)$, which contradicts the fact
that $G$ displays $xy|z$. Hence, $(xz,xy)\notin Q_K$. 
By an analogous argument, $(yz,xy) \notin Q_K$. Thus, Statement~(4) holds.

Now suppose Statement~(4) holds and we show that this implies Statement~(3). 
By assumption, $(xz,xy),(yz,xy) \notin Q_K$ for all $xy|z \in \RR$. We first argue that
$(xz,xy),(yz,xy) \notin \tc(R_\RR)$ for all $xy|z \in \RR$. To this end, observe that $\tc(R_\RR)
\subseteq \cl(R_\RR) = Q_0 \subseteq Q_K$ holds by transitivity of $\cl$ and by
Lemma~\ref{lem:Qi_properties_realizable}. This together with $(xz,xy),(yz,xy) \notin Q_K$ implies
$(xz,xy),(yz,xy) \notin \tc(R_\RR)$ for all $xy|z \in \RR$. We continue by showing that $\GG$
realizes $R_\RR$. By Lemma~\ref{lem:Qi_properties_realizable}, $Q_K$ is realizable and by
Theorem~\ref{thm:char}, $\GG$ realizes $Q_K$. We now verify \axiom{I1} and \axiom{I2} for $\GG$ and
$R_\RR$. Let $(ab,cd) \in R_\RR$. Since $R_\RR = \{(xy,xz),(xy,yz) \,:\, xy|z \in \RR\}$, there
exists some $xy|z \in \RR$ such that $(ab,cd) \in \{(xy,xz),(xy,yz)\}$. As argued earlier,
$(xz,xy),(yz,xy) \notin \tc(R_\RR)$ and, thus, $(cd,ab) \notin \tc(R_\RR)$ holds. Hence, \axiom{I2}
is vacuously true for $\GG$ and $R_\RR$. Moreover, by extensivity of $\cl$ and
Lemma~\ref{lem:Qi_properties_realizable}, $R_\RR \subseteq \cl(R_\RR) = Q_0 \subseteq Q_K =
\tc(Q_K)$ and, thus, $(ab,cd)\in R_\RR$ implies that $(ab,cd) \in Q_K$. Since $(cd,ab) \in
\{(xz,xy),(yz,xy) \,:\, xy|z \in \RR\}$, it follows, from the assumption $(xz,xy),(yz,xy)\notin Q_K$
for all $xy|z\in\RR$, that $(cd,ab)\notin Q_K$. Moreover, by
Lemma~\ref{lem:Qi_properties_realizable}, $Q_K=\tc(Q_K)$ and, hence, $(cd,ab)\notin\tc(Q_K)$. Since
$\GG$ realizes $Q_K$, it follows from \axiom{I1} that $ \lca_\GG(ab)\prec_\GG\lca_\GG(cd)$.
Therefore, \axiom{I1} is satisfied for $\GG$ and $R_\RR$. In summary, $\GG$ realizes $R_\RR$. Hence,
Statement~(3) holds. By Theorem~\ref{thm:TC-char}, Statement~(3) implies Statement~(2).

It remains to show that Statement~(2) implies (1). Hence, assume that $\GG$ displays all triples in
$\RR$. We first argue that $\GG$ does not display any triple in $\FFRR$. Let $xy|z \in \FFRR$.
Suppose, for contradiction, that $\GG$ displays $xy|z$. By Lemma~\ref{lem:displays-iff-restr-equal},
$\rel_\GG\relrestr=\Rext_{\{x y |z\}} = \{(xy,xz),(xy,yz),(xz,yz),(yz,xz)\}$ holds. By
Lemma~\ref{lem:Qi_properties_realizable}, $Q_K$ is realizable and, in addition, $\cl(Q_K) = Q_K$.
This together with \cite[Thm~47]{LAMSH:25} implies $Q_K = \rel_\GG \cap (\support_{Q_K}^+ \times
\support_{Q_K}^+)$. Since $xy|z \in \FFRR$, we conclude that $xy,xz,yz \in \tsupp_{\RR}$ and,
therefore, $xy,xz,yz \in \support_{R_\RR}$. Since $R_\RR \subseteq \cl(R_\RR) = Q_0 \subseteq Q_K$
by extensivity of $\cl$ and Lemma~\ref{lem:Qi_properties_realizable}, we obtain $xy,xz,yz \in
\support_{Q_K}$. With $Q_K = \rel_\GG \cap (\support_{Q_K}^+ \times \support_{Q_K}^+)$, we
conclude that $Q_K\restr = \rel_\GG\relrestr$. This together with $\rel_\GG\relrestr=\Rext_{\{x y
|z\}}$ implies that $Q_K\restr = \Rext_{\{x y |z\}}$. But then, $xy|z \in \FF_K \neq \emptyset$
where $\FF_K$ is the triple set as defined in \textsc{Sat($R_{\RR},\FF$)}; a contradiction to $Q_K$
being the final relation during this particular run of \textsc{Sat($R_{\RR},\FF$)}. Hence, $\GG$
does not display $xy|z$ and it follows that $\GG$ agrees with $(\RR,\FFRR)$. 
Since $Q_K$ is a relation on $\pairs(\XRF)$, its canonical DAG $\GG$ has leaf set
$\XRF$ and, by construction, the $\FFRR$-extension of $\GG$ has leaf set $\XRF$. 
By Proposition~\ref{prop:reduction_of_forbidden_triple_sets}, the $\FFRR$-extension of $\GG$ agrees
with $(\RR,\FF)$. Hence, Statement~(1) holds. 
\end{proof}

Theorem~\ref{thm:char-RF-display} provides characterizations for when there exists a DAG that agrees
with a pair of required and forbidden triple sets. In the following theorem, we give a specific
phylogenetic DAG and network that agrees with $(\RR,\FF)$, whenever such a DAG, respectively,
network exists.

\begin{theorem}\label{thm:char-RF-construction} 
Let $(\RR,\FF)$ be a pair of triple sets. Then, the following statements are equivalent.
\begin{enumerate}
    \item There is a DAG on $\XRF$ that agrees with $(\RR,\FF)$. 
    \item There is a phylogenetic network on $\XRF$ that agrees with $(\RR,\FF)$. 
\end{enumerate}
In particular, suppose that there exists a DAG that agrees with $(\RR,\FF)$ and let $Q_K$ be the
final relation constructed during some run of \textsc{Sat($R_{\RR},\FF$)}. Then, the
$\FFRR$-extension $G$ of $\GG_{Q_K}$ as well as the network $N$ obtained from $G$ by
Lemma~\ref{lem:DAG2Network} are phylogenetic and agree with $(\RR,\FF)$. 
\end{theorem}
\begin{proof}
	Let $(\RR,\FF)$ be a pair of triple sets, $Q_K$ be the final relation constructed during some run
	of \textsc{Sat($R_{\RR},\FF$)} and put $\GG \coloneqq \GG_{Q_K}$. Suppose that Statement~(1)
	holds. Thus, there exists a DAG $G$ on $\XRF$ that agrees with $(\RR,\FF)$. We first argue that
	there exists a phylogenetic DAG $H$ that agrees with $(\RR,\FF)$ and second show that the network
	obtained from $H$ by Lemma~\ref{lem:DAG2Network} agrees with $(\RR,\FF)$ and is phylogenetic. If
	$G$ is phylogenetic, we simply put $\widetilde G \coloneqq G$. Otherwise, there exists a vertex $v
	\in V(G)$ such that $\outdeg_G(v) = 1$ and $\indeg_G(v) \leq 1$. It is an easy task to verify that
	$v$ cannot be the least common ancestor of any two leaves. We now suppress $v$, i.e., we remove
	$v$ and its incident arcs and afterwards add the arc $(p,c)$ with $p$ and $c$ being the unique
	parent (if it exists) and child of $v$, respectively. This results in a directed graph $G'$. Note
	that if $v$ is a root in $G$ then its child $c$ becomes a new root in $G'$. By
	\cite[Thm~5.5]{HL:24}, $G'$ is a DAG that satisfies $\LCA_G(xy) = \LCA_{G'}(xy)$ for all $x,y\in
	\XRF$ as well as $a \preceq_{G'} b$ if and only if $a \preceq_G b$ for all $a,b \in V(G')$. The
	latter, in particular, implies that $G'$ agrees with $(\RR,\FF)$. We can now repeat this procedure
	until we obtain a phylogenetic DAG $\widetilde G$ that agrees with $(\RR,\FF)$. In particular,
	$\widetilde G$ is a DAG on $\XRF$. Now, let $H$ be the $\FFRR$-extension of $\widetilde G$. By
	Observation~\ref{obs:xy-extension1}, $H$ remains a phylogenetic DAG on $\XRF$ that agrees with
	$(\RR,\FF)$ and $|\LCA_H(xy)|>1$ for all $xy \in \tsupp_{\FF} \setminus
	\tsupp_{\RR}$. Now, let $N$ be the network obtained from $H$ as specified in
	Lemma~\ref{lem:DAG2Network}. Thus, the $\preceq_{H}$-order among any of the vertices in $H$ is
	preserved in $N$ and $\LCA_H(xy) = \LCA_N(xy)$ for all $xy \in \pairs(\XRF)$ with $\LCA_H(xy) \neq
	\emptyset$. Since $H$ agrees with $(\RR,\FF)$ and is an $\FFRR$-extension, $\LCA_H(xy) \neq
	\emptyset$ for all $xy \in \tsupp_{\RR} \cup \tsupp_{\FF}$. It readily follows that $N$ is a
	network that agrees with $(\RR,\FF)$. Moreover, by Lemma~\ref{lem:DAG2Network}, $N$ is
	phylogenetic network on $\XRF$. Thus, Statement~(2) holds. Moreover, Statement~(2) clearly implies
	Statement~(1).

	For the last statement, suppose there is a DAG that agrees with $(\RR,\FF)$. By
	Theorem~\ref{thm:char-RF-display}, this implies that $\GG$ displays all triples in $\RR$. In the
	final paragraph of the proof of Theorem~\ref{thm:char-RF-display}, we showed that this implies
	that the $\FFRR$-extension $G$ of $\GG$ agrees with $(\RR,\FF)$. Let $N$ be the network obtained
	from $G$ as specified in Lemma~\ref{lem:DAG2Network}. By the same arguments as in the previous
	paragraph, $N$ agrees with $(\RR,\FF)$. Finally, by
	Proposition~\ref{prop:properties_of_canonical_DAG} and since $Q_K$ is realizable by
	Lemma~\ref{lem:Qi_properties_realizable}, $\GG=\GG_{Q_K}$ is phylogenetic. Since during the
	construction of the $\FFRR$-extension $G$ of $\GG$ the in-degree of vertices is only increased and
	the out-degree of all added vertices equals two, $G$ is phylogenetic. By
	Lemma~\ref{lem:DAG2Network}, $N$ is phylogenetic.
\end{proof}

\begin{figure}
    \centering

    \includegraphics[width=0.8\textwidth]{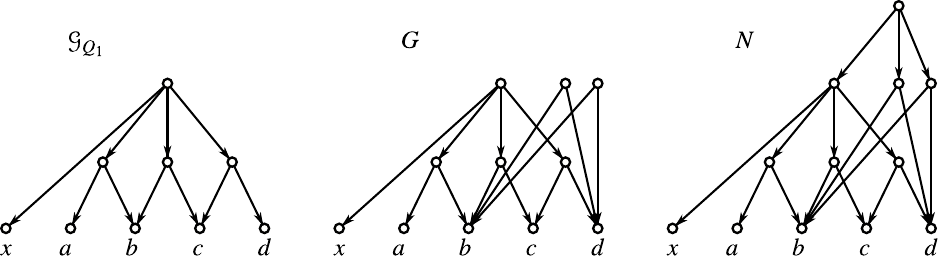}
    \caption{Consider the pair $(\RR,\FF)$ of triple sets with $\RR=\{ab|x, bc|x, cd|a\}$ and
    $\FF=\{ac|x, ab|d\}$. Let $Q_1$ be the final relation computed during the unique run of
    \textsc{Sat($R_{\RR},\FF$)}, illustrated in Figure~\ref{fig:saturate}. In this example, the canonical DAG
    $\GG_{Q_1}$ displays all triples in $\RR$, but does not agree with $(\RR,\FF)$. However, by Theorem~\ref{thm:char-RF-construction}, the
    $\FFRR$-extension $G$ of $\GG_{Q_1}$ and the network $N$ obtained from $G$ according to
    Lemma~\ref{lem:DAG2Network} agree with $(\RR,\FF)$ and are phylogenetic. 
    }
    \label{fig:TC-F_DAG_network}
\end{figure}

In Figure~\ref{fig:TC-F_DAG_network}, we show the canonical DAG $\GG_{Q_K}$, its $\FFRR$-extension
$G$, and the network $N$ obtained from $G$ by Lemma~\ref{lem:DAG2Network} for a pair $(\RR,\FF)$ of
triple sets.

\begin{figure}
\centering
\begin{minipage}{.9\textwidth}
\begin{algorithm}[H] 
    \small
  \caption{\textsc{Triple Consistency with Forbidden Triples}}
  \label{alg:TC-F-v2}
  \begin{algorithmic}[1]
  \Require  A pair $(\RR,\FF)$ of triple sets
    \Ensure  A phylogenetic DAG and a phylogenetic network on $\XRF$ that agrees with
             $(\RR,\FF)$ if one exists and, otherwise, \texttt{false} is returned
    \State $Q_K\gets \textsc{Sat($R_{\RR},\FF$)}$ 
    \State Compute the canonical DAG  $\GG_{Q_K}$  
    \State Compute the $\FFRR$-extension $G$ of $\GG_{Q_K}$ 
    \If{$G$ agrees with $(\RR,\FF)$}
        \State \Return $G$ and the network $N$ obtained from $G$ according to Lemma~\ref{lem:DAG2Network}
     \Else \ \Return \texttt{false}
    \EndIf
  \end{algorithmic}
\end{algorithm}		
\end{minipage}
\end{figure}

We argue now that the \PROBLEM{TC-F} problem can be solved in polynomial time. 

\begin{theorem}\label{thm:TCF-problem}
Let $(\RR,\FF)$ be a pair of triple sets. Then, Algorithm~\ref{alg:TC-F-v2}, applied to $(\RR,\FF)$
decides in polynomial time in $|\XRF|$ whether there exists a phylogenetic DAG, equivalently a phylogenetic 
network, on $\XRF$ that agrees with $(\RR,\FF)$ and, in the affirmative case, constructs such a
phylogenetic DAG and a corresponding phylogenetic network within the same time bound.
\end{theorem}
\begin{proof}
	Let $(\RR, \FF)$ be a pair of triple sets. We start with showing that Algorithm~\ref{alg:TC-F-v2}
	is correct. Observe that Theorem~\ref{thm:char-RF-construction} implies that there is a DAG that
	agrees with $(\RR,\FF)$ if and only if the $\FFRR$-extension $G$ of $\GG_{Q_K}$ and the network
	$N$ obtained from $G$ according to Lemma~\ref{lem:DAG2Network} agree with $(\RR, \FF)$. In
	particular, $G$ and $N$ are phylogenetic. Hence, Algorithm~\ref{alg:TC-F-v2} is correct. 

	For the runtime of Algorithm~\ref{alg:TC-F-v2}, put $X \coloneqq \XRF$. Observe first that $|\RR|,
	|\FF| \in O(|X|^3)$. Hence, $R_{\RR}$ can be constructed in polynomial time in $|X|$. In
	particular, $|R_\RR|\in O(|X|^3)$. Computation of $Q_K$ requires, by
	Lemma~\ref{lem:Qi_properties_realizable}, at most $|\FF| \in O(|X|^3)$ iterations of the
	\textbf{while}-loop of \textsc{Sat($R_{\RR},\FF$)}. In particular, the runtime of a single
	iteration of the \textbf{while}-loop of \textsc{Sat($R_{\RR},\FF$)} is dominated by the
	construction of the $\cl$. This can be done in polynomial time in $|X|$ by
	Theorem~\ref{thm:cl-polytime}. Hence, $Q_K$ can be computed in polynomial time in $|X|$. Moreover,
	the canonical DAG $\GG_{Q_K}$ can be computed in polynomial time in $|X|$, cf.\ Theorem~38 and
	Algorithm~2 in \cite{LAMSH:25}. The $\FFRR$-extension $G$ of $\GG_{Q_K}$ requires adding at most
	$|\tsupp_{\FF}\setminus\tsupp_{\RR} | \leq |\tsupp_{\FF}|$ distinct $xy$-extensions, each of which
	can be performed in constant time. Since $|\tsupp_\FF|\in O(|X|^2)$, the DAG $G$ can be
	constructed in polynomial time in $|X|$. Moreover, it is straightforward to check that the network
	$N$ obtained from $G$ by Lemma~\ref{lem:DAG2Network} is constructable in polynomial time in $|X|$.
	The final \emph{if}-condition amounts to computing LCAs in $G$ and verifying whether certain paths
	between vertices exist in $G$. Determining whether an LCA of some $xy\in\pairs(X)$ is well-defined
	in $G$ and if so, determining the vertex $\lca_G(xy)$ can be done in polynomial time in $|V(G)|$,
	see \cite[Thm~3]{Kowaluk:07}. Since each vertex of $\GG_{Q_K}$ is the unique LCA of some
	$xy\in\pairs(X)$ \cite[Prop~28]{LAMSH:25}, we have $|V(\GG_{Q_K})|\leq|\pairs(X)|$ and, as argued
	above, $|V(G)\setminus V(\GG_{Q_K})|\leq 2|\tsupp_\FF|\in O(|X|^2)$. Therefore, $V(G)$ is of
	polynomial size in $|X|$. Lastly, checking whether $u \prec_G v$ for two vertices $u,v \in V(G)$
	is equivalent to checking wether there exists a $vu$-path. This can be done in $O(|V(G)| +
	|E(G)|)$ time by breadth-first search \cite{cormen2022introduction}. Since $|E(G)| \leq |V(G)|^2$
	and $|V(G)|$ is polynomial in $|X|$, verifying whether $G$ agrees with $(\RR,\FF)$ can be done in
	polynomial time in $|X|$. Thus, Algorithm~\ref{alg:TC-F-v2} runs in polynomial time.
\end{proof}

%
%
\section{Consistency Problems with Additional LCA-Constraints}
\label{sec:strong_and_LCA_problem}

In the previous section, a triple $xy|z$ was considered not to be displayed by a DAG $G$ whenever at
least one of the conditions required for display fails. In particular, if at least one of
$\lca_G(xy)$, $\lca_G(xz)$ and $\lca_G(yz)$ is not well-defined, then $G$ does not display $xy|z$.
This convention allowed us to use $\FFRR$-extensions to exclude certain forbidden triples by making
at least one of their relevant LCAs not well-defined. Analogously, $\FR$-extensions were used in
Section~\ref{sec:ATC-F} to exclude forbidden anchored triples via non-well-defined relevant LCAs.

One may, however, want to impose a stronger interpretation of forbidden anchored or rooted triples:
a forbidden triple should fail to be displayed because the relative positions of the relevant LCAs
do not satisfy the corresponding display condition, and not merely because one of the underlying
LCAs is not well-defined. This motivates variants in which certain LCAs, possibly associated with
forbidden anchored or rooted triples, or even all LCAs, are required to be well-defined.

In the following, we first discuss the problems for anchored triples in Section~\ref{sec:strong-ATC-F_and_LCA-ATC-F}, 
before turning to ``ordinary'' rooted triples in Section~\ref{sec:strong-TC-F_and_LCA-TC-F}.

\subsection{Anchored Triples}
\label{sec:strong-ATC-F_and_LCA-ATC-F}

We first consider a general variant for anchored triples in which all pairwise LCAs in a prescribed
set $\mathscr W$ are required to be well-defined.

\begin{problem}[\PROBLEM{$\mathscr W$-Anchored Triples Consistency with Forbidden Triples
($\mathscr W$-ATC-F)}]\ \\
  \begin{tabular}{ll}
    \emph{Input:}    & Two sets $\RR$ and $\FF$ of anchored triples and a set
                       $\mathscr W\subseteq \pairs(\XRF)$. \\
    \emph{Question:} & Is there a phylogenetic network (resp., DAG) $G$ on $\XRF$
                       that $\anchor$-agrees with $(\RR,\FF)$ and for which \\ & $\lca_G(ab)$ is well-defined
                       for all $ab\in\mathscr W$?
  \end{tabular}
\end{problem}

Note that $\mathscr W$-\PROBLEM{ATC-F} generalizes \PROBLEM{ATC-F}, since every instance of
\PROBLEM{ATC-F} can be viewed as an instance of $\mathscr W$-\PROBLEM{ATC-F} by taking $\mathscr
W=\emptyset$, and the answer remains unchanged. 
The next definition therefore naturally extends Definition~\ref{def:anchor-display}.

\begin{definition}
Let $(\RR,\FF)$ be a pair of anchored triple sets and let $\mathscr W\subseteq \pairs(\XRF)$. A DAG
$G$ \emph{$\mathscr W$-$\anchor$-agrees} with $(\RR,\FF)$ if $G$ \anchor-agrees with $(\RR,\FF)$ and
$\lca_G(ab)$ is well-defined for every $ab\in\mathscr W$.
\end{definition}

The following result, analogous to Theorem~\ref{thm:weaklyagrees_iff_strictly_real}, also reflects
that $\mathscr W$-\PROBLEM{ATC-F} generalizes \PROBLEM{ATC-F}. Although
Theorem~\ref{thm:W-anchor-agree}, in fact, implies parts of
Theorem~\ref{thm:weaklyagrees_iff_strictly_real} we have treated \PROBLEM{ATC-F} separately in
Section~\ref{sec:ATC-F} in order to introduce the main ideas without the additional notation and
terminology needed for the more general $\mathscr W$-setting.

\begin{theorem}\label{thm:W-anchor-agree}
Let $(\RR,\FF)$ be a pair of anchored triple sets and let $\mathscr W\subseteq\pairs(\XRF)$.
Moreover, let $R \coloneqq \{(xy,xz) \,:\, \underline{x}y|z\in\RR\} \cup
\{(ab,ab) \,:\, ab\in\mathscr W\}$ and $F = \{(xy,xz) \,:\, \underline{x}y|z \in \FF\}$ be relations
on $\pairs(\XRF)$. Then the following statements are equivalent for every DAG $G$ on $\XRF$:
\begin{enumerate}
    \item $G$ $\mathscr W$-$\anchor$-agrees with $(\RR,\FF)$.
    \item $G$ \RF-realizes $(R,F)$ and $(xz,xy)\notin \tc(R)$ for all $\underline{x}y|z\in\RR$. 
\end{enumerate}
In particular, it can be decided in polynomial time in $|\XRF|$ whether there exists a phylogenetic
DAG or network on $\XRF$ that $\mathscr W$-$\anchor$-agrees with $(\RR,\FF)$. In the affirmative
case, the $\FR$-extension $G$ of the canonical DAG $\GG_{R,F}$ and the network obtained from $G$
according to Lemma~\ref{lem:DAG2Network} $\mathscr W$-$\anchor$-agree with $(\RR,\FF)$ and both can
be constructed in polynomial time in $|\XRF|$. 
\end{theorem} 
\begin{proof} 
	Let $\RR$, $\FF$, $\mathscr W$, $G$, $R$, and $F$ be as stated. Assume first that $G$ $\mathscr
	W$-$\anchor$-agrees with $(\RR,\FF)$. Then, for every $\underline{x}y|z\in\RR$, the LCAs
	$\lca_G(xy)$ and $\lca_G(xz)$ are well-defined and satisfy $ \lca_G(xy)\prec_G \lca_G(xz)$.
	Moreover, for every $ab\in\mathscr W$, the LCA $\lca_G(ab)$ is well-defined by assumption. Hence
	all LCAs required by elements of $\support_R^+$ are well-defined.

	We start by showing that $G$ realizes $R$. To verify \axiom{I1} and \axiom{I2}, let $(xy,xz)\in R$
	and suppose first that $(xz,xy)\notin\tc(R)$. In this case, $xy\neq xz$ holds and, with
	$(xy,xz)\in R$, it follows that there is an anchored triple $\underline{x}y|z\in\RR$. Since $G$
	\anchor-displays $\underline{x}y|z\in\RR$, we have $\lca_G(xy)\prec_G\lca_G(xz)$. Thus
	Condition~\axiom{I1} is satisfied. To check that \axiom{I2} is satisfied, assume that
	$(xz,xy)\in\tc(R)$. If $xy=xz$, then $xy \in \support^+_R$ ensures that $\lca_G(xy)$ is
	well-defined and $\lca_G(xy)=\lca_G(xz)$. Assume now that $xy\neq xz$ in which case there is an
	anchored triple $\underline{x}y|z\in\RR$. Since $(xz,xy)\in\tc(R)$, there is a $(xz,xy)$-chain
	$xz=p_0 \,R\, p_1 \,R\, \dots \,R\, p_k=xy$ with $k\geq 1$. Each element $(p_{j-1},p_j)\in R$ is
	of one of the following two types. If $(p_{j-1},p_j)=(ab,ab)$ holds for some $ab\in\mathscr W$,
	then $\lca_G(p_{j-1})=\lca_G(p_j)$. If, instead, $(p_{j-1},p_j)=(ab,ac)$ holds for some anchored
	triple $\underline{a}b|c\in\RR$, we have $\lca_G(p_{j-1})\prec_G \lca_G(p_j)$, since $G$
	$\anchor$-displays all anchored triples in $\RR$. Combining these two cases, we see that
	$\lca_G(p_{j-1})\preceq_G\lca_G(p_{j})$ for all $1\leq j\leq k$. By transitivity
	of $\preceq_G$, the entire $(xz,xy)$-chain therefore implies $\lca_G(xz)\preceq_G\lca_G(xy)$.
	However, this contradicts $\lca_G(xy)\prec_G\lca_G(xz)$, which follows from
	$\underline{x}y|z\in\RR$ being $\anchor$-displayed by $G$. Consequently, $(xz,xy)\notin\tc(R)$ for
	all $\underline{x}y|z\in\RR$, and Condition~\axiom{I2} is vacuously satisfied for all $(p,q)\in R$
	with $p\neq q$.

	It remains to verify the forbidden constraints. Let $(xy,xz)\in F$. Hence, there is
	$\underline{x}y|z\in\FF$. Since $G$ does not $\anchor$-display $\underline{x}y|z$, either at least
	one of $\lca_G(xy)$ and $\lca_G(xz)$ is not well-defined, or both are well-defined and
	$\lca_G(xy)\nprec_G\lca_G(xz)$. This is precisely Condition~\axiom{F}. Thus, $G$ \RF-realizes
	$(R,F)$.

	Conversely, suppose that $G$ \RF-realizes $(R,F)$ and that $(xz,xy)\notin\tc(R)$ for all
	$\underline{x}y|z\in\RR$. Let $\underline{x}y|z\in\RR$. Then $(xy,xz)\in R$, and since
	$(xz,xy)\notin\tc(R)$, Condition~\axiom{I1} implies $ \lca_G(xy)\prec_G\lca_G(xz) $. In
	particular, $\underline{x}y|z$ is $\anchor$-displayed by $G$. Now let $\underline{x}y|z\in\FF$.
	Then $(xy,xz)\in F$. By Condition~\axiom{F}, either at least one of $\lca_G(xy)$ and $\lca_G(xz)$
	is not well-defined, or both are well-defined and $\lca_G(xy)\nprec_G\lca_G(xz)$. Hence $G$ does
	not $\anchor$-display $\underline{x}y|z$. Finally, for every $ab\in\mathscr W$, we have
	$(ab,ab)\in R$. Thus $ab\in\support_R^+$, and since $G$ realizes $R$, the LCA $\lca_G(ab)$ is
	well-defined. Therefore, $G$ $\mathscr W$-$\anchor$-agrees with $(\RR,\FF)$.

	We now verify the last two statements in this theorem. By
	Theorem~\ref{thm:characterization_AF_realized}, \RF-realizability can be tested in polynomial time
	in $|\XRF|$. The additional conditions involving $\tc(R)$ can also be checked in polynomial time.
	If these conditions are satisfied, let $H$ be the $\FR$-extension of the canonical DAG
	$\GG_{R,F}$. By Theorem~\ref{thm:characterization_AF_realized}, $H$ is a DAG on $\XRF$ that
	\RF-realizes $(R,F)$. This together with $(xz,xy)\notin\tc(R)$ for all $ \underline{x}y|z\in\RR$
	and the equivalence between Statements (1) and (2) implies that $H$ $\mathscr
	W$-$\anchor$-agrees with $(\RR,\FF)$. Moreover, $\GG_{R,F}$ is phylogenetic by
	Proposition~\ref{prop:properties_of_canonical_DAG}, since $(R,F)$ is \RF-realizable. This together
	with the construction of the $\FR$-extension (cf.\ Definition~\ref{def:FRextension}) implies that
	$H$ is phylogenetic. Finally, let $N$ be the network obtained from $H$ according to
	Lemma~\ref{lem:DAG2Network}. 
    By this lemma, the ancestor relation among the vertices of $H$ is preserved in $N$. Moreover, for
  every pair $ab\in\pairs(\XRF)$ with $\LCA_H(ab)\neq\emptyset$, we have $\LCA_H(ab)=\LCA_N(ab)$.
  Thus, every LCA that is well-defined in $H$ remains well-defined and unchanged in $N$. It follows
  immediately that all anchored triples in $\RR$ remain \anchor-displayed by $N$, and that all LCAs
  required by $\mathscr W$ remain well-defined in $N$. Now let $\underline{x}y|z\in\FF$. Then
  $(xy,xz)\in F$, and hence $xy,xz\in\support_F$. By construction of the $\FR$-extension, each of
  the LCA sets $\LCA_H(xy)$ and $\LCA_H(xz)$ is nonempty. Thus, $\LCA_H(xy) = \LCA_N(xy)$ and
  $\LCA_H(xz) = \LCA_N(xz)$ holds by Lemma~\ref{lem:DAG2Network}. This together with the ancestor
  relation among the vertices of $H$ being preserved in $N$, implies that $N$ does not
  \anchor-display any forbidden anchored triple.
	Hence, the network $N$ also $\mathscr W$-$\anchor$-agrees with $(\RR,\FF)$. By
	Lemma~\ref{lem:DAG2Network}, $N$ is phylogenetic and has leaf set $\XRF$. Lastly, by
	Theorem~\ref{thm:characterization_AF_realized}, $H$ can be constructed in polynomial time in
	$|\XRF|$. The network $N$ is then obtained from $H$ by Lemma~\ref{lem:DAG2Network}, and one easily
	verifies that this construction is also polynomial in $|\XRF|$.
\end{proof}

Theorem~\ref{thm:W-anchor-agree} shows that the problem $\mathscr W$-\PROBLEM{ATC-F} can be solved 
in polynomial time in $|\XRF|$. We now specialize $\mathscr{W}$\PROBLEM{-ATC-F} to the two problems 
mentioned in the beginning of this section. The first specialization uses the following definition, 
analogous to Definition~\ref{def:tsupp}.

\begin{definition}
If $\mathcal S$ is a set of anchored triples, we define the \emph{anchored-triple-support} as
$\atsupp_{\mathcal{S}}\coloneqq \bigcup_{\underline{x}y|z \in \mathcal{S}} \{xy,xz\}$.
\end{definition}

By choosing $\mathscr{W}=\atsupp_\FF$ in $\mathscr{W}$\PROBLEM{-ATC-F}, we obtain the following
problem.

\begin{problem}[\PROBLEM{Strong Anchored Triples Consistency with Forbidden Triples (strong-ATC-F)}]\ \\
  \begin{tabular}{ll}
    \emph{Input:}    & Two sets $\RR$ and $\FF$ of anchored triples. \\
    \emph{Question:} & Is there a phylogenetic network (resp., DAG) $G$ on $\XRF$ that
                       \anchor-agrees with $(\RR,\FF)$ and for which the \\ 
                     &  LCAs $\lca_G(xy)$ and $\lca_G(xz)$ are well-defined for all
                        $\underline{x}y|z\in\FF$?
  \end{tabular}
\end{problem} 

Now, recall that a DAG $G$ on $X$ has the 2-lca-property if $\lca_G(xy)$ is well-defined for all
$xy\in\pairs(X)$. Therefore, choosing $\mathscr{W}=\pairs(\XRF)$ in $\mathscr{W}$\PROBLEM{-ATC-F}
gives rise to the following problem.

\begin{problem}[\PROBLEM{LCA Anchored Triples Consistency with Forbidden Triples (LCA-ATC-F)}]\ \\
  \begin{tabular}{ll}
    \emph{Input:}    & Two sets $\RR$ and $\FF$ of anchored triples. \\
   \emph{Question:}  & Is there a phylogenetic network (resp., DAG) on $\XRF$ with 2-lca-property 
                       that \anchor-agrees with $(\RR,\FF)$? 
  \end{tabular}
\end{problem}

Clearly, Theorem~\ref{thm:W-anchor-agree} immediately implies the following
\begin{corollary}
   The problems \PROBLEM{strong-ATC-F} and \PROBLEM{LCA-ATC-F} can be solved in polynomial time. 
\end{corollary}

For the two special choices of $\mathscr W$ considered above, the construction from
Theorem~\ref{thm:W-anchor-agree} simplifies further. In the general $\mathscr W$-\PROBLEM{ATC-F}
setting, the solution is obtained from the $\FR$-extension of the canonical DAG $\GG_{R,F}$. For
\PROBLEM{strong-ATC-F} and \PROBLEM{LCA-ATC-F}, however, the prescribed set $\mathscr W$ already
contains all leaf pairs that occur in the forbidden anchored triples. Consequently, the
$\FR$-extension becomes unnecessary, as made precise in the following corollary.

\begin{corollary}\label{cor:no-FR-extension-strong-ATC}
Let $(\RR,\FF)$ be a pair of anchored triple sets and let $\mathscr
W\in\{\atsupp_\FF,\pairs(\XRF)\}$. Moreover, let $R$ and $F$ be as in Theorem~\ref{thm:W-anchor-agree}.
If $(\RR,\FF)$ is a yes-instance of the corresponding $\mathscr W$-\PROBLEM{ATC-F} problem, then the
canonical DAG $\GG_{R,F}$ $\mathscr W$-$\anchor$-agrees with $(\RR,\FF)$. 
\end{corollary}
\begin{proof}
Let $\RR$, $\FF$, $\mathscr{W}$, $R$, and $F$ be as stated. By construction, $(ab,ab)\in R$ for
every $ab\in\mathscr W$, and hence $\mathscr W\subseteq\support_R^+$. Moreover, for the two choices
$\mathscr W=\atsupp_\FF$ and $\mathscr W=\pairs(\XRF)$, we have
$\support_F=\atsupp_\FF\subseteq\mathscr W$. Therefore, $\support_F\setminus\support_R^+=\emptyset$.
Thus, the $\FR$-extension of $\GG_{R,F}$ applies no $xy$-extension and coincides with $\GG_{R,F}$
itself. The claim now follows from Theorem~\ref{thm:W-anchor-agree}.
\end{proof}

In Figure~\ref{fig:strong-ATC-F}, we provide an example of the canonical DAG $\GG_{R,F}$ that
$\mathscr{W}$-\anchor-agrees with $(\RR,\FF)$ for $\mathscr{W} = \atsupp_\FF$.

\begin{figure}
    \centering
    \includegraphics[width=0.8\linewidth]{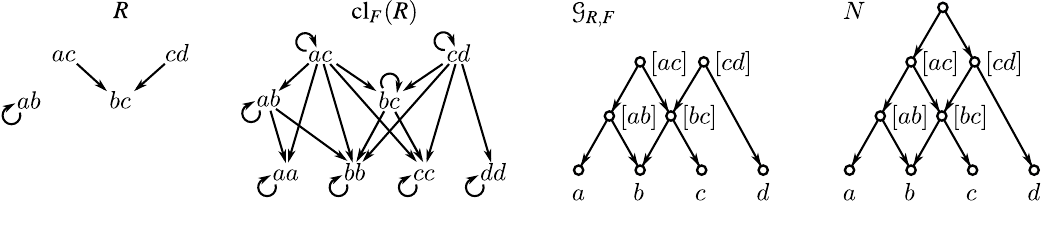}
    \caption{
	  Consider the pair $(\RR,\FF)$ of anchored triple sets with $\RR = \{\underline{c}b|a, \underline{c}b|d\}$ 
    and $\FF = \{\underline{b}c|a\}$ as in Figure~\ref{fig:anchor_displaying_RF}, and let $\mathscr{W} =
	  \atsupp_\FF$. Then $\XRF=\{a,b,c,d\}$ and, according to Theorem~\ref{thm:W-anchor-agree}, the
	  corresponding relations on $\pairs(\XRF)$ are $R=\{(bc,ac),(bc,cd),(ab,ab)\}$ and
	  $F=\{(bc,ab)\}$. For $S \in \{R,\Fcl(R)\}$, an arc $p\to q$ is drawn precisely if $(q,p) \in S$.
	  The canonical DAG $\GG_{R,F}$ $\mathscr{W}$-\anchor-agrees with $(\RR,\FF)$. Indeed, in
	  accordance with Theorem~\ref{thm:W-anchor-agree}, $\GG_{R,F}$ \RF-realizes $(R,F)$ and
	  $(ac,bc),(cd,bc) \notin \tc(R) = R$. Moreover, the phylogenetic network $N$ obtained from
	  $\GG_{R,F}$ by Lemma~\ref{lem:DAG2Network} also $\mathscr{W}$-\anchor-agrees with $(\RR,\FF)$.}
    \label{fig:strong-ATC-F}
\end{figure}

We finish this subsection by showing that \PROBLEM{strong-ATC-F} and \PROBLEM{LCA-ATC-F} are
equivalent. More precisely, an input pair $(\RR,\FF)$ is a yes-instance of \PROBLEM{strong-ATC-F} if and
only if it is a yes-instance of \PROBLEM{LCA-ATC-F}. In fact, we prove the following slightly more
general statement.

\begin{proposition}\label{prop:strong_ATC_F_equivalent_LCA_ATC_F}
Let $(\RR,\FF)$ be a pair of anchored triple sets, let
$\mathscr W = \atsupp_\FF$ and $\mathscr{W}' \subseteq\pairs(\XRF)$ such that 
$\mathscr{W} \subseteq \mathscr{W}'$. 
Then the following statements are equivalent:
\begin{enumerate}
    \item There exists a DAG that $\mathscr W$-$\anchor$-agrees with $(\RR,\FF)$.
    \item There exists a DAG that $\mathscr W'$-$\anchor$-agrees with $(\RR,\FF)$.
\end{enumerate}
\end{proposition}
\begin{proof} 
Let $\RR$, $\FF$, $\mathscr{W}$, and $\mathscr{W}'$ be as stated. Clearly, Statement~(2) implies
Statement~(1), as $\mathscr{W} \subseteq \mathscr{W'}$. 

Thus, it remains to show that Statement~(1) implies Statement~(2). Suppose that there exists a
DAG $G$ that $\mathscr W$-$\anchor$-agrees with $(\RR,\FF)$. By Theorem~\ref{thm:W-anchor-agree}, we
may assume that $G$ is a phylogenetic network. In particular, $\LCA_G(xy)\neq\emptyset$ for all
$xy\in\pairs(\XRF)$.
We now construct a DAG $H$ that $\mathscr W'$-$\anchor$-agrees with $(\RR,\FF)$. For that, we want
to preserve the structure of $G$ and ensure that $\lca_H(xy)$ becomes well-defined for each $xy \in
\mathscr W'$ for which $\lca_G(xy)$ is not well-defined. Thus, put $\mathscr Z \coloneqq \{xy \in
\mathscr W' \,:\, \lca_G(xy) \text{ is not well-defined}\}$. We then define the DAG $H$ by setting 
\[
V(H) \coloneqq V(G) \cup \{u_{xy} \,:\, xy \in  \mathscr Z\} 
\]
and 
\[
E(H) \coloneqq E(G) \cup \{(u_{xy},x),(u_{xy},y) \,:\, xy \in \mathscr Z\}
\cup \{(v,u_{xy}) \,:\, xy \in \mathscr Z \text{ and } v \in \LCA_G(xy)\}.  
\]

Clearly, $H$ is a DAG, since any directed cycle would have to contain a new vertex $u_{xy}$, but the
only outgoing arcs of $u_{xy}$ lead to the leaves $x$ and $y$ and, thus, no such cycle can exist in
$H$. We now argue that the ancestor relation among the vertices in $G$ is preserved in $H$. Since
$G$ is a subgraph of $H$ and both are DAGs, it follows immediately that $a \preceq_G b$ implies
$a\preceq_H b$ for all $a,b \in V(G)$. Conversely, let $a,b\in V(G)$ and suppose that $a\preceq_H
b$. Consider a directed $ba$-path $P$ in $H$. If $P$ contains only vertices of $G$, then $P$ is
already a path in $G$, and hence $a \preceq_G b$. Otherwise, $P$ contains
a vertex $u_{xy}$ for some $xy \in \mathscr Z$. Since $a,b \in V(G)$, the vertex $u_{xy}$ cannot be
the first or last vertex of $P$. Hence, the predecessor of $u_{xy}$ in $P$ must be some $v \in
\LCA_G(xy) \neq\emptyset$ and the successor of $u_{xy}$ in $P$ must either be $x$ or $y$. In
particular, the latter implies that $a\in\{x,y\}$, since $x$ and $y$ are leaves. Moreover, the
subpath of $P$ from $b$ to $v$ is contained in $G$, so we have $v\preceq_G b$ and, since
$a\in\{x,y\}$, clearly $a\preceq_G v$. Thus, $a\preceq_G v\preceq_G b$ holds.
In summary, $a\preceq_G b$ if and only if $ a\preceq_H b$ for all $a,b\in V(G)$.

We now compare the relevant LCAs. For every $ab\in\pairs(\XRF)\setminus\mathscr Z$, no new vertex
$u_{xy}$ is a least common ancestor of $a$ and $b$. Since the ancestor relation among vertices in
$G$ is preserved, it follows that $\LCA_H(ab)=\LCA_G(ab)$ for all
$ab\in\pairs(\XRF)\setminus\mathscr Z$. Now, let $xy\in\mathscr Z$. Clearly, $u_{xy}$ is a common
ancestor of $x$ and $y$ in $H$. To argue that $u_{xy} = \lca_H(xy)$, suppose there exists another
common ancestor $w\neq u_{xy}$ of $x$ and $y$ in $H$. Then $w\in V(G)$, and hence $w$ is a common
ancestor of $x$ and $y$ in $G$. Thus, there is some $v\in\LCA_G(xy)$ with $v\preceq_G w$. By
construction, $(v,u_{xy})\in E(H)$, and therefore $u_{xy}\preceq_H v\preceq_H w$. As this holds for
every common ancestor $w$ of $x$ and $y$ in $H$, it follows that $u_{xy}$ is the unique LCA of $x$
and $y$ in $H$.

In summary, $a\preceq_G b$ if and only if $a\preceq_H b$ for all $a,b\in V(G)$ and $\LCA_G(ab) =
\LCA_H(ab)$ for all $ab \in \pairs(\XRF) \setminus \mathscr Z$. Moreover, since $G$ $\mathscr
W$-$\anchor$-agrees with $(\RR,\FF)$ and $\mathscr W = \atsupp_\FF$, all relevant LCAs for the
anchored triples in $\RR$ and $\FF$ are well-defined in $G$ and, thus, in $H$. In particular, the
latter two arguments imply that $H$ $\mathscr W$-$\anchor$-agrees with $(\RR,\FF)$. This together
with $\lca_{H}(xy) = u_{xy}$ being well-defined for each $xy \in \mathscr Z$ and $\lca_G(xy) =
\lca_H(xy)$ for all $xy \in \mathscr{W}' \setminus \mathscr{Z}$ implies that $H$ $\mathscr
W'$-$\anchor$-agrees with $(\RR,\FF)$.
\end{proof}

Proposition~\ref{prop:strong_ATC_F_equivalent_LCA_ATC_F} applied to 
$\mathscr W=\atsupp_\FF$ and $\mathscr W'=\pairs(\XRF)$ yields the following result.

\begin{corollary}\label{cor:strong-ATC-F-equivalent-LCA-ATC-F}
The problems \PROBLEM{strong-ATC-F} and \PROBLEM{LCA-ATC-F} are equivalent.
\end{corollary}
Note that every yes-instance of \PROBLEM{strong-ATC-F} is clearly a yes-instance of \PROBLEM{ATC-F}. The
converse does not hold, as shown in Figure~\ref{fig:ATC-F_not_strong-ATC-F}.

\begin{figure}
    \centering
    \includegraphics[width=0.8\textwidth]{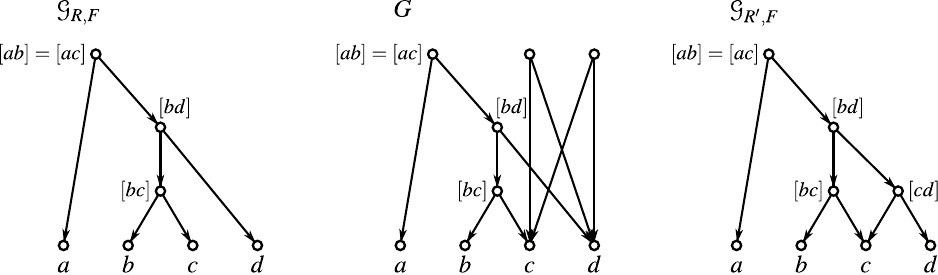}
    \caption{
    Consider the pair $(\RR,\FF)$ of anchored triple sets with
    $\RR = \{\underline{b}c|d, \underline{b}d|a, \underline{c}b|a\}$ and $\FF =
    \{\underline{c}d|a\}$ and let $\mathscr{W} = \atsupp_\FF=\{cd, ac\}$. Define $R =
    \{(bc,bd),(bd,ab),(bc,ac)\}$, $R' = R \cup \{(cd,cd),(ac,ac)\}$, and $F = \{(cd,ac)\}$ as
    relations on $\pairs(\XRF)$. Shown are the canonical DAG $\GG_{R,F}$, its $\FR$-extension $G$,
    and the canonical DAG $\GG_{R',F}$, where the latter coincides with its $\FR$-extension. Here,
    $G$ \anchor-agrees with $(\RR,\FF)$. However, $\GG_{R',F}$ does not $\mathscr{W}$-\anchor-agree
    with $(\RR,\FF)$, since $\underline{c}d|a$ is \anchor-displayed and, thus, by
    Theorem~\ref{thm:W-anchor-agree}, $(\RR,\FF)$ is not $\mathscr{W}$-\anchor-agreeable. 
    }
    \label{fig:ATC-F_not_strong-ATC-F}
\end{figure}

\subsection{Rooted Triples}
\label{sec:strong-TC-F_and_LCA-TC-F}

As in the previous subsection, we start with the following generalized consistency problem.

\begin{problem}[\PROBLEM{$\mathscr W$-Triples Consistency with Forbidden Triples}
($\mathscr W$-\PROBLEM{TC-F})]\ \\
  \begin{tabular}{ll}
    \emph{Input:}    & Two sets $\RR$ and $\FF$ of triples and a set
                       $\mathscr W\subseteq \pairs(\XRF)$. \\
    \emph{Question:} & Is there a phylogenetic network (resp., DAG) $G$ on
                       $\XRF$ that agrees with  $(\RR,\FF)$ and for which
                       \\ & $\lca_G(ab)$ is well-defined for all $ab\in\mathscr W$?
  \end{tabular}
\end{problem}

Note that, just as \PROBLEM{ATC-F} is a special case of
$\mathscr W$-\PROBLEM{ATC-F}, the problem \PROBLEM{TC-F} is a special case of
$\mathscr W$-\PROBLEM{TC-F}, obtained by taking $\mathscr W=\emptyset$.

We first consider an auxiliary version of the problem $\mathscr W$-\PROBLEM{ATC-F} 
in which the forbidden triples in $\FF$ play no role. This
will provide a realizability-based characterization for DAGs that display all required triples and
satisfy the prescribed well-definedness of relevant LCAs. Recall that $R_\RR \coloneqq \{(xy,xz),(xy,yz)
\,:\, xy|z \in \RR\}$.

\begin{lemma}\label{lem:display-with-W-lcas}
Let $(\RR,\FF)$ be a pair of triple sets, let $\mathscr W\subseteq \pairs(\XRF)$, and let $G$ be
a DAG on $\XRF$. Moreover, put $R\coloneqq R_\RR \cup \{(ab,ab) \,:\, ab\in\mathscr W\}$. Then the
following statements are equivalent:
\begin{enumerate}
    \item $G$ displays all triples in $\RR$ and $\lca_G(ab)$ is well-defined for all $ab\in\mathscr W$.
    \item $G$ realizes $R$ and $(xz,xy),(yz,xy)\notin \tc(R)$
for all $xy|z\in\RR$. 
\end{enumerate}
\end{lemma}

\begin{proof}
Let $\RR$, $\FF$, $\mathscr W$, $G$, and $R$ be as stated. Assume first that $G$ displays all triples
in $\RR$ and that $\lca_G(ab)$ is well-defined for all $ab\in\mathscr W$. We show first that $G$
realizes $R$. Let $(xy,xz)\in R$. If $(xy,xz)\in R_\RR$, then there is a triple $xy|z\in \RR$ and,
since $G$ displays all triples in $\RR$, we obtain $\lca_G(xy)\prec_G\lca_G(xz)$. Otherwise, i.e. if
$(xy,xz)\in R\setminus R_\RR$, we have $xy=xz$ for some $xy\in\mathscr W$ and, thus,
$\lca_G(xy)=\lca_G(xz)$ is well-defined. In summary, for all $(p,q)\in R$ it holds that
$\lca_G(p)\preceq_G \lca_G(q)$.

We now verify \axiom{I1} and \axiom{I2}. Let $(p,q)\in R$ and thus, $\lca_G(p)\preceq_G \lca_G(q)$.
If $(q,p)\in\tc(R)$, then there is a chain $q=p_0 \;R\; p_1 \;R\; \dots \;R\; p_k=p$. Since for all
$(p_{i},p_{i+1})\in R$ it holds $\lca_G(p_i)\preceq_G \lca_G(p_{i+1})$, transitivity of $\preceq_G$
implies that $\lca_G(q)\preceq_G \lca_G(p)$ and, consequently, $\lca_G(p)=\lca_G(q)$, and so
\axiom{I2} holds. Assume that $(q,p)\notin\tc(R)$. Hence, $p\neq q$. This and $(p,q)\in R$ implies
that $(p,q)\in R_\RR$. Hence, $(p,q)$ comes from some displayed triple $xy|z\in \RR$, i.e.,
$(p,q)=(xy,xz)$ or $(p,q)=(xy,yz)$. Since $G$ displays $xy|z$, we obtain $\lca_G(p)\prec_G
\lca_G(q)$ and \axiom{I1} holds. In summary, $G$ realizes $R$.

It remains to show that $(xz,xy),(yz,xy)\notin\tc(R)$ for every $xy|z\in\RR$. Suppose, for
contradiction, that $(xz,xy)\in\tc(R)$ for some $xy|z\in\RR$. By definition, we have $(xy,xz)\in
R_\RR\subseteq R$. As $G$ realizes $R$, \axiom{I2} would imply $ \lca_G(xy)=\lca_G(xz)$,
contradicting the fact that $G$ displays $xy|z$. Hence $(xz,xy)\notin\tc(R)$. The argument for
$(yz,xy)\notin\tc(R)$ is analogous. We have therefore shown that Statement (1) implies Statement
(2).

Conversely, assume that $G$ realizes $R$ and that $(xz,xy),(yz,xy)\notin\tc(R)$ for every
$xy|z\in\RR$. Since $(ab,ab)\in R$ for every $ab\in\mathscr W$, realizability of $R$ implies that
$\lca_G(ab)$ is well-defined for all $ab\in\mathscr W$.

Now let $xy|z\in\RR$. By definition, $(xy,xz),(xy,yz)\in R_\RR\subseteq R$. Since
$(xz,xy)\notin\tc(R)$ and $(yz,xy)\notin\tc(R)$, \axiom{I1} gives $ \lca_G(xy)\prec_G\lca_G(xz)$ and
$\lca_G(xy)\prec_G\lca_G(yz)$. In particular, the LCAs $\lca_G(xy)$, $\lca_G(xz)$, and $\lca_G(yz)$
are well-defined. Applying Lemma~\ref{lem:xyz-lca(yz)} to $\lca_G(xy)\prec_G\lca_G(xz)$ yields
$\lca_G(yz)\preceq_G\lca_G(xz)$. Similarly, applying Lemma~\ref{lem:xyz-lca(yz)} to
$\lca_G(xy)\prec_G\lca_G(yz)$ yields $ \lca_G(xz)\preceq_G\lca_G(yz)$. Thus, $
\lca_G(xz)=\lca_G(yz)$ holds. Together with $\lca_G(xy)\prec_G\lca_G(xz)$, this shows that $G$
displays $xy|z$. Since this holds for every $xy|z\in\RR$, Statement~(2) implies Statement~(1) and
the proof is complete.
\end{proof}

We now return to the setting with forbidden triples and provide the 
following strengthened notion of agreement.

\begin{definition}
Let $(\RR,\FF)$ be a pair of triple sets and let $\mathscr W\subseteq \pairs(\XRF)$. A DAG $G$
\emph{$\mathscr W$-agrees} with $(\RR,\FF)$ if $G$ agrees with $(\RR,\FF)$ 
 and $\lca_G(ab)$ is well-defined for every $ab\in\mathscr W$.
\end{definition}

The next result is the analogue of Theorem~\ref{thm:char-RF-construction} for $\mathscr
W$-\PROBLEM{TC-F}. On a technical note, the auxiliary relation $F$ used in this result is introduced
only to specify the $\FR$-extension of the canonical DAG $\GG_{Q_K}$ but not to construct
$\GG_{Q_K}$.

\begin{theorem}\label{thm:char-W-TCF}
Let $(\RR,\FF)$ be a pair of triple sets and let $\mathscr W \subseteq \pairs(\XRF)$. Let $R\coloneqq R_\RR \cup \{(ab,ab) \,:\, ab\in\mathscr W\}$ 
and $F = \{(xy,xy),(xz,xz),(yz,yz) \,:\, xy|z \in \FF\}$
be relations on $\pairs(\XRF)$.
Let $Q_K$ be the final relation constructed during some run of $\textsc{Sat}\bigl(R,\FF\bigr)$.
Then the following statements are equivalent:
\begin{enumerate}
    \item There exists a DAG on $\XRF$ that $\mathscr W$-agrees with $(\RR,\FF)$. 
    \item The $\FR$-extension of the canonical DAG $\GG_{Q_K}$ is phylogenetic and $\mathscr
          W$-agrees with $(\RR,\FF)$.
    \item There exists a phylogenetic network on $\XRF$ that $\mathscr W$-agrees with $(\RR,\FF)$. 
    \item The network obtained from the $\FR$-extension of $\GG_{Q_K}$ by
          Lemma~\ref{lem:DAG2Network} is phylogenetic and $\mathscr W$-agrees with $(\RR,\FF)$.
\end{enumerate}
In particular, $\mathscr W$-\PROBLEM{TC-F} can be solved in polynomial time in $|\XRF|$.
\end{theorem}
\begin{proof}
Let $\RR$, $\FF$, $\mathscr{W}$, $R$, $F$, and $Q_K$ be as stated. 
By the same arguments as used in in the proof of Lemma~\ref{lem:Qi_properties_realizable}, 
$\textsc{Sat}\bigl(R,\FF\bigr)$ terminates and thus, the final relation $Q_K$ is well-defined. 
Clearly, Statement~(3) implies Statement~(1),
since every phylogenetic network is, in particular, a DAG. We now argue that $(1) \Rightarrow (2)
\Rightarrow (4) \Rightarrow (3)$. We note in passing that the proof follows, in large parts, the
same ideas as the proof of Theorem~\ref{thm:char-RF-display}. The main difference is that we now
also have to account for the additional constraints $(ab,ab)\in R$ for all $ab\in\mathscr W$ where
$R$ is the relation used as input for the saturation procedure that constructs $Q_K$.

Assume first that Statement~(1) holds, and let $G$ be a DAG on $\XRF$ that $\mathscr W$-agrees with
$(\RR,\FF)$. By Lemma~\ref{lem:display-with-W-lcas}, $G$ realizes $R$ and, for every $xy|z\in\RR$,
we have $(xz,xy),(yz,xy)\notin \tc(R)$. We first show that every relation $Q_0,\dots,Q_K$
constructed during the saturation procedure is contained in $\rel_G$. For $i=0$, we have
$R\subseteq\rel_G$, since $G$ realizes $R$ and by \cite[Lem~7]{LAMSH:25}. Hence,
$Q_0=\cl(R)\subseteq \cl(\rel_G)=\rel_G$ holds by monotonicity of $\cl$ and
\cite[Prop~21]{LAMSH:25}. Now assume that $Q_i\subseteq\rel_G$ and that the saturation procedure
constructs $Q_{i+1} = \cl\bigl(Q_i\cup\{(xz,xy),(yz,xy)\}\bigr)$ because we have $
Q_i\restr=\Rext_{\{xy|z\}}$ for some forbidden triple $xy|z\in\FF_{i}$. Since $Q_i\subseteq\rel_G$,
this implies $\Rext_{\{xy|z\}} = \{(xy,xz),(xy,yz),(xz,yz),(yz,xz)\} \subseteq \rel_G$. Hence, the
LCAs $\lca_G(xy)$, $\lca_G(xz)$, and $\lca_G(yz)$ are well-defined and satisfy $ \lca_G(xy)\preceq_G
\lca_G(xz)=\lca_G(yz)$. Since $G$ does not display the forbidden triple $xy|z$, it cannot hold that
$ \lca_G(xy)\prec_G \lca_G(xz)=\lca_G(yz)$. Therefore, we have $\lca_G(xy)=\lca_G(xz)=\lca_G(yz)$.
Consequently, $ (xz,xy),(yz,xy)\in\rel_G$ and together with $Q_i\subseteq\rel_G$, we obtain $
Q_{i+1} = \cl\bigl(Q_i\cup\{(xz,xy),(yz,xy)\}\bigr) \subseteq \cl(\rel_G) = \rel_G$. By induction,
$Q_i\subseteq\rel_G$ holds for all $i$, $0\leq i \leq K$.

We now show that $Q_K$ is realizable. Note that by idempotency of $\cl$, every $Q_i$ is closed, that
is, $Q_i=\cl(Q_i)$. Additionally, $Q_i = \tc(Q_i) = \cl(Q_i)$ holds for all $0 \leq i \leq K$ by
extensivity of $\tc$ and transitivity of $\cl$. To argue that Condition~\axiom{X1} holds for $Q_K$,
suppose, for contradiction, that \axiom{X1} fails for $Q_K$. Then there are $a,b,x\in \XRF$ with
$ab\neq xx$ such that $(ab,xx)\in Q_K$. Since $Q_K\subseteq\rel_G$, this implies $(ab,xx)\in\rel_G$.
Thus, $\lca_G(ab)$ and $\lca_G(xx)$ are well-defined and $\lca_G(ab) \preceq_G \lca_G(xx)=x$. Since
$x$ is a leaf, $\lca_G(ab)=x$ and $a=b=x$ and, thus, $ab=xx$ holds; contradicting our assumption
$ab\neq xx$. Therefore, \axiom{X1} holds for $Q_K$. Since $Q_K=\cl(Q_K)$ and $Q_K$ satisfies
\axiom{X1}, it follows from \cite[Cor~32]{LAMSH:25} that $Q_K$ is realizable. Consequently, by
Theorem~\ref{thm:char}, the canonical DAG $\GG_{Q_K}$ realizes $Q_K$.

We next show that $\GG_{Q_K}$ displays all triples in $\RR$ and that all LCAs required by
$\mathscr{W}$ are well-defined. Since $(ab,ab)\in R\subseteq Q_K$ for all $ab\in\mathscr W$ and
$Q_K$ is realized by $\GG_{Q_K}$, each $\lca_{\GG_{Q_K}}(ab)$ for $ab\in\mathscr{W}$ is
well-defined. It remains to prove that all triples in $\RR$ are displayed by $\GG_{Q_K}$. Let
$xy|z\in\RR$. We first show that $(xz,xy),(yz,xy)\notin\tc(Q_K)$. Suppose, for contradiction, that
$(xz,xy)\in\tc(Q_K)$. Since $Q_K=\tc(Q_K)$, we have $ (xz,xy)\in Q_K$. Since $Q_K\subseteq\rel_G$,
we obtain $(xz,xy)\in\rel_G$ together with $(xy,xz)\in R\subseteq Q_K\subseteq\rel_G$. Hence,
$(xy,xz), (xz,xy) \in \rel_G$ and, by definition of $\rel_G$, $\lca_G(xy)\preceq_G\lca_G(xz)$ and $
\lca_G(xz)\preceq_G\lca_G(xy)$ must hold, which implies that $\lca_G(xy)=\lca_G(xz)$. This
contradicts the fact that $G$ displays $xy|z$, which requires $\lca_G(xy)\prec_G\lca_G(xz)$.
Therefore, $(xz,xy)\notin\tc(Q_K)$ must hold. The proof that $(yz,xy)\notin\tc(Q_K)$ is analogous.
Since $(xy,xz),(xy,yz)\in R_\RR\subseteq R\subseteq Q_K$ and $(xz,xy),(yz,xy)\notin\tc(Q_K)$ and
since $Q_K$ is realized by $\GG_{Q_K}$, \axiom{I1} implies that
$\lca_{\GG_{Q_K}}(xy)\prec_{\GG_{Q_K}}\lca_{\GG_{Q_K}}(xz)$ and
$\lca_{\GG_{Q_K}}(xy)\prec_{\GG_{Q_K}}\lca_{\GG_{Q_K}}(yz)$. By Lemma~\ref{lem:xyz-lca(yz)}, applied
twice, we then get $\lca_{\GG_{Q_K}}(xz)=\lca_{\GG_{Q_K}}(yz)$ and, hence, $\GG_{Q_K}$ displays
$xy|z$.

We continue by showing that $\GG_{Q_K}$ displays none of the triples in $\FF^\circ \coloneqq \{xy|z
\in \FF \,:\, xy,xz,yz \in \support^+_R\}$. Let $xy|z\in\FF^\circ$ and assume, for contradiction,
that $\GG_{Q_K}$ displays $xy|z$. By Lemma~\ref{lem:displays-iff-restr-equal}, we have $
\rel_{\GG_{Q_K}}\relrestr=\Rext_{\{xy|z\}}$. Since $Q_K$ is closed and realized by its canonical DAG
$\GG_{Q_K}$, \cite[Thm~47]{LAMSH:25} implies $Q_K = \rel_{\GG_{Q_K}} \cap
\bigl(\support_{Q_K}^+\times\support_{Q_K}^+\bigr)$. Since $R\subseteq Q_K$, where both $R$ and
$Q_K$ are relations on $\pairs(\XRF)$, it is easily seen that $\support^+_R
\subseteq\support_{Q_K}^+$. Hence, in particular, $xy,xz,yz \in \support^+_R
\subseteq\support_{Q_K}^+$. Consequently, the previous arguments taken together imply
$Q_K\restr=\rel_{\GG_{Q_K}}\relrestr=\Rext_{\{xy|z\}}. $ But then $xy|z$ would belong to the set of
forbidden triples selected by the final saturation step. This contradicts the fact that $Q_K$ is the
final relation. Thus $\GG_{Q_K}$ does not display $xy|z$. Since $xy|z\in\FF^\circ$ was arbitrary,
$\GG_{Q_K}$ displays no triple in $\FF^\circ$. In summary, $\GG_{Q_K}$ $\mathscr W$-agrees with
$(\RR,\FF^\circ)$. 

It remains to show that the $\FR$-extension $\tilde G$ of $\GG_{Q_K}$ $\mathscr W$-agrees with
$(\RR,\FF)$. Recall that to obtain $\tilde G$, we apply an $ab$-extension to $\GG_{Q_K}$ for every
$ab \in \support_F \setminus \support^+_R$. Thus, for every triple $xy|z \in \RR \cup \FF^\circ$, no
$xy$-, $xz$-, or $yz$-extension was applied to obtain $\tilde G$. This together with
Observation~\ref{obs:xy-extension1} implies that $\tilde G$ $\mathscr W$-agrees with
$(\RR,\FF^\circ)$. Moreover, for every triple $xy|z \in \FF \setminus \FF^\circ$, at least one
element $xy,xz,yz$ is not contained in $\support^+_R$. Hence, an $xy$-, $xz$-, or $yz$-extension was
applied. By Observation~\ref{obs:xy-extension1}, the corresponding LCA is not well-defined and,
thus, $xy|z$ is not displayed by $\tilde G$. Consequently, $\tilde G$ $\mathscr W$-agrees with
$(\RR,\FF)$. Lastly, by Proposition~\ref{prop:properties_of_canonical_DAG} and the fact that $Q_K$
is realizable, $\GG_{Q_K}$ is phylogenetic and by construction, $\tilde G$ is phylogenetic. Thus,
Statement~(2) holds.

We continue by showing that Statement~(2) implies Statement~(4). Assume that the $\FR$-extension $G$
of $\GG_{Q_K}$ is phylogenetic and $\mathscr W$-agrees with $(\RR,\FF)$. By
Observation~\ref{obs:xy-extension1}, $G$ is a DAG on $\XRF$, as $Q_K$ is a relation on
$\pairs(\XRF)$. Let $N$ be the network obtained from $G$ by Lemma~\ref{lem:DAG2Network}. By this
lemma, $N$ is a phylogenetic network on $\XRF$, the ancestor relation among the vertices of $G$ is
preserved in $N$, and every nonempty LCA set in $G$ is preserved in $N$. Since $G$ displays all
triples in $\RR$, the relevant LCAs are well-defined in $G$ and hence remain well-defined and
unchanged in $N$. Therefore, all triples in $\RR$ are displayed by $N$. Similarly, $\lca_N(ab)$ is
well-defined for every $ab\in\mathscr W$. 

It remains to show that no triple in $\FF$ is displayed by $N$. Reusing the notation $\FF^\circ$ of
a previous paragraph, consider first $xy|z\in\FF^\circ$. Then $xy,xz,yz\in\support_R^+$. Then, $xy
\in \support^+_{R_\RR}$ or $xy \in \mathscr{W}$. Since $G$ $\mathscr W$-agrees with $(\RR,\FF)$, the LCA
$\lca_G(xy)$ is well-defined in both cases. By a similar argument, $\lca_G(xz)$ and $\lca_G(yz)$ are
also well-defined. Thus, since $G$ does not display $xy|z$, the display condition must fail in $G$.
In $N$, all relevant LCAs are unchanged and the ancestor relation among vertices of $G$ is
preserved, so the display condition still fails in $N$. Thus $N$ does not display any triple in
$\FF^\circ$. Now let $xy|z\in\FF\setminus\FF^\circ$. Then at least one pair $ab\in{xy,xz,yz}$
belongs to $\support_F\setminus\support_R^+$, and hence an $ab$-extension was applied when
constructing $G$. By Observation~\ref{obs:xy-extension1}, this gives $|\LCA_G(ab)|\geq 2$. Since
$\LCA_G(ab)\neq \emptyset$, Lemma~\ref{lem:DAG2Network} yields $\LCA_N(ab)=\LCA_G(ab)$. Thus
$\lca_N(ab)$ is still not well-defined, and so $N$ does not display $xy|z$. Hence, $N$ $\mathscr
W$-agrees with $(\RR,\FF)$, and Statement~(4) follows.

Finally, Statement~(4) immediately implies Statement~(3), since the network described in
Statement~(4) is a phylogenetic network on $\XRF$ that $\mathscr W$-agrees with $(\RR,\FF)$.

In summary, Statements~(1), (2), (3), and (4) are equivalent. Moreover, as in the proof of
Theorem~\ref{thm:TCF-problem}, the final relation $Q_K$ and the canonical DAG $\GG_{Q_K}$ can be
computed in polynomial time in $|\XRF|$. The $\FR$-extension of $\GG_{Q_K}$ can also be constructed
in polynomial time, and checking whether it is phylogenetic and $\mathscr W$-agrees with
$(\RR,\FF)$ is polynomial as well. Hence $\mathscr W$-\PROBLEM{TC-F} can be solved in polynomial
time in $|\XRF|$.
\end{proof}

Moreover, for $\mathscr{W} = \emptyset$, the $\FR$-extension of a DAG coincides with the $\FFRR$-extension. 
Thus, Theorem~\ref{thm:char-W-TCF} yields the results for \PROBLEM{TC-F} as the special case obtained by taking
$\mathscr W=\emptyset$ (cf.\ Theorem~\ref{thm:char-RF-construction}).
Now consider the following two problems. 

\begin{problem}[\PROBLEM{Strong Triples Consistency with Forbidden Triples (strong-TC-F)}]\ \\ \label{prob:TCF}
  \begin{tabular}{ll}
    \emph{Input:}    & Two sets $\RR$ and $\FF$ of triples. \\
    \emph{Question:} & Is there a phylogenetic network (resp., DAG) $G$ on $\XRF$ that agrees with $(\RR,\FF)$ 
    \\ & and for which the LCAs $\lca_G(xy)$, $\lca_G(xz)$, and $\lca_G(yz)$ are well-defined for all $xy|z\in\FF$?\\
  \end{tabular}
\end{problem}

\begin{problem}[\PROBLEM{LCA Triples Consistency with Forbidden Triples (LCA-TC-F)}]\ \\
  \begin{tabular}{ll}
    \emph{Input:}    & Two sets $\RR$ and $\FF$ of triples. \\
   \emph{Question:}  & Is there a phylogenetic network (resp., DAG) on $\XRF$ with the 2-lca-property that agrees with $(\RR,\FF)$?
  \end{tabular}
\end{problem}

These two variants are just special cases of \PROBLEM{$\mathscr W$-TC-F}. Indeed,
\PROBLEM{strong-TC-F} is obtained by choosing $\mathscr W = \tsupp_\FF$ whereas \PROBLEM{LCA-TC-F}
is obtained by putting $\mathscr W = \pairs(\XRF)$. The following result shows that these two
variants of \PROBLEM{$\mathscr W$-TC-F} can be solved in polynomial time. 

\begin{corollary}\label{cor:no-FR-extension-strong-TC}
The problems \PROBLEM{strong-TC-F} and \PROBLEM{LCA-TC-F} can be solved in polynomial time.
Moreover, let $(\RR,\FF)$ be a pair of triple sets and let $\mathscr
W\in\{\tsupp_\FF,\pairs(\XRF)\}$. Define $R$, $F$, and $Q_K$ as in Theorem~\ref{thm:char-W-TCF}. If
$(\RR,\FF)$ is a yes-instance of the corresponding $\mathscr W$-\PROBLEM{TC-F} problem, then the
canonical DAG $\GG_{Q_K}$ $\mathscr W$-agrees with $(\RR,\FF)$.
\end{corollary}
\begin{proof}
The polynomial-time solvability of \PROBLEM{strong-TC-F} and \PROBLEM{LCA-TC-F} follows directly
from Theorem~\ref{thm:char-W-TCF} by taking $\mathscr W=\tsupp_\FF$ and $\mathscr W=\pairs(\XRF)$,
respectively.

It remains to prove the second assertion. Let $\RR$, $\FF$, $R$, $F$, and $Q_K$ be as stated. By
construction of $F$, we have $\support_F=\tsupp_\FF$. For both choices $\mathscr W\in
\{\tsupp_\FF,\pairs(\XRF)\}$, this implies $\support_F\subseteq\mathscr W$. Moreover, by definition
of $R$, we have $(ab,ab)\in R$ for every $ab\in\mathscr W$, and hence $\mathscr
W\subseteq\support_R^+$. Consequently, $\support_F\setminus\support_R^+=\emptyset$. Thus, the
$\FR$-extension of $\GG_{Q_K}$ used in Theorem~\ref{thm:char-W-TCF} applies no $ab$-extension and,
therefore, coincides with $\GG_{Q_K}$. Hence, if $(\RR,\FF)$ is a yes-instance of the corresponding
$\mathscr W$-\PROBLEM{TC-F} problem, Theorem~\ref{thm:char-W-TCF} implies that $\GG_{Q_K}$ $\mathscr
W$-agrees with $(\RR,\FF)$.
\end{proof}

Similar to the anchored triple case, we can now also prove that the \PROBLEM{strong-TC-F} and
\PROBLEM{LCA-TC-F} problem are equivalent by proving a slightly stronger statement.

\begin{proposition}\label{prop:strong_TC_F_equivalent_LCA_TC_F}
Let $(\RR,\FF)$ be a pair of triple sets, let $\mathscr W = \tsupp_\FF$ and $\mathscr{W}'
\subseteq\pairs(\XRF)$ such that $\mathscr{W} \subseteq \mathscr{W}'$. Then the following statements
are equivalent:
\begin{enumerate}
    \item There exists a DAG that $\mathscr W$-agrees with $(\RR,\FF)$.
    \item There exists a DAG that $\mathscr W'$-agrees with $(\RR,\FF)$.
\end{enumerate}
\end{proposition}
\begin{proof}
Let $\RR$, $\FF$, $\mathscr{W}$, and $\mathscr{W}'$ be as stated. 
Clearly, Statement~(2) implies Statement~(1), since $\mathscr W\subseteq\mathscr W'$.

It remains to show that Statement~(1) implies Statement~(2). 
We use here the same ideas as for the proof of Proposition~\ref{prop:strong_ATC_F_equivalent_LCA_ATC_F}. 
Suppose that there exists a DAG $G$ that
$\mathscr W$-agrees with $(\RR,\FF)$. By Theorem~\ref{thm:char-W-TCF}, we may assume that this DAG
is a phylogenetic network. In particular, $\LCA_G(xy)\neq\emptyset$ for every $xy\in\pairs(\XRF)$.
We now construct a DAG $H$ that $\mathscr W'$-agrees with $(\RR,\FF)$. For that, put 
$\mathscr Z \coloneqq \{xy \in \mathscr W' \,:\, \lca_G(xy) \text{ is not well-defined}\}$. 
We then define the DAG $H$ by setting 
\[
V(H) \coloneqq V(G) \cup \{u_{xy} \,:\, xy \in  \mathscr Z\} 
\]
and 
\[
E(H) \coloneqq E(G) \cup \{(u_{xy},x),(u_{xy},y) \,:\, xy \in \mathscr Z\}
\cup \{(v,u_{xy}) \,:\, xy \in \mathscr Z \text{ and } v \in \LCA_G(xy)\}.  
\]
As outlined in the proof of Proposition~\ref{prop:strong_ATC_F_equivalent_LCA_ATC_F}, $H$ remains a
DAG and $a\preceq_G b$ if and only if $ a\preceq_H b$ for all $a,b\in V(G)$. Moreover, $\LCA_G(ab) =
\LCA_H(ab)$ holds for all $ab \in \pairs(\XRF) \setminus \mathscr Z$ and $\lca_H(xy)=u_{xy}$ for
each $xy \in \mathscr Z$. Since $G$ $\mathscr W$-agrees with $(\RR,\FF)$ and $\mathscr
W=\tsupp_\FF$, all LCAs relevant to the triples in $\RR$ and $\FF$ are well-defined in $G$ and
therefore, in $H$. In particular, the latter two arguments imply that $H$ $\mathscr W$-agrees with
$(\RR,\FF)$. This together with $\lca_H(xy) = u_{xy}$ being well-defined for each $xy \in \mathscr
Z$ and $\lca_G(xy) = \lca_H(xy)$ for all $xy \in \mathscr{W}' \setminus \mathscr{Z}$ implies that
$H$ $\mathscr W'$-agrees with $(\RR,\FF)$.
\end{proof}

Applying Proposition~\ref{prop:strong_TC_F_equivalent_LCA_TC_F} with
$\mathscr W=\tsupp_\FF$ and $\mathscr W'=\pairs(\XRF)$ yields the following result.
\begin{corollary}\label{cor:strong-TC-F-equivalent-LCA-TC-F}
The problems \PROBLEM{strong-TC-F} and \PROBLEM{LCA-TC-F} are equivalent.
\end{corollary}

\section{Summary and Outlook}

In this work, we studied an alternative notion of displaying rooted triples in phylogenetic DAGs and
networks based on least common ancestor constraints. We also introduced anchored triples as a means
of inferring phylogenetic DAGs and networks. We considered several consistency problems involving
required and forbidden rooted or anchored triples and showed that all of them can be solved in
polynomial time. To this end, we constructed suitable relations of required and forbidden
LCA-constraints and analyzed the corresponding realizability problems. We further used
$xy$-extensions to ensure that two leaves $x$ and $y$ have multiple least common ancestors. This
allowed us to guarantee that certain forbidden rooted or anchored triples are not displayed,
respectively not \anchor-displayed. Conversely, we showed that the LCA of $x$ and $y$ can be forced
to be well-defined by adding the pair $(xy,xy)$ to the required relation. This leads to the problem
$\mathscr W$-\PROBLEM{TC-F}, which generalizes the four consistency problems considered for rooted
triples. Analogously, $\mathscr W$-\PROBLEM{ATC-F} generalizes the corresponding problems for
anchored triples. 
Our results are summarized in Table~\ref{tab:summary_results}. 

It would be interesting to determine whether the problem of constructing a phylogenetic network that
agrees with a pair of required and forbidden triple sets remains solvable in polynomial time when
the solution is restricted to a specific class of networks, such as binary or level-$k$
networks for fixed $k$. This question is particularly relevant because constructing a binary network
that topologically-displays all required triples and none of the forbidden ones is NP-hard
\cite{HHJS:06}.

Biologists may still wish to obtain a phylogenetic network even when no network agrees with the
given pair of required and forbidden triple sets. Suppose that $\RR$ is a triple set for
which the associated relation $R_\RR$ is not realizable. By construction, $R_\RR$ satisfies
\axiom{X1}, and hence \cite[Prop~28]{LAMSH:25} implies that the canonical DAG $\GG_{R_\RR}$ is
well-defined, although it does not realize $R_\RR$. This naturally raises the question of which
triples in $\RR$ are nevertheless displayed by $\GG_{R_\RR}$ and how many such triples can be
guaranteed to be displayed. Analogous questions arise when a set $\FF$ of forbidden triples is
prescribed as well.
More generally, this motivates the optimization problem of constructing a network that
maximizes the number of displayed required triples while simultaneously minimizing the number of
displayed forbidden triples. Similarly, suppose that $(\RR,\FF)$ is an instance for which an
agreeing phylogenetic network exists, but no such network can be required to have all pairwise LCAs
well-defined. One may then ask for an agreeing phylogenetic network that maximizes the number of
leaf pairs having a well-defined LCA. All of these optimization problems can naturally also be
studied for anchored triples.

Another commonly studied question in phylogenetic combinatorics is whether a phylogenetic network is
encoded by the triples it displays. For example, let $\mathscr T(N)$ denote the set of all triples
topologically-displayed by a network $N$. A class $\mathcal C$ of networks is encoded by $\mathscr
T$ if, for all $N,N'\in\mathcal C$, $ \mathscr T(N)=\mathscr T(N') $ implies $N\simeq N'$. Gambette
and Huber~\cite{Gambette2017} showed that topologically-displayed triples do not even encode level-1
networks. In the present paper, we considered a different, LCA-based notion of displaying triples.
By Lemma~\ref{lem:T1=>T2}, every triple displayed in the LCA sense is also topologically-displayed,
whereas the converse does not hold in general. Consequently, negative encoding results for
topologically-displayed triples do not automatically transfer to LCA-based displayed triples, since
$\mathscr{T}(N)=\mathscr{T}(N')$ does not automatically imply equality of the sets of displayed
triples. Nevertheless, displayed triples still do not encode level-1 networks. Indeed, consider the
network $N_2$ in Figure~\ref{fig:anchordisplay} and the star tree $T$ on $X=\{x,y,z\}$. Both are
level-1 networks, and neither displays any rooted triple in the LCA sense, but they are not
isomorphic. Hence level-1 networks are not encoded by the rooted triples they display. It remains an
interesting open question whether level-1 networks, or more generally which classes of networks, are
encoded by their \anchor-displayed anchored triples. Relatedly, it was recently shown that several
natural classes of networks are encoded by their full LCA relation $\rel_N$, including regular
level-1 networks, shortcut-free binary level-1 networks, and binary normal networks~\cite{HML:26}.
For several of these classes, even the restricted relation $\rel_N^3$, consisting only of
LCA-constraints of the form $(ab,ac)\in\rel_N$ for pairwise distinct $a,b,c\in X$, suffices to
encode them. Moreover, for a fixed leaf set $X$, the aforementioned classes
are encoded by the strict subrelation $\srel_N$ consisting of all pairs $(ab,xy)\in\rel_N$ for which
$\lca_N(ab)\prec_N\lca_N(xy)$. This relation is closely connected to the anchored-triple setting,
since an anchored triple $\underline{a}b|c$ corresponds precisely to the strict LCA-constraint $
\lca_N(ab)\prec_N\lca_N(ac)$, and hence to $(ab,ac)\in\srel_N$. However, $\srel_N$ may contain more
information than the set of anchored triples \anchor-displayed by $N$, because it also records
strict LCA-comparisons between arbitrary leaf pairs $ab$ and $xy$.

\renewcommand{\arraystretch}{1.5} 
\begin{table}[H]
    \centering
    {\footnotesize
    \begin{tabular}{p{0cm}p{1.7cm}p{1.4cm}p{5.9cm}lp{2.15cm}} \toprule 
        & Problem & Instance & Triple Set \& Relation & Phylogenetic DAG & Alg. / Thm. \\ \midrule 
        \multicolumn{5}{c}{\textit{Rooted Triple Problems}} \\ 
        \midrule
        & \PROBLEM{TC} 
        & $\RR$
        & $R_\RR \coloneqq \{(xy,xz),(xy,yz) \,:\, xy|z \in \RR\}$ \newline 
        $R' \coloneqq \{(xy,xz) \,:\, \underline{x}y|z \in \RR'\}$
        & $\GG_{R_\RR}$
        & Alg.~\ref{alg:ATC-F} \newline (input $(\RR',\emptyset)$) \\ 

        & \PROBLEM{TC-F} 
        & $(\RR,\FF)$
        & $R_\RR \coloneqq \{(xy,xz),(xy,yz) \,:\, xy|z \in \RR\}$ \newline 
          $Q_K \gets \textsc{Sat($R_{\RR},\FF$)}$ \newline 
        & $\FFRR$-extension of $\GG_{Q_K}$
        & Alg.~\ref{alg:TC-F-v2} \newline (input $(\RR,\FF)$)\\ 

        & $\mathscr W$-\PROBLEM{TC-F}
        & $(\RR,\FF)$, $\mathscr{W}$ 
        & $R\coloneqq \{(xy,xz),(xy,yz) \,:\, xy|z \in \RR\}$ \newline 
        $\textcolor{white}{R\coloneqq}\cup \{(ab,ab) \,:\, ab\in\mathscr W\}$ \newline 
        $F \coloneqq \{(xy,xy),(xz,xz),(yz,yz) \,:\, xy|z \in \FF\}$ \newline 
        $Q_K\gets$\textsc{Sat($R,\FF$)} 
        & $\FR$-extension of $\GG_{Q_K}$
        & Thm.~\ref{thm:char-W-TCF}\\ 

        & \PROBLEM{strong-TC-F}
        & $(\RR,\FF)$ 
        & $R\coloneqq \{(xy,xz),(xy,yz) \,:\, xy|z \in \RR\}$ \newline 
        $\textcolor{white}{R\coloneqq}\cup \{(ab,ab) \,:\, ab\in \tsupp_\FF\}$ \newline 
        $Q_K\gets$\textsc{Sat($R,\FF$)}  
        & $\GG_{Q_K}$
        & Cor.~\ref{cor:no-FR-extension-strong-TC} with \newline $\mathscr{W} = \tsupp_\FF$\\ 

        & \PROBLEM{LCA-TC-F} 
        & $(\RR,\FF)$ 
        & $R\coloneqq \{(xy,xz),(xy,yz) \,:\, xy|z \in \RR\}$ \newline 
        $\textcolor{white}{R\coloneqq}\cup \{(ab,ab) \,:\, ab\in \pairs(\XRF)\}$ \newline 
        $Q_K\gets$\textsc{Sat($R,\FF$)}  
        & $\GG_{Q_K}$
        & Cor.~\ref{cor:no-FR-extension-strong-TC} with \newline $\mathscr{W} = \pairs(\XRF)$ \\ 

        \midrule
        \multicolumn{5}{c}{\textit{Anchored Triple Problems}} \\ 
        \midrule
        & \PROBLEM{ATC}
        & $\RR$
        & $R \coloneqq \{(xy,xz) \, : \, \underline{x}y|z \in \RR\}$ 
        & $\GG_{R}$ 
        & Alg.~\ref{alg:ATC-F} \newline (input $(\RR,\emptyset)$)\\ 
        
        & \PROBLEM{ATC-F} 
        & $(\RR,\FF)$
        & $R \coloneqq \{(xy,xz) \, : \, \underline{x}y|z \in \RR\}$ \newline 
        $F \coloneqq \{(xy,xz) \, : \, \underline{x}y|z \in \FF\}$ 
        & $\FR$-extension of $\GG_{R,F}$ 
        & Alg.~\ref{alg:ATC-F} \newline (input $(\RR,\FF)$)\\   

        & $\mathscr W$-\PROBLEM{ATC-F} 
        & $(\RR,\FF)$, $\mathscr{W}$ 
        & $R \coloneqq \{(xy,xz) \,:\, \underline{x}y|z\in\RR\}$ \newline
        $\textcolor{white}{R \coloneqq} \cup \{(ab,ab) \,:\, ab\in\mathscr W\}$ \newline 
        $F \coloneqq \{(xy,xz) \,:\, \underline{x}y|z \in \FF\}$ 
        & $\FR$-extension of $\GG_{R,F}$
        & Thm.~\ref{thm:W-anchor-agree}\\ 

        & \PROBLEM{strong-ATC-F} 
        & $(\RR,\FF)$ 
        & $R \coloneqq \{(xy,xz) \,:\, \underline{x}y|z\in\RR\}$ \newline
        $\textcolor{white}{R \coloneqq} \cup \{(ab,ab) \,:\, ab\in \atsupp_\FF\}$
        \newline 
        $F \coloneqq \{(xy,xz) \,:\, \underline{x}y|z \in \FF\}$ 
        & $\GG_{R,F}$
        & Cor.~\ref{cor:no-FR-extension-strong-ATC}  with \newline $\mathscr{W} = \atsupp_\FF$ \\ 

        & \PROBLEM{LCA-ATC-F}
        & $(\RR,\FF)$ 
        & $R \coloneqq \{(xy,xz) \,:\, \underline{x}y|z\in\RR\}$ \newline 
        $\textcolor{white}{R \coloneqq}\cup \{(ab,ab) \,:\, ab\in \pairs(\XRF)\}$
        \newline 
        $F \coloneqq \{(xy,xz) \,:\, \underline{x}y|z \in \FF\}$ 
        & $\GG_{R,F}$
        & Cor.~\ref{cor:no-FR-extension-strong-ATC} with \newline $\mathscr{W} = \pairs(\XRF)$ \\ \bottomrule
    \end{tabular}}
    \caption{
		Summary of the consistency problems considered in this paper. The first block contains
        the rooted triple problems, where $\RR$ and $\FF$ are sets of rooted triples; the second
        block contains the anchored triple problems, where $\RR$ and $\FF$ are sets of anchored
        triples. The column ``Instance'' lists the input of the corresponding decision problem. The
        column ``Triple Set \& Relation'' specifies the auxiliary triple sets, required relations,
        forbidden relations, and saturated relations used to translate the problem into an
        LCA-constraint realization problem. The column ``Phylogenetic DAG'' gives the DAG
        constructed for yes-instances. By Lemma~\ref{lem:DAG2Network}, the corresponding
        phylogenetic network obtained from this DAG is also a valid solution. The final column gives
        the algorithm or theorem establishing the stated characterization and construction.
    }
     \label{tab:summary_results}
\end{table}

\section*{Data availability}

No datasets were generated or analysed during the current study.

\bibliographystyle{spbasic}
\bibliography{common}

\end{document}